\documentclass[letterpaper,11pt]{amsart}
\pdfoutput=1

\usepackage{xcolor}
\definecolor{darkblue}{rgb}{0.,0.,0.4}
\definecolor{darkred}{rgb}{0.5,0.,0.}
\definecolor{darkorange}{rgb}{1.0, 0.55, 0.0}
\usepackage[pdftex,colorlinks=true,linkcolor=blue,citecolor=darkred,urlcolor=darkblue]{hyperref}

\usepackage{amsmath,amsthm,amsfonts,amssymb,braket}

\usepackage{tikz}
\usetikzlibrary{arrows.meta, decorations.markings, patterns}
\usetikzlibrary{decorations.pathreplacing, bending}
\usetikzlibrary{calc,decorations.pathreplacing}
\usepackage{tikz-cd}

\usepackage{fullpage}
\usepackage{mathtools}
\usepackage[numbers]{natbib}
\usepackage[normalem]{ulem}
\usepackage{comment}

\numberwithin{equation}{section}

\newtheorem{lemma}[equation]{Lemma}
\newtheorem{theorem}[equation]{Theorem}

\newtheorem{corollary}[equation]{Corollary}

\newtheorem{proposition}[equation]{Proposition}
\newtheorem{proposition-definition}[equation]{Proposition-Definition}
\theoremstyle{definition}
\newtheorem{definition}[equation]{Definition}

\newtheorem{remark}[equation]{Remark}
\newtheorem{example}[equation]{Example}

\newcommand{\CC}{{\mathbb{C}}}
\newcommand{\RR}{{\mathbb{R}}}
\newcommand{\QQ}{{\mathbb{Q}}}
\newcommand{\ZZ}{{\mathbb{Z}}}
\newcommand{\NN}{{\mathbb{N}}}

\newcommand{\calA}{{\mathcal{A}}}
\newcommand{\calB}{{\mathcal{B}}}
\newcommand{\calC}{{\mathcal{C}}}
\newcommand{\calD}{{\mathcal{D}}}
\newcommand{\calF}{{\mathcal{F}}}
\newcommand{\calG}{{\mathcal{G}}}

\newcommand{\calI}{{\mathcal{I}}}
\newcommand{\calJ}{{\mathcal{J}}}

\newcommand{\calM}{{\mathcal{M}}}

\newcommand{\calX}{{\mathcal{X}}}
\newcommand{\calY}{{\mathcal{Y}}}
\newcommand{\calZ}{{\mathcal{Z}}}

\newcommand{\one}{{\mathbf{1}}}

\newcommand{\SWAP}{{\mathrm{SWAP}}}

\newcommand{\dd}{{\mathsf{d}}} 

\newcommand{\ii}{{\mathrm{i}}}
\newcommand{\id}{{\mathbf{id}}}

\newcommand{\FermionPlus}[2]{{\mathcal F^+}(#1,#2)}
\newcommand{\FermionPlusEven}[2]{{\mathcal F^+}(#1,#2)_0}
\newcommand{\FermionPlusOdd}[2]{{\mathcal F^+}(#1,#2)_1}
\newcommand{\FermionMinus}[1]{{\mathcal F^-}(#1)}
\newcommand{\FermionMinusEven}[1]{{\mathcal F^-}(#1)_0}
\newcommand{\FermionMinusOdd}[1]{{\mathcal F^-}(#1)_1}

\newcommand{\FermiCommutator}[2]{[{#1}, {#2}]_{\mathrm{f}}}
\newcommand{\FermiCommutant}[2]{\mathrm{Comm^f}\left({#1}, {#2}\right)}

\DeclareMathOperator{\tr}{{{tr}}}
\DeclareMathOperator{\Tr}{{Tr}}
\DeclareMathOperator*{\im}{{im}}

\DeclareMathOperator*{\Supp}{{Supp}}

\DeclareMathOperator{\rank}{{rank}}
\DeclareMathOperator{\ann}{{ann}}
\DeclareMathOperator{\spn}{{span}}

\DeclareMathOperator*{\Mat}{{\mathsf{Mat}}}
\DeclareMathOperator*{\fMat}{{\mathsf{fMat}}}

\DeclareMathOperator{\End}{{End}}
\DeclareMathOperator*{\mul}{{mul}}
\DeclareMathOperator{\dist}{{dist}}
\DeclareMathOperator{\rad}{{rad}}

\DeclareMathOperator{\Cent}{{Cent}} 
\DeclareMathOperator{\Comm}{{Comm}} 
\DeclareMathOperator{\Commf}{{Comm^f}} 
\DeclareMathOperator{\FermiCent}{{Cent^f}} 
\DeclareMathOperator{\ind}{{ind}}

\DeclarePairedDelimiter{\abs}{|}{|}

\newcommand{\dl}[1]{\textcolor{red}{DL: {#1}}}
\newcommand{\newtext}[1]{\textcolor{purple}{#1}}

\begin{document}

\title{Fermionic quantum cellular automata in 2d are trivial}
\author{Jeffrey Kwan$^*$}
\author{David M. Long$^*$}
\author{Jeongwan Haah$^{*\dagger}$}
\address{$^*$ Leinweber Institute of Theoretical Physics, Stanford University, Stanford, California, USA}
\address{$^\dagger$ Google Quantum AI, California, USA}

\begin{abstract}
    Fermionic quantum cellular automata (QCA)
    are automorphisms of local fermionic operator algebras ($\mathbb Z_2$-graded superalgebras) that have bounded spread: they map local operators to nearby operators.
    We prove that every 2-dimensional fermionic QCA on a locally finite-dimensional algebra is a composition of local automorphisms and a fermionic shift.
    This is implied by our result that
    every locally finite-dimensional
    fermionic invertible subalgebra in a one-dimensional lattice
    is Brauer trivial,
    \textit{i.e.}, it is stably bounded-spread isomorphic 
    to a tensor product fermionic algebra.
\end{abstract}

\maketitle
\tableofcontents

\section{Introduction}

Quantum dynamics is postulated by the Schr\"odinger or Heisenberg equation
in which time is a continuous parameter,
but it is ultimately the unitary, the integral of the equation,
that governs the dynamics.
Mathematically, unitary dynamics are synonymous to 
$*$-automorphisms of operator algebras, and therefore
it is natural to ask whether every such $*$-automorphism
is generated by some local Hamiltonian.
However, the set of all $*$-automorphisms of
an operator algebra with infinitely many degrees of freedom
is too large, such that this question is trivially answered no.

Schmacher and Werner~\cite{SchumacherWerner2004} 
have proposed a class of $*$-automorphisms of physical relevance,
namely quantum cellular automata (QCAs).
A quantum cellular automaton is by definition 
a $*$-automorphism of a local operator algebra
obeying an extra property that
a local operator must be mapped to a local operator nearby.
This is consistent with usual Hamiltonian time evolution
as implied by Lieb--Robinson bounds~\cite{Hastings2010}.
Interestingly, in $1+1\dd$, not every QCA
is generated by a local Hamiltonian~\cite{GNVW}.
For some time, the only known examples of such a nontrivial QCA,
not generated by any local Hamiltonian,
were the shifts~\cite{fermionGNVW2}
and free fermionic pumps~\cite{Harper2019}.

In addition to the viewpoint that QCAs are a model of abstract discrete time dynamics,
QCAs can be considered
as symmetries~\cite{Rudner2013,ElseNayak,Fidkowski2020,Ma_2026,Tu2026anomalies,Shirley2026onsiteable}.
In this context, QCAs form a representation of some symmetry,
and certain aspects of the anomaly of the symmetry
trace back to the possibility that QCAs
may not admit local generators~\cite{Tu2026anomalies}.
Indeed, Else and Nayak~\cite{ElseNayak} explained
why group cohomology is relevant for symmetry protected topological phases
by \emph{assuming} that induced symmetry action on spatial boundaries
admits local generators (quantum circuits).
The exactly solvable model in~\cite{Fidkowski2020}
shows that the failure of this assumption
is related to the so-called beyond-cohomology
symmetry protected topological phases.

The construction in~\cite{Fidkowski2020} builds upon 
a new kind of QCA in~$3+1\dd$,
which was found in~\cite{nta3}
during the attempt to disentangle the ground state of a $3+1\dd$ Walker--Wang model~\cite{BCFV}.
This ground state is believed to represent a trivial invertible topological state
in the absence of any symmetry.
However, it was shown~\cite{nta3} that
if this particular $3+1\dd$ QCA admits a unitary circuit representation
combined with any shift,
then some $2+1\dd$ unfrustrated commuting local Hamiltonian 
must realize the three-fermion anyon theory,
which must be accompanied by chiral edge modes.
The latter is believed to be impossible~\cite{Kitaev_2005,Levin2013},
and hence the QCA is believed to be nontrivial.
This is an example situation where disentangling the ground state
can be genuinely different from
disentangling the full spectrum of a commuting local Hamiltonian~\cite{Tu2026anomalies,Shirley2026onsiteable,Czajka2025anomalieslatticehomotopyquantum},
where the latter is tightly related to many-body localized phases~\cite{Long_2025}.

The long chain of beliefs towards the nontriviality (against circuits and shifts)
of the $3+1\dd$ QCA in~\cite{nta3} is fully proved in the special context of 
Clifford circuits and commuting Pauli Hamiltonians
over prime dimensional qudits~\cite{nta3,Haah2018,clifQCAclassification},
but remains conjectural in general.
In fact, the chain of beliefs definitely requires some revision
or at least an extra assumption~\cite{Haah2023invertible}.
To our knowledge, all other examples of qudit QCA~\cite{clifqca1,Shirley2022,Chen2021,SPTQCA}
enjoy a similar level of rigor towards their nontriviality
and suffer from similar caveats.
As for mathematical proofs toward the nontriviality of QCA
in higher dimensions without any symmetry,
we are aware of only one result in~\cite{nta3}
that either there exists a nontrivial \emph{fermionic} QCA in $2+1\dd$
or the three-fermion Walker--Wang disentangler QCA in $3+1\dd$ is nontrivial.

In this paper, we prove that every fermionic QCA in $2+1\dd$ is trivial 
in the sense that
it is always a composition of unitary circuit and a ``primitive'' shift.
Our proof does not follow the aforementioned line of beliefs;
instead, we follow the argument by Freedman and Hastings~\cite{FreedmanHastings2019QCA}
to show that every locally finite-dimensional%
\footnote{In~\cite{JonesThorngrenSopenko} it is reported that there exists a nontrivial QCA
in $2+1\dd$ using \emph{infinite}-dimensional local degrees of freedom.}
$1+1\dd$ fermionic boundary algebra (invertible subalgebra~\cite{Haah2023invertible})
is bounded-spread isomorphic to a tensor product fermionic algebra
and that the fermionic $2+1\dd$ QCA is always a 
circuit of local automorphisms and primitive shifts.
In a sense, this is a generalization of the classic argument for~$1+1\dd$ QCA
by Gross, Nesme, Vogt, and Werner~\cite{GNVW}.

According to~\cite{nta3} our triviality theorem for fermionic QCA in $2+1\dd$
would directly imply the nontriviality of the three-fermion Walker--Wang disentangler $3+1\dd$ QCA.
However, there are many technical details for this implication
that we must elaborate on,
so we provide the details in a separate article~\cite{paper2}.
Roughly, the hurdles that we overcome in~\cite{paper2} are as follows.
First, there is a slight difference in the scope of shift QCA.
Our triviality result decomposes a given $2+1\dd$ fermionic QCA
into a circuit of local automorphisms
and a shift that permutes ``primitive'' fermionic algebras.
In contrast, the shift QCA used in~\cite{nta3} are a smaller class.
Second, the argument in~\cite{nta3} appeals to
the distinction of emergent bosons and fermions,
but the detailed definition of the topological spins of emergent particles
was not given.
There exists a mathematical definition of emergent fermions
(or more generally topological spin of point-like topological excitations)~\cite{Ogata2022}
in such a way that it only depends on the ground state (see also \cite{Fidkowski2021});
however, to use such a definition in a proof,
one must verify conditions that the definition requires 
and perform a calculation tailored to a specific Hamiltonian,
which appears not straightforward.

The main text below is organized as follows.
We begin with elementary but thorough treatment of finite-dimensional fermionic ($\ZZ_2$-graded) algebras.
This sets the stage for an important technical ``simple-between-semisimple'' lemma (\ref{thm:FindYinbetweenXZ}).
The statement of the lemma in the ungraded case appears in~\cite{FreedmanHastings2019QCA}
and is used as the main supply of finite-dimensional simple algebras.
We then develop foundational aspects of fermionic invertible subalgebras on infinite lattices,
mirroring the treatment of~\cite{Haah2023invertible}.
The main triviality theorem for $1\dd$ fermionic invertible subalgebras
is~\ref{thm:1DBrauerGroupTrivial},
establishing that every fermionic invertible subalgebra in $1\dd$
tensored (stabilized) with a simple Majorana algebra on every site
is bounded-spread isomorphic to a tensor algebra.
Unlike the ungraded case, we use stabilization in an important way.
This section contains additional results using the torus-trick~\cite{hastings2013classifying},
which may be of independent interest.
Next, we study the tight relationship between invertible subalgebras in $\dd-1$ dimensions
and QCA in $\dd$ dimensions (\ref{thm:QCAToBrauer}).
This mirrors the argument in~\cite{Haah2023invertible} 
with a broader scope and simpler proofs.
From this algebra-QCA correspondence, we prove that every $2+1\dd$ fermionic QCA
on a locally finite-dimensional Brauer trivial algebra is trivial.
To be complete in low dimensions,
we also include the index theory and classification results for $1+1\dd$ fermionic QCA~\cite{GNVW,fermionGNVW2}.
All results up to this point
assume very little about the local algebras except that they are always finite dimensional;
however, whenever we stabilize the algebra by introducing some ancillary algebras (ancillas),
we bring only those with local dimensions uniformly bounded across the whole lattice.
In the last section, we explore consequences (mainly~\ref{thm:QCAWithAHoleIsTrivial})
from nonuniform stabilization,
under which an ancillary algebra is allowed to have arbitrary finite local algebras.
In particular, the dimension of the single-site algebra grows unbounded
as we look farther and farther away from the origin of the lattice. 
With nonuniform ancillas, every shift QCA in $2+1\dd$ or higher becomes a finite depth circuit,
and the equivalence class of a QCA on an infinite lattice is fully determined by a finite patch.

\vspace{1ex}
{\it Acknowledgments.}
JK is supported by a Stanford Q-FARM Shoucheng Zhang
fraduate fellowship and an NSF graduate research fellowship.
DML is supported by a Stanford Q-FARM Bloch postdoctoral fellowship,
a Packard Fellowship in Science and Engineering (PI: Vedika Khemani),
and the US Department of Energy, Office of Science (Award No.\ DE-SC0019380).
In this work,
we used AI (Claude and ChatGPT) to generate TikZ scripts for the figures
and for the literature search.

\subsection{Conventions}

Every algebra will be a $*$-algebra over~$\CC$ with a multiplicative identity~$\one$,
unless specified otherwise.
A $*$-homomorphism~$\phi$ is a $\CC$-linear map between $*$-algebras
such that $\phi(xy) = \phi(x)\phi(y)$ and $\phi(x^\dagger) = \phi(x)^\dagger$ for all~$x,y$.
All homomorphisms will be $*$-homomorphisms.
Always, $\phi(\one)$ is a self-adjoint projector, 
but we do \emph{not} require that \mbox{$\phi(\one) = \one$};
$\phi(\one)$ may even be zero.
If $\phi(\one) = \one$, then $\phi$ is called unital.
All important homomorphisms will be unital,
but nonunital homomorphisms will be considered for intermediate results.
In particular, the multiplicative identity of a subalgebra may or may not 
be the multiplicative identity of the mother algebra;
the inclusion may or may not be unital.
The dimension of an algebra is the dimension as a $\CC$-vector space;
for example, the dimension of the $*$-algebra~$\Mat(n)$ of all endomorphisms%
    \footnote{In~\cite{Haah2023invertible} it was ``$\dim \Mat(n) = n$.''}
of the standard inner product space~$\CC^n$ is~$n^2$.

For a subset~$R$ of a metric space~$X$ and a real number~$\ell > 0$,
we define the $\ell$-neighborhood
$R^{+\ell} = \{x \in X \mid \exists r \in R: \dist(r,x) \le \ell\}$
and $\ell$-interior
$R^{-\ell} = \{s \in X \mid \{s\}^{+\ell}\subseteq R\}$.
We regard the \(\dd\)-dimensional integer lattice \(\ZZ^\dd\)
as a metric space with the \(\ell_\infty\) metric,
\(\dist(s,t) = \max_{i} |s_i-t_i|\)
where \(s_i\), \(t_i\) are components of \(s,t \in \ZZ^\dd\);
however all of our results for QCAs on \(\ZZ^\dd\)
would also hold for other metrics such as $\ell_2$ with slightly different constants.

\section{Fermionic algebras}

Here we give basic definitions and facts.

\begin{definition}
    A $*$-algebra $\calA$ with~$\one$ over complex numbers~$\CC$
    is \emph{spectrally nondegenerate} if
    for every nonzero self-adjoint element~$h \in \calA$
    there exists $\lambda \in \RR^\times$ such that $\one - \frac 1 \lambda h$ is not invertible.
\end{definition}

The condition is automatic if we have a nonempty spectrum of~$h$,
which is always the case for any $C^*$-algebras.
We will later consider approximately finite algebras
which have a natural norm but are not norm complete from the onset.
They are not $C^*$ but will be spectrally nondegenerate.

\begin{lemma}\label{thm:spectral-nondegeneracy-on-subalgebras}
    If $\calA \hookrightarrow \calB$ is a unital injective $*$-homomorphism 
    and $\calB$ is spectrally nondegenerate,
    then $\calA$ is also spectrally nondegenerate.
\end{lemma}
\begin{proof}
    Denote the homomorphism by \(\phi\) and, in contrapositive, suppose that there is a self-adjoint and nonzero element \(h \in \calA\) such that \(a_\lambda :=\one - \tfrac{1}{\lambda} h\) is invertible for all \(\lambda \in \RR^\times\). Then \(\phi(a_\lambda) = \one - \tfrac{1}{\lambda} \phi(h)\) is also invertible for all \(\lambda \in \RR^\times\), its inverse being \(\phi(a^{-1}_\lambda)\) by unitality of \(\phi\), so that \(\calB\) is spectrally degenerate.
\end{proof}

\begin{lemma}\label{thm:PhysicalImpliesZeroJacobsonRadical}
    Every spectrally nondegenerate algebra~$\calA$ has zero Jacobson radical (semiprimitive).
\end{lemma}
\begin{proof}
    We recall a general fact that the Jacobson radical~$\rad \calA$ of a ring~$\calA$ with~$\one$
    is the intersection of all maximal left ideals,
    which happens to be equal to the intersection of all maximal right ideals~\cite{Lam2001}.
    Since the adjoint of a left ideal is a right ideal and vice versa,
    the Jacobson radical is $*$-closed.
    Observe that for any $a \in \rad \calA$,
    $\one + a$ is invertible;
    if it is not left or right invertible,
    then $\calA(\one + a)$ or $(\one +a )\calA$ is a proper left or right ideal respectively, which must be
    contained in a (proper) maximal left or right ideal~$I$ by Zorn's lemma. Then one has
    $\one = (\one + a) - a \in I + \rad \calA \subseteq I$, in contradiction to \(I\) being proper.
    
    Now, if $a \in \rad \calA$ is nonzero, then $a^\dagger \in \rad \calA$ and
    either $a+a^\dagger \in \rad \calA$ or $\ii(a-a^\dagger) \in \rad \calA$ is nonzero self-adjoint.
    Let $h \in \rad \calA$ be nonzero self-adjoint.
    Choose $\lambda \in \RR^\times$ such that $\one - \frac 1 \lambda h$ is not invertible.
    But this is a contradiction, and we must have $a = 0$.
\end{proof}

We also record a simple consequence that will be used later.

\begin{lemma}\label{thm:OddSquareRootOfIdentityIsOneDim}
    Let $\calA$ be a spectrally nondegenerate $*$-algebra.
    Let $\calC \subseteq \calA$ be a nonzero $*$-linear subspace 
    such that $\calC \cdot \calC \subseteq \CC\one$
    and $\calC \cap \CC \one = 0$.
    Then, there exists $\gamma = \gamma^\dagger \in \calC$ such that $\calC = \CC \gamma$ and $\gamma^2 = \one$.
\end{lemma}

\begin{proof}
    Let $\alpha \in \calC$ be nonzero.
    By considering $\alpha + \alpha^\dagger$ and $\ii(\alpha - \alpha^\dagger)$,
    we find a nonzero self-adjoint element $\beta = \beta^\dagger \in \calC$.
    Let $\sigma(\beta) = \{ b \in \CC \mid \beta - b \one \text{ is noninvertible.}\}$.
    By spectral nondegeneracy, there exists $r \in \RR^\times \cap \sigma(\beta)$.
    If $\beta^2 = 0$,
    then $(\one \mp \frac 1 r \beta)(\one \pm \frac 1 r \beta) = \one - \frac 1 {r^2} \beta^2 = \one$,
    which is a contradiction,  so $\beta^2 \neq 0$.
    Separately, since $\beta^2 \in \CC \one \setminus \{0\}$
    and $\CC$ is closed under taking square roots,
    we have some $c \in \CC^\times$ such that $(\frac 1 c \beta)^2 = \one$ and $0 \notin \sigma(\beta)$.
    From $(\one + \frac 1 c \beta)(\one - \frac 1 c \beta) = 0$ where neither factor is zero by assumption,
    we see that $ \pm c \in \sigma(\beta)$.
    On the other hand, if $\pm c \neq d \in \CC^\times$,
    then $(\one + \frac 1 d \beta)(\one - \frac 1 d \beta) = \one(1 -  c^2 d^{-2})$ is invertible,
    and hence $d \notin \sigma(\beta)$.
    Therefore, $c$ is real, and we set $\gamma = \frac 1 c \beta = \gamma^\dagger$.
    If $\gamma' \in \calC$ is another element,
    then $\gamma' \gamma \in \CC \one$, and hence $\gamma' \in \CC \gamma$.
\end{proof}

\begin{definition}[Fermionic algebras]
    A $*$-algebra~$\calA$ (over $\CC$) is $\ZZ_2$-graded if $\calA = \calA_0 \oplus \calA_1$
    as a vector space and $\calA_i \cdot \calA_j \subseteq \calA_{i+j\bmod 2}$
    and $\calA_i ^\dagger = \calA_i$ for all~$i,j$.
    Any element of either~$\calA_0$ or $\calA_1$ is called \emph{homogeneous};
    of course, $\calA$ contains inhomogeneous elements.
    An algebra homomorphism~$\phi : \calA \to \calB$ 
    between graded algebras is \emph{graded}
    if $\phi(\calA_j) \subseteq \calB_j$.
    A \emph{fermionic algebra} \(\calA\) 
    is a \(\ZZ_2\)-graded spectrally nondegenerate algebra~$\calA = \calA_0 \oplus \calA_1$.
    Call $\calA_0$ \emph{even} and $\calA_1$ \emph{odd}.
    A graded $*$-automorphism $\calA \ni a + b \mapsto a - b \in \calA$ 
    for all $a \in \calA_0$ and $b \in \calA_1$
    is called the \emph{grading involution} or \emph{fermion parity} of~$\calA$.
    An odd element~$\gamma$ is called \emph{Majorana}
    if $\gamma = \gamma^\dagger$ and $\gamma^2 = \one$.
\end{definition}

So, every fermionic algebra comes with two commuting involutions.

\begin{example}
    Let \(\Mat(p+q)\) be the \(*\)-algebra of \((p+q) \times (p+q)\) matrices over \(\CC\). Any \(M \in \Mat(p+q)\) can be written as a block matrix,
    \begin{equation}
        M = 
        \begin{pmatrix}
            A_{p \times p} & B_{p \times q} \\
            C_{q \times p} & D_{q \times q}
        \end{pmatrix}.
    \end{equation}
    We see that \(\Mat(p+q)\) is, as a vector space, 
    the direct sum of block-diagonal matrices $A,D$ and block-off-diagonal matrices $B,C$.
    Define a fermionic algebra \(\FermionPlus{p}{q}\) by equipping \(\Mat(p+q)\) 
    with a grading \(\FermionPlusEven{p}{q} \oplus \FermionPlusOdd{p}{q}\) 
    such that \(\FermionPlusEven{p}{q}\) consists of block-diagonal matrices ($B=0, C=0$)
    and \(\FermionPlusOdd{p}{q}\) consists of block-off-diagonal matrices ($A=0,D=0$). 
    A product of either two block-diagonal matrices or two block-off-diagonal matrices 
    is block diagonal, and the product of block-diagonal and block-off-diagonal matrices is block-off-diagonal,
    so this is indeed a \(\ZZ_2\)-grading. 
    Further, each component is invariant under the involution \(M \mapsto M^\dagger\),
    which acts as conjugate transposition.
    Being a matrix algebra over \(\CC\), 
    \(\FermionPlus{p}{q}\) is spectrally nondegenerate.
\end{example}

\begin{example}
    The algebra \(\FermionMinus{1} = \FermionMinusEven{1} \oplus \FermionMinusOdd{1}
    = \CC \one \oplus \CC \gamma\),
    where \(\gamma^2 = \one\) and \(\gamma^\dagger = \gamma\),
    is a commutative two-dimensional fermionic algebra. 
    This algebra corresponds to a single Majorana fermion.
    It is isomorphic to the group algebra \(\CC[\ZZ_2]\), 
    where the span of the basis group elements are graded by said group element.
    Without the grading, this is isomorphic to \(\CC \oplus \CC\), which is spectrally nondegenerate.
\end{example}

\begin{definition}
    The \emph{fermionic tensor product} of \(\calA\) and \(\calB\) is \(\calA \otimes \calB\) as a vector space,
    equipped with the grading \(\calA \otimes \calB = (\calA_0 \otimes \calB_0 \oplus \calA_1 \otimes \calB_1) \oplus (\calA_0 \otimes \calB_1 \oplus \calA_1 \otimes \calB_0)\). 
    The multiplication of homogeneous elements is given by
    \begin{equation}
        (a_i \otimes b_j) (a'_{i'} \otimes b'_{j'}) = (-1)^{i' j}  a_i a'_{i'} \otimes b_j b'_{j'} \,\, ,
    \end{equation}
    where \(i,i',j,j' \in \{0,1\}\) index the grading.
    Extending by linearity defines multiplication for all other elements.
    The $*$-structure $\dagger$ acts on homogeneous elements
    by $(a_i\otimes b_j)^\dagger =
    (-1)^{ij}(a_i^\dagger\otimes b_j^\dagger)$.
\end{definition}

\begin{example}\label{ex:TableOfTensorProduct}
    Defining \(\FermionMinus{p} = \FermionPlus{p}{0} \otimes \FermionMinus{1} \), we have
    \begin{align}
        \FermionPlus{p}{q} \otimes \FermionPlus{p'}{q'} &\cong \FermionPlus{p p' + q q'}{pq' + qp'}, \nonumber\\
        \FermionMinus{p} \otimes \FermionMinus{q} &\cong \FermionPlus{pq}{pq}, \\
        \FermionPlus{p}{q} \otimes \FermionMinus{p'} &\cong \FermionPlus{pp'}{q p'} \otimes \FermionMinus{1}, \nonumber
    \end{align}
    where all isomorphisms are graded \(*\)-isomorphisms.
\end{example}

\begin{definition}
    A \emph{graded (\(*\)-)subalgebra} \(\calB\) of a fermionic algebra \(\calA\) 
    is a (\(*\)-)subalgebra of the form 
    \(\calB = \calB_0 \oplus \calB_1 = (\calA_0 \cap \calB) \oplus (\calA_1 \cap \calB)\). 
    In particular, if \(b_0 + b_1 \in \calB\), where \(b_i \in \calA_i\) are homogeneous, 
    then \(b_0 \in \calB\) and \(b_1 \in \calB\).
    A graded linear subspace \(\calJ = (\calJ \cap \calA_0) \oplus (\calJ \cap \calA_1) \),
    possibly without the multiplicative identity of~\(\calA\), 
    is a \emph{graded left ideal} if \(\calA \calJ \subseteq \calJ\). 
    Similarly, define graded right and two-sided ideals.
    When unspecified, ideals are two-sided.
    A fermionic subalgebra is a graded $*$-subalgebra of a fermionic algebra.
\end{definition}

\begin{example}
    The center of \(\calA\), 
    defined to be the set of elements which commute with every element of \(\calA\) 
    and written 
    \begin{equation}
        \Cent(\calA) = \{z \in \calA \mid [a,z] = az - za =0 \text{ for all } a\in \calA\},
    \end{equation}
    is a graded \(*\)-subalgebra.
    Indeed, for \(z \in \calA\) to be in the center, 
    it is sufficient that \(z\) commutes with all homogenous elements \(a_g \in \calA\).
    Decompose \(z = z_0 + z_1\) in terms of homogeneous elements.
    Then \([z_0 + z_1, a_g] = [z_0, a_g] + [z_1,a_g] = 0\), 
    but \([z_0,a_g]\) and \([z_1, a_g]\) belong to different components, 
    so they are both zero.
    We conclude that \(z_0, z_1 \in \Cent(\calA)\), so that \(\Cent(\calA)\) is correctly graded.
    That \(\Cent(\calA)\) is compatible with the $*$-involution 
    follows from the same calculations as in the ungraded case.
    The center is typically not an ideal.
\end{example}

\begin{definition}
    A fermionic algebra~\(\calA\) is \emph{simple} if its only graded two-sided ideals are zero and~\(\calA\).
    A fermionic algebra~\(\calA\) is \emph{central} if the even component of the center 
    consists only of scalar multiples of the unit: \(\Cent(\calA)_0 = \CC \one\).
\end{definition}

\begin{definition}
    The \emph{fermionic commutator} \(\FermiCommutator{\cdot}{\cdot}: \calA \times \calA \to \calA\) 
    is the $\CC$-bilinear map on~\(\calA\) defined on homogeneous elements by
    \begin{equation}
        \FermiCommutator{a_0}{a'_0} = a_0 a'_0 - a'_0 a_0, \quad
        \FermiCommutator{a_0}{a_1} = a_0 a_1 - a_1 a_0 = -\FermiCommutator{a_1}{a_0}, \quad
        \FermiCommutator{a_1}{a'_1} = a_1 a'_1 + a'_1 a_1,
    \end{equation}
    where \(a_0, a_0' \in \calA_0\) and \(a_1, a_1' \in \calA_1\).
    For a fermionic subalgebra \(\calB \leq \calA\), the \emph{fermionic commutant}%
    \footnote{The fermionic commutator and fermionic commutant 
    are frequently referred to as supercommutator and supercommutant, respectively.
    We choose to use ``fermionic'' instead of prefix ``super-'' 
    just to avoid confusion with super as in superset.}
    is 
    \begin{equation}
        \FermiCommutant{\calB}{\calA} = \{c \in \calA \mid \FermiCommutator{c}{b} = 0 ~~\forall b \in \calB\}.
    \end{equation}
\end{definition}

\begin{lemma}
    If \(\calB \leq \calA\) is a fermionic subalgebra, 
    then \(\calC = \FermiCommutant{\calB}{\calA}\) is a unital fermionic subalgebra of~$\calA$.
\end{lemma}
\begin{proof}
    The unit \(\one \in \calA\) is even, 
    and commutes with all of \(\calA\), so \(\one \in \calC\), and \(\calC \hookrightarrow \calA\) is unital.
    We also have that \(\FermiCommutator{c}{b}^\dagger = \pm \FermiCommutator{c^\dagger}{b^\dagger}\) 
    for any homogeneous elements \(c,b \in \calA\). 
    It follows that \(c \in \FermiCommutant{\calB}{\calA}\) 
    implies \(c^\dagger \in \FermiCommutant{\calB}{\calA}\).
    Bilinearity of the fermionic commutator implies that \(\calC\) is a vector subspace.

    It remains to check closure under multiplication. 
    It suffices to consider homogeneous elements. 
    We have
    \begin{equation}
        \FermiCommutator{b_i}{c_j c'_k} = \FermiCommutator{b_i}{c_j} c'_k + (-1)^{ij} c_j \FermiCommutator{b_i}{c'_k},\label{eq:graded_derivation}
    \end{equation}
    which vanishes,
    so that \(\calC\) is a graded \(*\)-subalgebra. 
    Spectral nondegeneracy follows by~\ref{thm:spectral-nondegeneracy-on-subalgebras}.
\end{proof}

\begin{definition}
    For a fermionic algebra~$\calA$,
    we define the \emph{fermionic center} of~$\calA$
    analogously to the center, replacing the commutator
    with the fermionic commutator:
    \begin{equation}
        \FermiCent(\calA) = \FermiCommutant{\calA}{\calA} = \{z\in \calA \mid \FermiCommutator{a}{z} = 0 \text{ for all } a\in \calA\}.
    \end{equation}
\end{definition}

\begin{lemma}\label{thm:FermiCenterCentrality}
    A fermionic algebra $\calA$ is central if and only if $\FermiCent(\calA) = \CC\one$.
\end{lemma}

\begin{proof}
    (``if'') This is trivial since $\FermiCent(\calA)_0 = \Cent(\calA)_0$.

    (``only if'')
    By definition, $\FermiCent(\calA)_0 = \Cent(\calA)_0 = \CC \one$.
    If $\FermiCent(\calA)_1 \neq 0$,
    then 
    \begin{equation}
        \FermiCent(\calA)_1\cdot \FermiCent(\calA)_1\subseteq \FermiCent(\calA)_0 = \CC \one,
    \end{equation}
    so we have a nonzero element~$\gamma \in \FermiCent(\calA)_1 = \CC \gamma$ with $\gamma^2 = \one$ by~\ref{thm:OddSquareRootOfIdentityIsOneDim}.
    But, $\gamma \in \calA_1$, so $\FermiCommutator{\gamma}{\gamma} = 0$ by definition of~$\FermiCent(\calA)$,
    which is impossible.
\end{proof}

\begin{lemma}\label{thm:CentralTensorCentralIsCentral}
    If fermionic algebras $\calA$ and $\calB$ are both central,
    so is their fermionic tensor product.
\end{lemma}

\begin{proof}
    If $a_0 \otimes b_0 + a_1 \otimes b_1 \in \FermiCent(\calA \otimes \calB)$,
    then for any homogeneous~$a'_i \otimes \one \in \calA \otimes \one$
    \begin{equation}
        \FermiCommutator{a_0 \otimes b_0 + a_1 \otimes b_1}{a'_i \otimes \one} = 
        \FermiCommutator{a_0}{a'_i} \otimes b_0 + (-1)^i \FermiCommutator{a_1}{a'_i} \otimes b_1 = 0,
    \end{equation}
    which implies that $a_0, a_1 \in \FermiCent(\calA) = \CC\one$.
    A symmetric calculation shows $b_0,b_1 \in \FermiCent(\calB)$.
\end{proof}

\begin{remark}
    A unitary in a fermionic algebra does not have to be homogeneous;
    for example, $e^{\ii \theta \gamma + \ii \theta'} = \one e^{\ii \theta'} \cos \theta + e^{\ii \theta'} \ii \gamma \sin \theta \in \FermionMinus{1} = \CC\one \oplus \CC \gamma$
    is a unitary for any~$\theta, \theta' \in \RR$,
    where $\gamma$ is odd self-adjoint and squares to~$\one$.
    In fact, every unitary of~$\FermionMinus{1}$ has that form.
\end{remark}

\begin{lemma}\label{thm:AutomorphismByUnitaryConjugationCanAlwaysBeHomogeneous}
    Let $\alpha: \calF \ni x \mapsto u x u^\dagger \in \calF$ be a graded 
    automorphism by some unitary~$u$ of a fermionic central algebra~$\calF$.
    Then, there exists a homogeneous unitary~$w \in \calF$ 
    such that $\alpha(x) = w x w^\dagger$ for all~$x \in \calF$.
\end{lemma}

\begin{proof}
    Let $\phi$ denote the grading involution.
    Observe that $\phi(u\phi(x) u^\dagger) = u x u^\dagger$ for any homogeneous~$x \in \calF$.
    Rearranging, we have $u^\dagger \phi(u) x = x u^\dagger \phi(u)$,
    implying that 
    $v = u^\dagger \phi(u)$ is in the center, 
    which is $\CC\one$ or $\CC\one + \CC \gamma$.
    (To see this,
    apply~\ref{thm:OddSquareRootOfIdentityIsOneDim} to the odd part of the center, not the fermionic center).
    From $v = u^\dagger \phi(u) = \phi(u) u^\dagger$,
    we see that $\phi(v) = \phi(u)^\dagger u = v^\dagger$.
    Since $v^\dagger v = \phi(u)^\dagger u u^\dagger \phi(u) = \one$,
    we have $v = e^{\ii (\theta \gamma + \theta')}$ for some~$\theta,\theta' \in \RR / 2\pi \ZZ$.
    The condition~$\phi(v) = v^\dagger$ implies that $e^{\ii (-\theta \gamma + \theta')} = e^{\ii(-\theta \gamma - \theta')}$
    so $e^{\ii \theta'} = \pm 1$.
    Let $s = e^{\ii(\theta \gamma /2) + \ii \theta'}$.
    Since $s \in \Cent(\calF)$, the conjugation by~$w=us$ is equal to~$\alpha$,
    and $\phi(u s) = uv \phi(s) = u e^{\ii \theta \gamma /2} e^{2 \ii \theta'} = \pm u s$.
\end{proof}

\begin{lemma}
    \label{lem:TensorProdUniversalProperty}
    The fermionic tensor product \(\calA \otimes \calB\) obeys the universal property that, 
    given any graded homomorphisms \(\phi: \calA \to \calC\) and \(\psi: \calB \to \calC\) 
    such that \(\FermiCommutator{\phi(\calA)}{\psi(\calB)} = 0\), 
    there is a graded homomorphism \(\chi: \calA \otimes \calB \to \calC\) 
    such that \(\phi = \chi \circ \iota_\calA\) and \(\psi = \chi \circ \iota_\calB\), 
    where \(\iota_\calA : \calA \ni a \mapsto a \otimes \one_\calB \in \calA \otimes \calB\) 
    and \(\iota_\calB: \calB \ni b \mapsto \one_\calA \otimes b \in \calA \otimes \calB\).
\end{lemma}

\begin{proof}
    Define \(\chi\) on the basis of simple tensors in \(\calA \otimes \calB\) by
    \begin{equation}
        \chi(a \otimes b) = \phi(a) \psi(b),
    \end{equation}
    and extend by linearity. 
    We have that \(\phi = \chi \circ \iota_\calA\) and \(\psi = \chi \circ \iota_\calB\), 
    and must check that \(\chi\) is a graded homomorphism.

    The map \(\chi\) is graded because \(\calA \otimes \calB\), \(\phi\), and \(\psi\) are graded.
    To be a homomorphism, it suffices to check \(\chi([a\otimes b][a'\otimes b']) = \chi(a\otimes b) \chi(a'\otimes b')\) for homogeneous elements \(a,a',b,b'\). 
    Let the grading of \(a'\) be \(i\) and that of \(b\) be \(j\). Then
    \begin{align}
        \chi([a\otimes b][a'\otimes b']) &= (-1)^{ij}\chi(a a'\otimes b b') \nonumber\\
        &= (-1)^{ij}\phi(a)\phi(a')\psi(b)\psi(b') \nonumber\\
        &= (-1)^{ij}\phi(a)\left( (-1)^{ij} \psi(b) \phi(a') +  \FermiCommutator{\phi(a')}{\psi(b)} \right) \psi(b') \\
        &= \chi(a\otimes b)\chi(a'\otimes b').\nonumber \qedhere
    \end{align}
\end{proof}

\section{Finite-dimensional fermionic algebras}

Since every fermionic algebra is semiprimitive (\ref{thm:PhysicalImpliesZeroJacobsonRadical}),
if we assume that a fermionic algebra~$\calA$ is finite dimensional,
then a right module~$\calA$ over~$\calA$ is a direct sum of right simple $\calA$-modules,
{\it i.e.}, $\calA$ is semisimple.
Since the algebra~$\calA$ with grading forgotten is a semisimple algebra over~$\CC$,
we will routinely use the Wedderburn--Artin theorem
and refine the structure using grading.

\begin{lemma}[Wedderburn--Artin]
    Every finite-dimensional semisimple algebra over~$\CC$ is isomorphic to a finite direct sum of matrix algebras.
\end{lemma}

\begin{corollary}
    Let $\calA$ be a finite-dimensional spectrally nondegenerate $*$-algebra.
    Then, $\calA$ is central if and only if it is simple.
    In addition, any subalgebra of~$\calA$ with a multiplicative identity,
    even if the inclusion homomorphism is not unital,
    is always spectrally nondegenerate.
\end{corollary}

\begin{proposition-definition}\label{thm:FiniteDimensionalCentralIsSimple}
    Let $\calA$ be a finite-dimensional fermionic algebra.
    Then, $\calA$ is central if and only if $\calA$ is simple.
    If $\calA$ is central, then either $\Cent(\calA) = \CC \one$, 
    in which case $\calA$ is called \emph{even},
    or $\Cent(\calA) = \CC \one \oplus \CC \gamma$ for some odd element $\gamma$ such that $\gamma^\dagger = \gamma$ and $\gamma^2 = \one$, 
    in which case $\calA$ is called \emph{odd}.
    The odd part of an odd fermionic algebra~$\calA_0 \oplus \calA_1$
    is~$\calA_1 = \gamma \calA_0$,
    and $\calA_0$ as an ungraded algebra is central simple.
\end{proposition-definition}

The ``only if'' implication was noted in Lemma~A.3 of~\cite{nta3}.

\begin{proof}
    (``if'')
    Let $\pi \in \calA$ be a nonzero even central element.
    Then, $\pi \calA = \calA \pi$ is a nonzero two-sided graded ideal of~$\calA$,
    so it must be the whole: $\calA \pi = \calA$.
    Invoking the Wedderburn--Artin theorem for ungraded~$\calA$,
    we see that this can happen only when $\pi$ is a scalar multiple of~$\one$.

    (``only if'')
    Let $\calJ$ be a nonzero two-sided graded ideal of~$\calA$.
    Grading forgotten, it is a nonzero two-sided ideal of ungraded~$\calA$,
    which is, by the Wedderburn--Artin theorem, a direct sum of ungraded simple subalgebras,
    which has its own multiplicative identity~$\pi \in \calJ$, that lies in the center of~$\calA$.
    This unit~$\pi$ must preserve the grading of~$\calJ$, implying that $\pi$ is even.
    Since $\calA$ is assumed graded central, $\pi$ can only be a scalar multiple of~$\one \in \calA$,
    and $\calJ = \calA$.

    (Center)
    Suppose $\Cent(\calA)$ has a nonzero odd component.
    Then, the odd component satisfies 
    the assumptions of~\ref{thm:OddSquareRootOfIdentityIsOneDim}
    and we find~$\gamma$.
    Any odd element times $\gamma$ is even, and therefore $\calA_1 = \gamma \calA_0$.
    The claim that $\calA_0$ is central is now clear.
\end{proof}

\begin{theorem}[Wall~\cite{Wall1964gradedbrauer}]
    \label{thm:tensor_prod_simple}
    If $\calA$ and $\calB$ are finite-dimensional central fermionic algebras,
    then the fermionic tensor product $\calA\otimes \calB$ is central simple.
\end{theorem}

\begin{proof}
    By~\ref{thm:CentralTensorCentralIsCentral}, the fermionic algebra $\calA \otimes \calB$ is central.
    The claim follows from~\ref{thm:FiniteDimensionalCentralIsSimple}.
\end{proof}

\begin{proposition-definition}\label{thm:Wedderburn-Fermionic}
    Every finite-dimensional fermionic algebra~$\calA$ is a direct sum of simple fermionic algebras.
    The fermionic center is $\CC$-span of the multiplicative identities~$\pi_\mu$ of the direct summands,
    called \emph{canonical projectors} of~$\calA$.
    \begin{equation}
        \calA \cong \bigoplus_\mu (\pi_\mu \calA)
    \end{equation}
\end{proposition-definition}
\begin{proof}
    If \(\calA\) is simple, we are done.
    Otherwise, let \(\calJ\) be a proper fermionic ideal of \(A\). 
    It is also an ungraded ideal, 
    so the Wedderburn--Artin theorem implies that \(\calA = \calJ \oplus \ann \calJ\), 
    where \(\ann \calJ = \{ a \in \calA \,\mid a \calJ = \calJ a = 0\}\) is the annihilator. 
    Indeed, it is necessary and sufficient that \(a \pi = 0\) 
    where \(\pi\) is the multiplicative identity in \(\calJ\), which is even.
    That there is $\pi \in \calJ$ follows from the Wedderburn--Artin theorem.
    Decomposing $a \in \ann \calJ$ into even and odd parts, \(a = a_0 + a_1\), 
    we have that \(a_0 \pi + a_1 \pi = 0\). 
    We conclude that they are each in the annihilator.
    It follows that \(\ann \calJ\) is fermionic.
    By induction in the dimension of~$\calA$,
    we see that $\calA$ is the direct sum of minimal ideals, 
    each of which is a fermionic subalgebra.
    
    The multiplicative identities~$\pi_\mu$ are certainly in the fermionic center
    and sum to~$\one_\calA$.
    Since any smaller projector in the fermionic center would result in a smaller fermionic subalgebra,
    we have exhausted all~$\FermiCent(\calA)$.
\end{proof}

\subsection{Faithful representations}

\begin{lemma}[Skolem--Noether]\label{thm:AutoIsInnerInMat}
    Every automorphism of a $*$-algebra~$\Mat(n)$ is a conjugation by some unitary~$u \in \Mat(n)$.
\end{lemma}
\begin{proof}
    By Skolem--Noether theorem, we find an invertible element~$y \in \Mat(n)$ such that
    the automorphism is $x \mapsto y x y^{-1}$.
    Since we have a $*$-automorphism,
    $y x y^{-1} = ( y x^\dagger y^{-1} )^\dagger = y^{-\dagger} x y^\dagger$,
    implying that
    $y^\dagger y x = x y^\dagger y$ for any~$x$.
    Let $y = u p$ be the polar decomposition where $p \succeq 0$.
    It follows that $p^2 = y^\dagger y$ is in the center, which must be a nonzero scalar multiple of~$\one$.
    Therefore, the unique positive square root~$p$ of~$p^2$ 
    is a nonzero scalar multiple of~$\one$,
    and $y x y^{-1} = u x u^\dagger$.
\end{proof}

If~$\calA$ is even central,
then $\calA$ as an ungraded algebra is central simple,
and hence is isomorphic to a matrix algebra~$\Mat(n)$.
By~\ref{thm:AutoIsInnerInMat}, 
there exists a unitary~$\Gamma \in \Mat(n)$
that implements the grading automorphism by conjugation.
Being an involution, $\Gamma$ must square to a scalar multiple of~$\one$,
which we can choose to be~$1$ since a square root always exists in~$\CC$.
Working in the basis where $\Gamma$ is diagonal,
we see that the isomorphism $\phi : \calA \to \Mat(n)$ can be chosen such that
the grading automorphism is
\begin{align}
    \begin{pmatrix}
        A & B \\ C & D    
    \end{pmatrix}
    \mapsto
    \begin{pmatrix}
        A & -B \\ -C & D
    \end{pmatrix} \, . \label{eq:evenrep}
\end{align}
The dimensions of~$A$ and~$D$ may be different.
This representation of~$\calA$ is unique up to a block-diagonal unitary.
Any other finite-dimensional representation of~$\calA$
(a unital embedding of~$\calA$ into some matrix algebra)
is a finite direct sum of copies of this,
which we may denote as $\Mat(n) \otimes \one_k$ where $k \ge 1$ is an integer.

If $\calA$ is odd central,
then the even part $\calA_0$ is central simple and hence is isomorphic to a matrix algebra~$\Mat(n)$.
Let the center of~$\calA$ be generated by~$\gamma^\dagger = \gamma \in \calA_1$ by~\ref{thm:OddSquareRootOfIdentityIsOneDim}.
The grading involution sends~$\gamma$ to~$-\gamma$ and fixes all elements of~$\calA_0$.
If $\rho : \calA \to \Mat(m)$ is faithful with $\rho(\one) = \one$,
the image~$\rho(\gamma)$ must be a self-adjoint unitary of order~$2$, 
which we may assume to be diagonal.
The representation space~$\CC^m$ then contains 
two mutually orthogonal $\calA$-invariant proper subspaces $\frac 1 2(1 \pm \rho(\gamma)) \CC^m$.
In each subrepresentation~$\rho_\pm$ of~$\calA$,
$\rho_\pm(\gamma)$ is $\pm \one$,
so the image $\rho_\pm(\calA)$ is exactly $\rho_\pm(\calA_0) \cong \Mat(n)$.
It follows that $\rho_\pm(\calA) = \one_{k_\pm} \otimes \Mat(n)$ for some integer~$k_\pm \ge 1$.
Note that $k_+$ does not have to be equal to~$k_-$.
In terms of concrete matrices, the representation is
\begin{align}
    \rho : \calA \ni a + \gamma b \mapsto 
    \begin{pmatrix}
         (\phi(a) + \phi(b)) \otimes \one_{k_+} & 0 \\ 
         0 & (\phi(a) - \phi(b)) \otimes \one_{k_-}
    \end{pmatrix} \in \Mat(n(k_+ + k_-)) \label{eq:generalOddRep}
\end{align}
where $a,b \in \calA_0$ and $\phi : \calA_0 \to \Mat(n)$ is an isomorphism.

Note that the grading involution can be implemented 
by a unitary conjugation in~$\Mat(n(k_+ + k_-))$
if and only if $k_+ = k_-$.
Indeed, if $k_+ = k_- = k$,
then
\begin{align}
    \begin{pmatrix}
        (\phi(a) + \phi(b)) \otimes \one_k & 0 \\ 
        0 & (\phi(a) - \phi(b)) \otimes \one_k
    \end{pmatrix}
    =
    (\one_2 \otimes \phi(a) + \sigma^z \otimes \phi(b)) \otimes \one_k
\end{align}
and the conjugation by~$\sigma^x \otimes \one_n \otimes \one_k$ implements the grading involution.
If $k_+ \neq k_-$, then $\Tr(\rho(\gamma)) \neq 0$, but any inner automorphism
such that $\rho(\gamma) \mapsto -\rho(\gamma)$ forces $\Tr(\rho(\gamma)) = 0$,
so the grading cannot be inner.
If $k_+ = k_- = k$, the representation~\eqref{eq:generalOddRep} may be taken as
\begin{align}
    \rho : \calA \ni a + \gamma b \mapsto 
    \begin{pmatrix}
        \phi(a) & \phi(b) \\ 
        \phi(b)  & \phi(a)
    \end{pmatrix} \otimes \one_k \in \Mat(2nk) \, . \label{eq:oddrep}
\end{align}

In summary, we have
\begin{proposition}\label{thm:faithfulreps}
    Let $\calA$ be a finite-dimensional central fermionic algebra.
    If $\calA$ is even,
    then any nonzero unital representation is faithful and is of form~\eqref{eq:evenrep}.
    If $\calA$ is odd and is unitally faithfully represented on a space 
    where the grading involution is conjugation by some unitary,
    then the representation is of form~\eqref{eq:oddrep},
    which is always the case for a minimal faithful representation.
\end{proposition}

\begin{definition}\label{def:rank}
    The \emph{envelope rank} of a finite-dimensional central fermionic algebra~$\calA$
    is the minimum positive integer~$m$ such that there exists an injective
    ungraded homomorphism from~$\calA$ into~$\Mat(m)$.
\end{definition}

\noindent
The envelope rank of an odd central fermionic algebra is always even 
and fully determines the graded isomorphism class of the algebra.
The envelope rank of an even central fermionic algebra can be any positive integer
and does not determine the graded isomorphism class of the algebra.

The following is now obvious:

\begin{theorem}[Wall~\cite{Wall1964gradedbrauer}]\label{thm:fermionicCSAs}
    Any finite-dimensional central fermionic algebra 
    is graded-isomorphic to either \(\FermionPlus{p}{q}\) for some integer $p,q \ge 0$ (even)
    or \(\FermionMinus{s}\) for some integer $s \ge 1$ (odd).
\end{theorem}

\subsection{Rank tuples}

\begin{definition}
    Define an \emph{unordered} tuple, called the \emph{rank tuple}, 
    for each isomorphism class of finite-dimensional central fermionic algebras
    by
    \begin{equation}
        \begin{cases}
            \FermionPlus{a}{b} &\mapsto (a,b) \\
            \FermionMinus{c} &\mapsto (c / \sqrt{2}, c / \sqrt{2}) 
        \end{cases} \, .
    \end{equation}
    Define a binary operation on rank tuples by
    \begin{equation}
        (x,y) \cdot (x',y') \mapsto (x x' + y y' , x y' + x' y ) \, . \label{eq:RankTupleOperation}
    \end{equation}
\end{definition}

The binary operation can be thought of as 
the multiplication of matrices of form $\begin{pmatrix} x & y \\ y & x \end{pmatrix}$.

\begin{proposition}\label{thm:RankTupleMonoid}
    With the binary operation~\eqref{eq:RankTupleOperation},
    the set of all rank tuples is a monoid with $(1,0)$ the identity,
    which is isomorphic to the monoid of all isomorphism classes of 
    finite-dimensional central fermionic algebras under graded tensor product operation
    with $\FermionPlus{1}{0}$ the identity.
\end{proposition}

\begin{proof}
    Since $(x,y) \cdot (y',x') = (xy' + y x', x x' + y y' )$,
    the operation is well defined for unordered tuples.
    The assertion is easily checked by counting dimensions of even and odd part;
    see~\ref{ex:TableOfTensorProduct}.
\end{proof}

\begin{remark}\label{rem:EvenOddUnderTensorProduct}
    The rank tuple of a finite-dimensional central fermionic algebra
    is an unordered tuple of integers if it is even or integers divided by~$\sqrt 2$ if it is odd.
    Hence, the parity of finite-dimensional central ferminonic algebra
    gives a monoid homomorphism into the additive group~$\ZZ_2$.
\end{remark}

As shown in~\ref{ex:TableOfTensorProduct} above,
we have $(a,b)\cdot(1,1) / \sqrt{2} = (a+b,a+b)/\sqrt{2} = (a+b, 0) \cdot (1,1)/\sqrt{2}$.
This means that we may have $\calA \otimes \calB \cong \calA \otimes \calC$ 
even if $\calB$ and $\calC$ are not isomorphic.
However, this failure of cancellation property is remedied if the isomorphism 
$\calA \otimes \calB \cong \calA \otimes \calC$ 
restricts to the identity map on~$\calA$:
\begin{lemma}\label{thm:CancellationPropertyForIsomorphisms}
    Let $\calA,\calB,\calC$ be fermionic algebras, possibly infinite dimensional,
    where $\calA$ is central.
    If $\phi : \calA \otimes \calB \to \calA \otimes \calC$ is an isomorphism of fermionic algebras 
    such that
    \begin{equation}
        \begin{tikzcd}
            \calA \otimes \calB \arrow[rr, "\phi"] & &  \calA \otimes \calC \\
             &  \calA \arrow[ur, hook]  \arrow[ul, hook]&
        \end{tikzcd}
    \end{equation}
    commutes where the slanted maps are $a \mapsto a \otimes \one$,
    then $\phi|_{\one \otimes \calB}: \one \otimes \calB \to \one \otimes \calC$ is an isomorphism.
\end{lemma}
\begin{proof}
    First, $\phi|_{\one \otimes \calB}$ is injective because $\phi$ is injective, 
    and lands in $\one \otimes \calC$ because $\calA$ is central.
    Since $\phi$ is invertible,
    we also know that $\phi^{-1}(\one \otimes \calC)$ lands in $\one \otimes \calB$.
    This shows that $\phi|_{\one \otimes \calB}$ is surjective.
\end{proof}

\begin{definition}
    An \emph{irreducible} rank tuple is one that is not divisible by any other rank tuple.
    The following irreducible rank tuples are called \emph{primitive}:
    \begin{equation}
        \begin{cases}
            (p,0) \text{ with an odd prime }p\\
            (1,1)/\sqrt 2
        \end{cases}
    \end{equation}
    Accordingly, a \emph{primitive fermionic algebra} is 
    a finite-dimensional central fermionic algebra with a primitive rank tuple.
\end{definition}

\begin{remark}
    One may also define that
    a rank tuple~$p$ is \emph{prime} if $p | a b$ implies $p|a$ or $p|b$.
    An irreducible rank tuple may not be prime;
    for example, $(3,0)$ is irreducible and $(3,0) ~|~ (\frac{1}{\sqrt 2},\frac{1}{\sqrt 2}) \cdot (1,2)$
    but does not divide any of the factors.
    An example prime rank tuple is $p = (1,1)/\sqrt 2$.
    (\emph{Proof}: 
    Suppose $p | ab$. 
    If $a$ or $b$ is odd, then $p$ divides the odd.
    If both $a=(x,y)$ and $b = (x',y')$ are even, 
    then the product is $(xx'+yy', xy'+x'y)$ that must be a product of~$p$ and some other odd tuple.
    The latter has two equal entries,
    so $xx'+yy' = xy'+x'y$, which rearranges to $x(x'-y')=y(x'-y')$,
    which implies that either~$a$ or~$b$ is divisible by~$p$.)
\end{remark}

\begin{lemma}\label{thm:StabilizedThenDecomposed}
    For any finite-dimensional central fermionic algebra~$\calA$,
    the tensor product $\calA \otimes \FermionMinus{1}$ 
    is a unique tensor product of primitive fermionic algebras
    determined by $\dim_\CC \calA$,
    up to automorphisms.
    If $\calA$ is a tensor product of primitive algebras,
    the constituent primitive algebras are unique up to automorphisms.
\end{lemma}
\begin{proof}
    The tensor product $\calA \otimes \FermionMinus{1}$
    is isomorphic to either $\FermionMinus{n}$ or $\FermionPlus{n}{n}$
    for some integer~$n > 0$.
    Let $n = 2^{m_2} 3^{m_3} 5^{m_5} \cdots$ be the prime factorization.
    We see that $(n,n)/\sqrt{2} = ((1,1)/\sqrt{2})^{2m_2 + 1} (3,0)^{m_3} (5,0)^{m_5} \cdots$
    and $(n,n) = ((1,1)/\sqrt{2})^{2 m_2 + 2} (3,0)^{m_3} (5,0)^{m_5} \cdots$.
    The exponents here are determined by the $\CC$-dimension of~$\calA$.
\end{proof}

\subsection{Subalgebras}

Let $\calX$ and $\calZ$ be finite-dimensional central fermionic algebras
with a graded injective homomorphism from~$\calX$ into~$\calZ$,
possibly nonunital.
We identify $\calZ$ with the image of its minimal faithful representation,
and aim to understand the image of~$\calX$ in~$\calZ$.
Since $\calZ$ is minimally represented, we have a self-adjoint unitary~$\Gamma_\calZ$
that squares to~$\one$ such that the grading involution of~$\calZ$ is represented 
by conjugation by~$\Gamma_\calZ$.
The unitary~$\Gamma_\calZ$ belongs to~$\calZ$ if and only if $\calZ$ is even.
By~\ref{thm:faithfulreps}, we know that a minimal representation is always of form
\begin{align}
    \left(\begin{array}{c|c}
        A & B \\ \hline C & D
    \end{array}\right)
\end{align}
where the off-diagonal blocks are odd and the diagonal blocks are even.
If the algebra is even, then $A,B,C,D$ are independently arbitrary.
If the algebra is odd, then $A = D$ and $B = C$ but otherwise arbitrary.
The latter case of an odd algebra can be understood by considering the commutant (not fermionic)
of an odd central element that is represented as~$A = D = 0$ and $B = C = \one_n$ for some~$n$.

The following lemma will simplify our investigation.
\begin{lemma}\label{lem:EvenOddXZFactorThrough}
    Let $\calX$ and $\calZ$ be finite-dimensional central fermionic algebras
    with a graded injection~$\calX \hookrightarrow \calZ$, not necessarily unital.
    Suppose that $\calX$ is even and $\calZ$ is odd, or that $\calX$ is odd and $\calZ$ is even.
    In both cases, 
    the embedding~$\calX \hookrightarrow \calZ$ factors through
    the minimal extension~$\calY = \calX \otimes \FermionMinus{1}$ of~$\calX$:
    \begin{equation}
        \begin{tikzcd}
        \calX \arrow[dr, hook]\arrow[rr, hook] & &  \calZ \\
         &  \calX \otimes \FermionMinus{1} \arrow[ur, hook]  &
        \end{tikzcd}
        \label{eq:evenXOddZFactorThrough}
    \end{equation}
    where the injection $\calX \ni x \mapsto x \otimes \one \in \calX \otimes \FermionMinus{1}$ is unital.
\end{lemma}
\begin{proof}
    ($\calX$ is even and $\calZ$ is odd)
    Since $\calX$ is even, there is a self-adjoint unitary even element~$\Gamma_\calX \in \calX$,
    the conjugation by which is the grading automorphism of~$\calX$.
    Since $\calZ$ is odd, there is a self-adjoint unitary odd central (not in the fermionic center)
    element~$\gamma_\calZ \in \calZ$ by~\ref{thm:FiniteDimensionalCentralIsSimple}.
    The product~$\gamma_\calZ \Gamma_\calX$ is odd,
    and fermionically commutes with all of~$\calX$.
    The subalgebra of~$\calZ$ generated by~$\calX$ and $\gamma_\calZ \Gamma_\calX$
    is isomorphic to~$\calX \otimes \FermionMinus{1}$.

    ($\calX$ is odd and $\calZ$ is even)
    Let $\gamma_\calX \in \calX$ be a self-adjoint unitary odd central element,
    and $\Gamma_\calZ \in \calZ$ be a self-adjoint unitary even element,
    the conjugation by which is the grading automorphism of~$\calZ$.
    Then, $\ii \gamma_\calX \Gamma_\calZ \in \calZ$ is a self-adjoint unitary odd element,
    which fermionically commutes with all of~$\calX$.
    The subalgebra of~$\calZ$ generated by~$\calX$ and $\ii \gamma_\calX \Gamma_\calZ$
    is isomorphic to~$\calX \otimes \FermionMinus{1}$.    
\end{proof}

\subsubsection{Odd \texorpdfstring{$\calX$}{X}, Odd \texorpdfstring{$\calZ$}{Z}}

A minimal faithful representation of~$\calZ$
has the same representation dimension 
as the minimal even extension~$\calZ \otimes \FermionMinus{1}$.
So, we examine the embedding~$\calX \hookrightarrow \calZ \otimes \FermionMinus{1}$
where the extra Majorana algebra~$\FermionMinus{1}$ is represented by
$\one_m \otimes \sigma^x$,
which is to say that $\calZ \otimes \FermionMinus{1}$
is identified with $2m \times 2m$ matrix algebra,
with the grading involution implemented by the conjugation by~$\one_m \otimes \sigma^z$.
The odd algebra~$\calZ$ is then identified with the commutant of~$\one_m \otimes \sigma^x$.

The even part~$\calX_0$ of~$\calX$ is ungraded simple 
and hence is isomorphic to a matrix algebra~$\Mat(n)$.
Using an even unitary of~$\calZ$,
we may bring the embedding of~$\calX_0$ to a block-diagonal form
\begin{equation}
    \calX_0 = \left\{ \left(\begin{array}{cc|cc}
        A \otimes \one_k & 0 &  & \\
        0 & 0 &  & \\
        \hline 
         &  & A \otimes \one_k & 0 \\
         &  & 0 & 0
    \end{array}\right)\,\, \middle|\,\,
    A \in \Mat(n)
    \right\}.
\end{equation}
The odd center generator~$\gamma \in \calX$ can be chosen to be self-adjoint 
and of square~$\one_\calX$.
Since $\gamma$ commutes with~$\calX_0$,
we must have
\begin{equation}
    \gamma = \left(\begin{array}{cc|cc}
         &  & \one_n \otimes X & 0\\
         &  & 0 & 0\\
        \hline 
        \one_n \otimes X & 0 &  &  \\
        0 & 0 &  & 
    \end{array}\right)\, , \quad X^2 = \one_k \, ,\,  X^\dagger = X \, . \label{eq:gammaRep}
\end{equation}
We may further conjugate~$\gamma$ by some 
even unitary\footnote{An odd unitary may be used but it is no more powerful.} 
of~$\calZ$ to diagonalize~$X$.
It follows that the injection~$\calX \to \calZ$ is unitarily equivalent to
\begin{align}
    \calX \ni
    \left(\begin{array}{c|c}
        A_{n \times n} & B_{n \times n} \\ 
        \hline 
        B_{n \times n} & A_{n \times n}
    \end{array}\right)
    \mapsto
    \left(\begin{array}{ccc|ccc}
        A \otimes \one_{k_+} &  &  & B \otimes \one_{k_+} &  & \\
         & A \otimes \one_{k_-} & &  & -B \otimes \one_{k_-} & \\
         &  & 0 &  & & 0 \\
        \hline
        B \otimes \one_{k_+} &  &   & A \otimes \one_{k_+} &  &  \\
         & -B \otimes \one_{k_-} &  &  & A \otimes \one_{k_-} &  \\
         &  & 0 &  &  & 0
    \end{array}\right)
    \in \calZ \label{eq:oddXoddZ}
\end{align}
for some $k_+,k_- \ge 0$ with $k_+ + k_- > 0$.

\subsubsection{Odd \texorpdfstring{$\calX$}{X}, Even \texorpdfstring{$\calZ$}{Z}}
Since the even part~$\calX_0$ of~$\calX$ is simple
as an ungraded algebra,
it is isomorphic to~$\Mat(n)$ for some~$n$.
Using an even unitary of~$\calZ$,
we may bring the embedding of~$\calX_0$ to a block diagonal form
\begin{equation}
    \calX_0 = \left\{ \left(\begin{array}{cc|cc}
        A \otimes \one_{k} & 0 &  & \\
        0 & 0 &  & \\
        \hline 
         &  & A \otimes \one_{k'} & 0 \\
         &  & 0 & 0
    \end{array}\right)\,\, \middle|\,\,
    A \in \Mat(n)
    \right\} \, .
\end{equation}
The odd center generator~$\gamma \in \calX$ can be chosen to be self-adjoint 
and of square~$\one_\calX$.
This forces $k = k'$.
The rest of the analysis is identical to that of the
odd-$\calX$-odd-$\calZ$ case,
and the embedding $\calX \hookrightarrow \calZ$
is of form~\eqref{eq:oddXoddZ} up to unitaries of~$\calZ$.
But the minus sign in \eqref{eq:oddXoddZ}
can be removed by an even unitary of~$\calZ$
since the upper and lower blocks of the unitary 
do not have to be synchronized.
Hence, only the sum $k_+ + k_-$ is determined by the embedding,
and we obtain a simpler characterization:
\begin{align}
    \calX \ni
    \left(\begin{array}{c|c}
        A_{n \times n} & B_{n \times n} \\ 
        \hline 
        B_{n \times n} & A_{n \times n}
    \end{array}\right)
    \mapsto
    \left(\begin{array}{cc|cc}
        A \otimes \one_{k} & 0 & B \otimes \one_{k} & 0 \\
        0 & 0 & 0 & 0 \\
        \hline
        B \otimes \one_{k} & 0  & A \otimes \one_{k} & 0 \\
        0 & 0 & 0 & 0
    \end{array}\right)
    \in \calZ 
\end{align}

\subsubsection{Even \texorpdfstring{$\calX$}{X}, Odd \texorpdfstring{$\calZ$}{Z}}

Using~\ref{lem:EvenOddXZFactorThrough},
we may first consider the embedding~$\calX \otimes \FermionMinus{1} \hookrightarrow \calZ$
and use the analysis of the previous case.
The subalgebra~$\calX$ is identified with the 
fermionic commutant of an odd self-adjoint unitary element~$\gamma \in \calX \otimes \FermionMinus{1}$
by~\ref{thm:CancellationPropertyForIsomorphisms}.
In general we have~$\gamma$ represented in the minimal faithful representation 
of~$\calX \otimes \FermionMinus{1}$ 
as~\eqref{eq:gammaRep} with~$n=1$
but without any zero columns or rows.
Using an even unitary of~$\calX \otimes \FermionMinus{1}$,
we diagonalize~$X$ that is displayed in~\eqref{eq:gammaRep}.
The even part of~$\calX$ must commute with~$X$ while the odd part must anticommute.
Solving these commutativity equations,
one checks that the multiplicities of the eigenvalues $\pm 1$ of~$X$
must be equal to the two integers of the rank tuple of~$\calX$.
It follows that $\calX \hookrightarrow \calZ$ is unitarily equivalent to
\begin{align}
    &\calX \ni \left(\begin{array}{c|c}
        A_{p \times p} & B_{p \times q} \\ 
        \hline 
        C_{q \times p} & D_{q \times q}
    \end{array}\right)
    \mapsto \nonumber \\ 
     &\left(\begin{array}{ccccc|ccccc}
         A^{\oplus k_{+}} & 0 &  &  & & 0 & B^{\oplus k_{+}} & & & \\
         0 & D^{\oplus k_{+}} &  & & & C^{\oplus k_{+}} & 0 & & & \\
         & & A^{\oplus k_{-}} & 0 &  &  & & 0 & -B^{\oplus k_{-}} & \\
         & & 0 & D^{\oplus k_{-}} &  & & &-C^{\oplus k_{-}} & 0  \\
          &  &  &  & 0 & &  & & &  0 \\
          \hline
        0 & B^{\oplus k_+} &  & & & A^{\oplus k_+} & 0 &  & & \\
        C^{\oplus k_+} & 0 & &  &  & 0 & D^{\oplus k_+} &  & &\\
        & & 0 & -B^{\oplus k_-} &  & & & A^{\oplus k_-} & 0 &  \\
        & &-C^{\oplus k_-} & 0 & &  &  & 0 & D^{\oplus k_-} & \\
        & &  &  & 0 &  & & & & 0
     \end{array}\right) \in \calZ \, . \label{eq:evenX-OddZ}
\end{align}
for some integers $k_+, k_- \ge 0$ with $k_+ + k_- > 0$.
Here, the notation~$M^{\oplus n}$ for a matrix~$M$ and an integer~$n$ means $M \otimes \one_n$.

\subsubsection{Even \texorpdfstring{$\calX$}{X}, Even \texorpdfstring{$\calZ$}{Z}}

An even fermionic algebra~$\calX$ is, as an ungraded algebra,
isomorphic to a full matrix algebra~$\Mat(n)$,
and the grading involution is an inner automorphism by
a self-adjoint unitary~$\Gamma_\calX$.
Hence, any representation of~$\calX$ can be determined
by a set of mutually orthogonal projectors~$\pi_j$ of the ambient representation space~$\CC^m$
onto each copy of~$\Mat(n) \cong \End(\pi_j \CC^m)$ 
and the represented unitary of~$\Gamma_\calX$.
If $\calX = \FermionPlus{n_+}{n_-}$, 
then $\abs{\Tr(\Gamma_\calX \pi_j)} = \abs{n_+ - n_-}$ for all~$j$.

In~$\calZ$,
the grading unitary~$\Gamma_\calX$ 
is an even element,
so $\Gamma_\calX$ commutes with a grading unitary~$\Gamma_\calZ$ of~$\calZ$.
We may thus assume that both~$\Gamma_\calX$ and $\Gamma_\calZ$ are diagonal.
Note that $\Gamma_\calX^2$ is the multiplicative identity~$\one_\calX$ for~$\calX$
but may not be equal to~$\one_\calZ$ since $\calX \hookrightarrow \calZ$ may not be unital.
We know that $\sum_j \pi_j = \Gamma_\calX^2$.
Since $\pi_j \Gamma_\calX \Gamma_\calZ $ commutes with all elements of~$\calX$,
we must have $\pi_j \Gamma_\calX = \pm \pi_j \Gamma_\calZ$.
Let $k_\pm = n_\pm^{-1} \Tr ( (\one_\calZ + \Gamma_\calZ)(\one_\calX \pm \Gamma_\calX) / 4) $.
We see that the injection~$\calX \to \calZ$ is unitarily equivalent to 
\begin{align}
    \calX \ni
    \left(\begin{array}{c|c}
        A_{n_+ \times n_+} & B_{n_+ \times n_-} \\ 
        \hline 
        C_{n_- \times n_+} & D_{n_- \times n_-}
    \end{array}\right)
    \mapsto
    \left(\begin{array}{ccc|ccc}
        A \otimes \one_{k_+} &  &  & B \otimes \one_{k_+} &  &  \\
         & D \otimes \one_{k_-} & &  & C \otimes \one_{k_-} & \\
         &  & 0 &  &  & 0 \\
        \hline
        C \otimes \one_{k_+} &  &   & D \otimes \one_{k_+} &  &  \\
         & B \otimes \one_{k_-} &  &  & A \otimes \one_{k_-} &  \\
         &  & 0 &  & & 0
    \end{array}\right)
    \in \calZ  
\end{align}
for some integers~$k_+, k_- \ge 0$ with $k_+ + k_- > 0$.

\subsubsection{Unital inclusions}

We are primarily interested in unital inclusions and their fermionic commutants.
\begin{lemma}[Lemma~A.5 of~\cite{nta3}]\label{thm:Commutant-of-central-in-central}
    Let \(\calA \hookrightarrow \calB\) be a unital inclusion of 
    finite-dimensional central fermionic algebras.
    Then, \(\calC = \FermiCommutant{\calA}{\calB}\) is also central and
    we have a graded isomorphism \(\calB \cong \calA \otimes \calC\).
\end{lemma}

\begin{proof}
    There is an obvious $*$-map $\phi : \calA \otimes \calC \ni a \otimes c \mapsto ac \in \calB$,
    which is a graded homomorphism because $\FermiCommutator{\calA}{\calC} = 0$.
    We will show that $\phi$ is bijective.
    
    (Case i: $\calA$ is even)
    The nonfermionic, usual commutant $\Comm(\calA,\calB) = \calD$ is graded.
    Let $\Gamma_\calA \in \calA$ be an even unitary of order 2, 
    which implements the grading involution of~$\calA$ by conjugation.
    We claim $\calC = \calD_0 \oplus \Gamma_\calA \calD_1$.
    The even part $\calC_0 = \calD_0$ commutes with all~$\calA$.
    If $c_1 \in \calC_1$, then odd element~$\Gamma_\calA c_1$ commutes with all~$\calA_0$ 
    and also with~$\calA_1$ because $\Gamma_\calA$ anticommutes with~$\calA_1$,
    so $\Gamma_\calA c_1 \in \calD_1$, implying $\Gamma_\calA \calC_1 \subseteq \calD_1$.
    A similar argument shows that $\Gamma_\calA \calD_1 \subseteq \calC_1$.
    
    Since $\calB \cong \calA \otimes \calD$ as ungraded central $*$-algebras
    (in both cases where $\calB$ is even or odd),
    we know that $\dim \calD = \dim \calB / \dim \calA$.
    Since $\dim \calD = \dim \calC$, we have $(\dim \calA)(\dim \calC) = \dim \calB$.
    Hence, to show $\phi$ is bijective, it suffices to show that $\phi$ is surjective.
    Since $\calB \cong \calA \otimes \calD$ as ungraded algebras,
    we know every element of~$\calB$ 
    is of form~$\sum_i a_i d_{0,i} + \sum_j a_j d_{1,j}$
    where we have displayed the grading of~$\calD$.
    We can equally write it as
    $\sum_i a_i d_{0,i} + \sum_j (a_j \Gamma_\calA) (\Gamma_\calA d_{1,j})$.
    We see that $\phi$ is surjective.

    (Case ii: $\calA$ is odd)
    Consider a unital inclusion $\varphi: \FermionMinus{1}\otimes \calA \hookrightarrow \FermionMinus{1} \otimes \calB$ 
    of fermionic algebras.
    The fermionic commutant of the image of~$\varphi$ is~$\calC$.
    By~\ref{rem:EvenOddUnderTensorProduct},
    the domain of~$\varphi$ is even central,
    and (Case i) implies that 
    $\FermionMinus{1} \otimes \calA \otimes \calC \cong \FermionMinus{1} \otimes \calB$ as fermionic algebras.
    Here, the isomorphism takes $\FermionMinus{1}$ on the left-hand side to that of the right-hand side,
    and hence by~\ref{thm:CancellationPropertyForIsomorphisms} we conclude the proof.
\end{proof}

\begin{corollary}\label{thm:Bicommutant-of-Central}
    If $\calA \hookrightarrow \calB$ is a unital inclusion of
    a central fermionic algebra~$\calA$ into a finite-dimensional, possibly noncentral,
    fermionic algebra~$\calB$,
    then $\calA$ is a tensor factor of~$\calB$ 
    and the fermionic bicommutant of~$\calA$ is itself.
\end{corollary}
\begin{proof}
    By~\ref{thm:Wedderburn-Fermionic}, we have $\calB = \bigoplus_\mu \pi_\mu \calB$ 
    where $\pi_\mu$ are canonical projectors.
    The composition~$\calA \to \calB \to \pi_\mu \calB$ is unital and hence nonzero;
    since $\calA$ is simple by~\ref{thm:FiniteDimensionalCentralIsSimple},
    it must be injective.
    By~\ref{thm:Commutant-of-central-in-central}, we have $\pi_\mu \calB \cong \calC_\mu \otimes \calA$,
    where $\calC_\mu = \FermiCommutant{\pi_\mu\calA}{\pi_\mu\calB}$,
    so that $\calB = \bigoplus_\mu \calC_\mu \otimes \calA$.
    It follows that $\FermiCommutant{\calA}{\calB} \cong \bigoplus_\mu \calC_\mu \otimes \one$.
    Then the fermionic bicommutant of $\calA$ is $\bigoplus_\mu \pi_\mu \otimes \calA = \calA$.
\end{proof}

\subsection{Trace and Hilbert--Schmidt inner product}

\begin{definition}
    A \emph{tracial state} on a $*$-algebra~$\calA$ 
    is a $*$-linear functional
    \begin{equation}
        \tr : \calA \to \CC
    \end{equation}
    such that
    $\tr(\one) = 1$,
    $\tr(ab) = \tr(ba)$ for all~$a,b \in \calA$,
    $\tr(a^\dagger a) \ge 0$ for all~$a \in \calA$, and
    $\tr(a^\dagger a) = 0$ only if~$a = 0$.
    A tracial state on a fermionic algebra~$\calA = \calA_0 \oplus \calA_1$ is a tracial state on the ungraded $*$-algebra~$\calA$
    such that $\tr(a) = 0$ for all odd~$a \in \calA_1$.
\end{definition}

\begin{proposition}
    A $*$-algebra~$\calA$ equipped with a tracial state~$\tr$,
    is an inner product space under
    \begin{equation}
        \calA \times \calA \ni (x,y) \mapsto \tr(x^\dagger y) \in \CC .
    \end{equation}
\end{proposition}

\begin{proof}
    Automatic by the definition of~$\tr$.
\end{proof}

\begin{proposition}\label{thm:trIsUnique}
    A tracial state on a finite-dimensional central fermionic algebra~$\calA$ 
    is the usual matrix trace~$\Tr$
    on the minimal faithful representation 
    divided by the envelope rank of~$\calA$.
    In particular, there exists a unique $\tr$ on~$\calA$.
\end{proposition}

\begin{proof}
    If $\calA$ is even,
    it suffices to show that the ungraded trace functional is unique.
    If $\sum_{i=1}^n \pi_i = \one$ 
    is a decomposition with mutually orthogonal minimal orthogonal projectors,
    then the unitary invariance of~$\tr$, which follows from the cyclic property,
    implies that $\tr(\pi_i)$
    is independent of~$i$.
    Since $\tr(\one) = 1$, we must have $\tr(\pi_i) = 1/n$.
    The cyclic property forces $\tr(M) = 0$ 
    for any matrix~$M$ with zero diagonal.
    This determines~$\tr$ on~$\Mat(n)$.

    If $\calA$ is odd,
    then $\tr$ on $\calA$ restricts to the ungraded trace on~$\calA_0$,
    which is unique.
    Because $\tr$ is zero on~$\calA_1$,
    this determines~$\tr$.
    By~\ref{thm:faithfulreps}, we see that $\tr$ coincides
    with the normalized matrix trace.
\end{proof}
\noindent
Note that the positivity ($\tr(a^\dagger a) \ge 0$) and 
the nondegeneracy ($\tr(a^\dagger a) = 0 \Longrightarrow a = 0$) conditions on the trace
are redundant for finite-dimensional central simple algebras.
It follows that for any finite-dimensional fermionic algebra, not necessarily simple,
there exists a tracial state, which may not be unique.
However, the positivity and nondegeneracy conditions
imply that any tracial state must assign a positive real number to
every canonical central projector.

\subsection{Simple between semisimple}

\begin{proposition}\label{thm:FindYinbetweenXZ}
    Let $\phi : \calX \hookrightarrow \calZ$ be a unital injective homomorphism 
    of a fermionic algebra~$\calX$ 
    into a finite-dimensional fermionic algebra~$\calZ$.
    Let $\pi_\mu \in \calX$ and $\tau_\nu \in \calZ$ 
    be canonical central projectors of~$\calX$ and $\calZ$, respectively.
    Endow a tracial state~$\tr$ on~$\calZ$
    and suppose that 
    \begin{equation}
        \tr(\phi(\pi_\mu) \tau_\nu) \tr(\one_\calZ) = \tr(\phi(\pi_\mu)) \tr(\tau_\nu) \label{eq:tracecondition}
    \end{equation}    
    for all~$\mu,\nu$.
    Then, there exists a \emph{central} fermionic algebra~$\calY$
    with a commuting diagram of unital injective graded homomorphisms
    \begin{equation}
        \begin{tikzcd}
             & \calY \ar[rr, hook] & & \calZ \otimes \FermionMinus{1}\\
            \calX \ar[ru, hook ] \ar[rr, "\phi", hook] & & \calZ \ar[ru, "\otimes \one_{\FermionMinus{1}}"', hook] &
        \end{tikzcd} .
    \end{equation}
\end{proposition}

\begin{proof}
    Write $\calF$ for $\FermionMinus{1}$.
    We reduce the proof to the case
    where all simple components~$\pi_\mu \calX$ and~$\tau_\nu \calZ$ are odd.
    In this reduction will we use the stabilization~$\otimes \calF$.
    Define
    \begin{align}
        \bar \calX &= \bigoplus_\mu \bar \calX_\mu & \bar \calZ &= \bigoplus_\nu \bar \calZ_\nu \nonumber\\
        \bar \calX_\mu &= 
        \begin{cases}
            \pi_\mu \calX \otimes \calF & \text{if } \pi_\mu \calX \text{ is even},\\
            \pi_\mu \calX & \text{otherwise.}
        \end{cases}
        &
        \bar \calZ_\nu &= 
        \begin{cases}
            \tau_\nu \calZ \otimes \calF & \text{if } \tau_\nu \calZ \text{ is even},\\
            \tau_\nu \calZ & \text{otherwise.}
        \end{cases}
    \end{align}
    We claim that there is a unital monomorphism 
    $\bar \phi : \bar \calX \hookrightarrow \bar \calZ$
    such that the following diagram of unital monomorphisms commutes:
    \begin{equation}
        \begin{tikzcd}
        \bar \calX\ar[r, "\bar \phi", hook] & \bar \calZ \ar[r,hook] & \calZ \otimes \calF \\
        \calX \ar[u,hook]\ar[r,hook, "\phi"]  &  \calZ \ar[ru,hook,"{\otimes \one_\calF}"']   &
        \end{tikzcd}
    \end{equation}
    where for the unannotated arrows we tensor $\one_\calF$ as needed.
    The only case that is less obvious for $\bar \phi$ is 
    when $\pi_\mu\calX$ is even and $\tau_\nu\calZ$ is odd,
    so we have to define a monomorphism 
    $\pi_\mu \calX \otimes \calF \hookrightarrow \tau_\nu \calZ$.
    In that case we can use~\ref{lem:EvenOddXZFactorThrough}.
    We will check the injectivity of~$\bar \phi$ shortly.    
    
    The canonical central projectors of~$\bar \calZ$ are obviously 
    either $\bar \tau_\nu = \tau_\nu$ if $\tau_\nu \calZ$ is odd
    or $\bar \tau_\nu = \tau_\nu \otimes \one_\calF$ otherwise.
    The canonical central projectors $\bar \pi_\mu \in \bar \calX$ are analogous.
    Since we have constructed~$\bar \phi$ on each pair~$(\mu,\nu)$,
    the images~$\bar \phi(\bar \pi_\mu)$ deserve attention.
    Clearly, $\bar \phi(\bar \pi_\mu) = \sum_\nu  \bar \phi(\bar \pi_\mu) \bar \tau_\nu$,
    but each summand here needs to be examined case by case:
    \begin{equation}
        \bar \phi(\bar \pi_\mu) \bar \tau_\nu =
        \begin{cases}
            \phi(\pi_\mu) \tau_\nu & \text{if $\pi_\mu \calX$ is odd and $\tau_\nu \calZ$ is odd},\\
            \phi(\pi_\mu) \tau_\nu & \text{if $\pi_\mu \calX$ is even and $\tau_\nu \calZ$ is odd},\\
            (\phi(\pi_\mu) \otimes \one_\calF) (\tau_\nu \otimes \one_\calF) & \text{if $\pi_\mu \calX$ is odd and $\tau_\nu\calZ$ is even},\\
            (\phi(\pi_\mu) \otimes \one_\calF) (\tau_\nu \otimes \one_\calF) & \text{if $\pi_\mu \calX$ is even and $\tau_\nu\calZ$ is even}.
        \end{cases}
    \end{equation}
    In the second case, we used the factor-through property from~\ref{lem:EvenOddXZFactorThrough};
    the other cases are only easier.
    It is now clear that $\bar \phi(\bar \pi_\mu) \neq 0$, and hence $\bar \phi$ is injective.
    
    We have to check the trace condition~\eqref{eq:tracecondition} to complete the reduction.
    We first have to find a tracial state~$\bar \tr$ on~$\bar \calZ$ 
    that restricts to~$\tr$ on~$\calZ$.
    By~\ref{thm:trIsUnique},
    every simple algebra has a unique tracial state,
    and hence the only allowed variation of a tracial state on a semisimple algebra
    is its value on the canonical central projectors.
    Since we require that $\bar \tr$ restricts to~$\tr$,
    we have $\bar \tr(\bar \tau_\nu) = \tr(\tau_\nu)$,
    and therefore $\bar \tr$ is uniquely determined by~$\tr$.
    Equipped with~$\bar \tr$, we easily have the trace condition~\ref{eq:tracecondition} 
    for~$\bar \calX \hookrightarrow \bar \calZ$:
    \begin{align}
        \frac{\tr(\tau_\nu)}{\tr(\one_\calZ)} = \frac{\tr(\tau_\nu \otimes \one_\calF)}{\tr(\one_\calZ \otimes \one_\calF)}, \qquad
        \frac{\tr(\phi(\pi_\mu)\tau_\nu)}{\tr(\tau_\nu)} = \frac{\tr((\phi(\pi_\mu)\otimes \one_\calF)(\tau_\nu \otimes \one_\calF))}{\tr(\tau_\nu \otimes \one_\calF)} \, .
    \end{align}
    Since $\calX \hookrightarrow \bar \calX$ and $\bar \calZ \hookrightarrow \calZ \otimes \calF$ are unital,
    we aim to find~$\calY$ between $\bar \calX$ and $\bar \calZ$
    where all components are odd.
    This completes the reduction.

    When all simple components are odd,
    we recall the fact that 
    every odd central fermionic algebra of envelope rank~$2n$
    is a graded tensor product~$\FermionPlus{n}{0} \otimes \calF$.
    Pulling out the common tensor factor~$\calF$,
    the proof further reduces to the situation 
    where all simple components are even with the trivial grading involution.
    Hence, a desired odd central algebra~$\calY$ is found by
    appealing to \cite[Lemma~3.8]{FreedmanHastings2019QCA},
    which we restate and reprove in~\ref{thm:FindYinbetweenXZBoson} below.
\end{proof}

\begin{remark}\label{rem:rankIgnorantTrace}
    Observe that the claim of~\ref{thm:FindYinbetweenXZ} 
    allows $\tr(\tau_\nu)$ to be unrelated to the envelope rank of~$\tau_\nu \calZ$;
    even if the envelope ranks of $\tau_1 \calZ$ and $\tau_2 \calZ$ are the same, 
    $\tr(\tau_1)$ and $\tr(\tau_2)$ may still be different.
    In terms of concrete matrices,
    this is rooted in the fact that the claim has nothing to do with
    any representation of~$\calZ$ 
    in which different simple subalgebras of~$\calZ$
    may appear with arbitrarily multiplicities.
    This remark also applies to~\ref{thm:FindYinbetweenXZBoson} below.
\end{remark}

\begin{remark}\label{rem:CounterexampleSimpleBetweenSemisimple}
    The stabilization in~\ref{thm:FindYinbetweenXZ} cannot be omitted.
    Here we give an example to this end.
    Consider fermionic algebras 
    \begin{align}
        \pi_1 \calX &= \FermionMinus{1} = \left\{ 
            \left(\begin{array}{c|c}
                x & y \\ \hline y & x
            \end{array}\right) ~\middle|~ x,y \in \CC
        \right\}\, ,\\
        \pi_2 \calX &= \FermionPlus{1}{1} = \left\{
            \left(\begin{array}{c|c}
                 a & b  \\ \hline c & d 
            \end{array}\right)
            ~\middle|~
            a,b,c,d \in \CC
            \right\} \, ,\nonumber
    \end{align}
    unitally embedded in~$\calZ = \tau_1 \calZ \oplus \tau_2 \calZ$ where $\tau_1 \calZ = \FermionMinus{3}$ 
    and $\tau_2 \calZ = \FermionPlus{3}{3}$ as
    \begin{equation}
        \left(\begin{array}{ccc|ccc}
            x &   &   & y &   &  \\ 
             & a & 0 &  & 0 & b \\
             & 0 & d &  & c & 0 \\
            \hline
            y &  &  & x &  &  \\
             & 0 & b &  & a & 0 \\
             & c & 0 &  & 0 & d 
        \end{array}\right) \in \tau_1 \calZ \, ,
        \qquad
        \left(\begin{array}{ccc|ccc}
            x &   &   & y &   &  \\ 
             & a & 0 &  & b & 0 \\
             & 0 & a &  & 0 & b \\
            \hline
            y &  &  & x &  &  \\
             & c & 0 &  & d & 0 \\
             & 0 & c &  & 0 & d 
        \end{array}\right) \in \tau_2 \calZ \, . \label{eq:repEOEO}
    \end{equation}
    Defining $\tr(\tau_1) = \tr(\tau_2) = \frac 1 2$, 
    we see that $\tr(\pi_1) = \frac 1 3$, $\tr(\pi_2) = \frac 2 3$,
    and the trace condition of~\ref{thm:FindYinbetweenXZ} is satisfied:
    $\tr(\pi_1 \tau_1) / \tr(\tau_1) = \tr(\pi_1 \tau_2) / \tr(\tau_2) = \tr(\pi_1) / \tr(\one)$.
    However, \emph{there is no central fermionic algebra~$\calY$ with unital 
    inclusions such that $(\calX \hookrightarrow \calY \hookrightarrow \calZ) = (\calX \hookrightarrow \calZ)$.}
    \begin{proof}
        By~\ref{thm:Commutant-of-central-in-central} the hypothetical~$\calY$ is a tensor factor 
        of~$\tau_1 \calZ$ and of~$\tau_2 \calZ$.
        By~\ref{thm:RankTupleMonoid}, the rank tuple of~$\calY$ must divide~$(3,3)/\sqrt{2}$.
        The divisors of~$(3,3)/\sqrt{2}$ are 
        $(1,0)$, $(3,0)$, $(2,1)$, $(1,1)/\sqrt{2}$, and $(3,3)/\sqrt{2}$.
        On the other hand, $\calY$ must have the dimension of the even part $\ge 3$ and the odd part $\ge 3$
        by inspection of the embedding.
        This rules out $(1,0)$, $(3,0)$, and $(1,1)/\sqrt{2}$.
        
        If $\calY \cong \FermionPlus{2}{1}$,
        then we must have an isomorphism $\calY \otimes \FermionMinus{1} \cong \tau_1 \calZ$.
        Hence, we must have an odd element~$\gamma = \gamma^\dagger \in \tau_1 \calZ$ with~$\gamma^2 = \one$
        such that $\gamma$ anticommutes with every odd element of~$\calX$.
        Such an element~$\gamma$ is represented as a $3$-by-$3$ off-diagonal block~$U$ such that $U^2 = \one_3$.
        Setting $y = 1$ and $c = b = 0$ in~\eqref{eq:repEOEO},
        we have an equation
        \begin{equation}
            U \begin{pmatrix} 1 & & \\  & 0 & \\ & & 0 \end{pmatrix} U = \begin{pmatrix} -1 & & \\  & 0 & \\ & & 0 \end{pmatrix}\, ,
        \end{equation}
        which contradicts~$U^2  = \one_3$.
        
        If $\calY \cong \FermionMinus{3}$,
        then we must have an isomorphism $\calY \otimes \FermionMinus{1} \cong \tau_2 \calZ$.
        The same calculation with an odd element~$\gamma \in \tau_2 \calZ$ shows that this is impossible.
    \end{proof}
\end{remark}

\begin{example}
    Continuing the example in~\ref{rem:CounterexampleSimpleBetweenSemisimple},
    we examine how the stabilization~$\calZ \to \calZ \otimes \FermionMinus{1}$ 
    reveals a sandwiched central algebra,
    as promised by~\ref{thm:FindYinbetweenXZ}.
    According to~\eqref{eq:evenX-OddZ},
    the embedding of~$\calX$ in the stabilized component~$\bar \calZ_2 = \tau_2 \calZ \otimes \FermionMinus{1}$
    can be written as
    \begin{equation}
        \left(\begin{array}{c|c}
            A & B  \\ \hline B & A
        \end{array}\right) \in \bar \calZ_2,
        \quad
        A = 
        \begin{pmatrix}
            x&0&0& & & \\
            0&a&0& & & \\
            0&0&a& & & \\
             & & &x&0&0\\
             & & &0&d&0\\
             & & &0&0&d
        \end{pmatrix},
        \quad
        B = 
        \begin{pmatrix}
             & & &y&0&0\\
             & & &0&b&0\\
             & & &0&0&b\\
            y&0&0& & & \\
            0&c&0& & & \\
            0&0&c& & &
        \end{pmatrix}\, .
    \end{equation}
    To recognize a central algebra, we conjugate it by a permutation unitary
    \begin{equation}
        \left(\begin{array}{c|c}
            S & 0  \\ \hline 0 & S
        \end{array}\right) \in \bar \calZ_2,
        \qquad
        S = 
        \begin{pmatrix}
            1&0& & & & \\
            0&1& & & & \\
            & &0&0&0&1\\
             & &0&1&0&0\\
             & &1&0&0&0\\
             & &0&0&1&0
        \end{pmatrix}
    \end{equation}
    to obtain
    \begin{equation}
        \left(\begin{array}{c|c}
            A' & B'  \\ \hline B' & A'
        \end{array}\right) \in \bar\calZ_2,
        \,\,
        A' = 
        \begin{pmatrix}
            x&0&0& & & \\
            0&a&0& & & \\
            0&0&d& & & \\
             & & &x&0&0\\
             & & &0&a&0\\
             & & &0&0&d
        \end{pmatrix},\,\,
        B' = 
        \begin{pmatrix}
             & & &y&0&0\\
             & & &0&0&b\\
             & & &0&c&0\\
            y&0&0& & & \\
            0&0&b& & & \\
            0&c&0& & &
        \end{pmatrix}\, .
    \end{equation}
    We see that there are two copies of~$\FermionMinus{3}$ in~$\bar \calZ_2$,
    in which the embedding of~$\calX$ is the same as that in~$\tau_1 \calZ$.
    Hence, $\calY \cong \FermionMinus{3}$.
\end{example}

\begin{proposition}[Lemma 3.8 of \cite{FreedmanHastings2019QCA}]\label{thm:FindYinbetweenXZBoson}
    Let $\calX \subseteq \calZ$ be a unital inclusion 
    of an ungraded $*$-algebra~$\calX$ 
    into a finite-dimensional ungraded $*$-algebra~$\calZ$.
    Let $\pi_\mu \in \calX$ and $\tau_\nu \in \calZ$ 
    be canonical central projectors of~$\calX$ and $\calZ$, respectively,
    and suppose that $\tr(\pi_\mu \tau_\nu) \tr(\one_\calZ) = \tr(\pi_\mu) \tr(\tau_\nu)$ for all~$\mu,\nu$.
    Then, there exists a \emph{central} ungraded $*$-algebra~$\calY$
    with unital inclusions $\calX \subseteq \calY \subseteq \calZ$.
\end{proposition}

We present a proof that is the same in essence as that of~\cite{FreedmanHastings2019QCA},
but we explicitly use mathematical induction
in which only the base case needs some calculation.

\begin{proof}
    Let $C(m,n)$ be the statement of~\ref{thm:FindYinbetweenXZBoson} 
    with the number of projectors specified,
    i.e., the canonical central projectors are $\pi_1, \ldots, \pi_m \in \calX$
    and $\tau_1,\ldots,\tau_n \in \calZ$ and no more.

    We show that $C(2,n)$ is true if $C(2,2)$ and $C(2,n-1)$ are true.
    Let $\calZ_{1,2} = (\tau_1 + \tau_2)\calZ$
    and $\calX_{1,2} = (\tau_1 + \tau_2)\calX$.
    Clearly, $\calX_{1,2} \subseteq \calZ_{1,2}$.
    Furthermore, this is unital since the multiplicative identity
    is $(\tau_1 + \tau_2)\one_\calX = (\tau_1 + \tau_2)\one_\calZ$.
    The canonical projectors are $\pi_\mu (\tau_1 + \tau_2) \in \calX_{1,2}$
    with $\mu = 1,2$
    and $\tau_\nu \in \calZ_{1,2}$ with $\nu = 1,2$.
    The trace condition is checked as
    $\tr((\pi_\mu (\tau_1 + \tau_2)) \tau_\nu) \tr(\tau_1 + \tau_2) 
    = \tr(\pi_\mu \tau_\nu) \tr(\tau_1 + \tau_2) = \tr(\pi_\mu) \tr(\tau_\nu) \tr(\tau_1 + \tau_2) = \tr(\pi_\mu(\tau_1 + \tau_2)) \tr(\tau_\nu)$.
    By $C(2,2)$, we find a central~$\calY'$ with unital inclusions 
    $\calX_{1,2} \subseteq \calY' \subseteq \calZ_{1,2}$.
    Now, consider unital inclusions 
    \begin{equation}
        \calX \subseteq \calZ' = \left(\calY' \oplus \bigoplus_{j = 3}^n \tau_j \calZ \right) \subseteq \calZ.
    \end{equation}
    The new algebra~$\calZ'$ has canonical projectors
    $\tau_1 + \tau_2, \tau_3, \tau_4, \ldots, \tau_n$,
    the number of which is $n-1$.
    The trace condition for $\calX \subseteq \calZ'$ is clearly true.
    By~$C(2,n-1)$ applied to $\calX \subseteq \calZ'$,
    we find a central $\calY$ 
    with unital inclusions $\calX \subseteq \calY \subseteq \calZ' \subseteq \calZ$ as desired.
    Since $C(2,1)$ is trivially true,
    if $C(2,2)$ is true, then $C(2,n)$ is true for all~$n$.

    Next, we show that $C(m,n)$ is true if $C(2,n)$ and $C(m-1,n)$ is true.
    Let $\calZ^{1,2} = (\pi_1 + \pi_2) \calZ (\pi_1 + \pi_2)$ and $\calX^{1,2} = (\pi_1 + \pi_2)\calX$.
    The inclusion $\calX^{1,2} \subseteq \calZ^{1,2}$ is unital with the common multiplicative identity~$\pi_1 + \pi_2$.
    The canonical projectors are $\pi_\mu \in \calX^{1,2}$ with $\mu = 1,2$
    and $(\pi_1 + \pi_2)\tau_\nu \in \calZ^{1,2}$ with $\nu = 1,2,\ldots,n$.
    The trace condition $\tr(\pi_\mu ((\pi_1+\pi_2)\tau_\nu)) \tr(\pi_1 + \pi_2) = \tr(\pi_\mu) \tr((\pi_1+\pi_2)\tau_\nu)$ 
    is easily checked as before.
    So, by~$C(2,n)$ applied to~$\calX^{1,2} \subseteq \calZ^{1,2}$
    we have a central~$\calY''$, 
    whose multiplicative identity is~$\pi_1 + \pi_2$,
    with unital inclusions~$\calX^{1,2} \subseteq \calY'' \subseteq \calZ^{1,2}$.
    This implies that
    \begin{equation}
        \calX' = \left(\calY'' \oplus \bigoplus_{j=3}^m \pi_j \calX \right) \subseteq \calZ
    \end{equation}
    is a unital inclusion.
    The canonical projectors are
    $\pi_1 + \pi_2, \pi_3, \ldots, \pi_m \in \calX'$,
    the number of which is $m-1$,
    and $\tau_1, \ldots, \tau_n \in \calZ$.
    The trace condition for $\calX' \subseteq \calZ$ is obvious.
    By~$C(m-1,n)$ applied to~$\calX' \subseteq \calZ$, 
    we find a sandwiched~$\calY$ that is desired for~$C(m,n)$.
    
    Since $C(1,1),C(1,2),C(2,1)$ are trivially true, 
    we conclude that $C(2,2)$ implies $C(m,n)$
    for all~$m,n$.

    To show~$C(2,2)$,
    we write $\calX = \calA \oplus \calB$ and $\calZ = \calC \oplus \calD$ 
    where $\calA,\calB,\calC,\calD$ are simple algebras.
    If we identify $\calC$ with $\End(C)$,
    then $C$ becomes a representation space of~$\calA \oplus \calB$.
    Since $\calX \to \calC$ is unital,
    the space~$C$ is an orthogonal sum of an $\calA$-module~$C_\calA$ and a $\calB$-module~$C_\calB$,
    where $\calA \to \End(C_\calA)$ and $\calB \to \End(C_\calB)$ are both unital.
    Since $\calA$ is simple, $C_\calA$ is an orthogonal sum of some number, say~$a$, of copies of
    the unique simple $\calA$-module~$A$; 
    that is, $C_\calA \cong A^{\oplus a}$ where $\dim A = \rank \calA$.
    Likewise, $C_\calB \cong B^{\oplus b}$ for some positive integer~$b$.
    Completely analogously, with $\calD = \End(D_\calA \oplus D_\calB)$,
    we have $D_\calA \cong A^{\oplus a'}$ and $D_\calB \cong B^{\oplus b'}$
    for some positive integers $a',b'$.
    The trace condition implies that
    \begin{equation}
        \frac{a \rank \calA}{a \rank \calA + b \rank \calB} 
        = \frac{a' \rank \calA}{a' \rank \calA + b' \rank \calB} 
        \quad \Longleftrightarrow \quad
        \frac{b}{a} = \frac{b'}{a'} \, .
    \end{equation}
    Let $p,q$ be coprime positive integers such that $p/q = a/b = a'/b'$ where $p \le a$.
    Define $Y = A^{\oplus p} \oplus B^{\oplus q}$ and $\calY = \End(Y)$.
    Since $\rank \calC = k \rank \calY$ where $k = a/p = b/q$ 
    and $\rank \calD = k' \rank \calY$ where $k' = a'/p = b'/q$,
    we see that there exists a commuting diagram of unital monomorphisms:
    \begin{equation}
    \begin{tikzcd}
    \calX \arrow[dr, hook]\arrow[rr, hook] & &  \calZ \\
     &  \calY \arrow[ur, hook]  &
    \end{tikzcd} \qedhere
\end{equation}
\end{proof}

\section{Fermionic invertible subalgebras}

Consider a $\dd$-dimensional lattice $\ZZ^\dd$
for an integer $\dd\ge 0$.
We construct the ambient operator algebra as follows.
On each site $s$, we choose a
finite-dimensional central simple fermionic algebra
$\fMat(\{s\}, p(s))$, where $p(s)$ represents
the data of isomorphism class of the algebra (the rank tuple).
We will sometimes call $p(s)$ a \emph{local assignment}.
Local assignments form a multiplicative monoid.
For any finite subset $S\subset \ZZ^\dd$,
we define $\fMat(S,p) = \bigotimes_{s\in S}\fMat(\{s\}, p(s))$
to be the fermionic tensor product of
the local operator algebras.
The full operator algebra is the direct limit
\begin{equation}
    \fMat(\ZZ^\dd, p) = \bigcup_{S} \fMat(S, p).
\end{equation}
In what follows, we will suppress the local assignment $p$ 
where it is not important.

\begin{definition}
    An \emph{approximately finite} fermionic algebra
    is the inductive limit of a directed system
    of unital inclusions of finite-dimensional fermionic algebras.%
    \footnote{In contrast to the typical usage of the term ``approximately finite''
    for $C^*$-algebras,
    we do not consider any norm topology on algebras
    and hence we do not consider topological closure,
    though there is a $C^*$-norm in all our examples.}
\end{definition}

Clearly, $\fMat(\ZZ^\dd)$ is an approximately finite fermionic algebra.

\begin{proposition}
    The fermionic tensor product of approximately finite fermionic algebras is fermionic.
\end{proposition}

\begin{proof}
    By~\ref{thm:Wedderburn-Fermionic} every finite-dimensional fermionic algebra
    is a finite direct sum of central simple fermionic algebras,
    whose tensor product is again fermionic; that is, the spectral nondegeneracy is obeyed.
    Every element in the tensor product of two approximately finite fermionic algebras
    is in some finite-dimensional tensor product,
    where the spectral nondegeneracy is clear.
\end{proof}

We note that the construction of~$\fMat$ is a
strict generalization of~\cite{Haah2023invertible}.
If the odd part of $\fMat(\{s\})$ is zero,
then we can forget the grading structure
and the single-site operator algebra is bosonic.
The only central simple finite-dimensional
ungraded algebras over $\CC$
are full matrix algebras,
which is exactly the case considered in~\cite{Haah2023invertible}.

For an operator $x\in \fMat(\ZZ^\dd)$,
the smallest $S\subset \ZZ^\dd$ such that
$x\in \fMat(S)$ is the \emph{support} of $x$,
which we denote by~$\Supp(x)$.
Scalar multiples of $\one$ have empty support.
For any set of sites $S\subseteq \ZZ^\dd$
and $\ell > 0$, we will write $S^{+\ell}$ to denote the
set of sites whose $\ell_\infty$-distance
from~$S$ is at most~$\ell$.
A subalgebra of $\fMat(\ZZ^\dd)$ is
$\ell$-\emph{locally generated}
if there exists a generating set for the subalgebra
such that every generator is supported on a region of
diameter at most $\ell$.

\begin{lemma}
    \label{lem:supp_by_commutation}
    For any $x\in \fMat(\ZZ^\dd)$
    and any subset $S\subseteq \ZZ^\dd$,
    the support~$\Supp(x)$ of~$x$ is contained in~$S$ if and only if
    $\FermiCommutator{x}{y}=0$ for all
    $y\in \fMat(\ZZ^\dd\setminus S)$.
\end{lemma}

\begin{proof}
    We have $\fMat(\ZZ^\dd) = \fMat(S)\otimes \fMat(\ZZ^\dd\setminus S)$,
    and by the definition of the fermionic tensor product,
    $\FermiCommutator{a \otimes \one}{\one \otimes b} = 0$
    for any operators $a$, $b$ acting on different tensor factors.
    This implies that $\FermiCommutator{x}{y} = 0$
    if $\Supp(x)\subseteq S$ and $y\in \fMat(\ZZ^\dd\setminus S)$.
    
    For the converse, assume that $\FermiCommutator{x}{y} = 0$
    for all $y\in \fMat(\ZZ^\dd\setminus S)$.
    Decompose $x = \sum_i a_i\otimes b_i$,
    where $a_i\in \fMat(S)$ and $b_i\in \fMat(\ZZ^\dd\setminus S)$,
    the $a_i$ are linearly independent,
    and all $a_i$'s and $b_i$'s are homogeneous.
    Then for homogeneous $y$,
    applying the formula~\eqref{eq:graded_derivation} gives
    \begin{equation}
        \FermiCommutator{x}{y}
        = \sum_i\FermiCommutator{a_i\otimes b_i}{\one\otimes y}
        = \sum_i a_i \otimes \FermiCommutator{b_i}{y} = 0.
    \end{equation}
    Since the $a_i$ are linearly independent,
    this forces $\FermiCommutator{b_i}{y} = 0$
    for all $y\in \fMat(\ZZ^\dd\setminus S)$.
    Now, as $\fMat(\ZZ^\dd\setminus S)$ is central,
    we must have $b_i\in \CC\one$ for all $i$.
    Thus $x$ is supported on $S$.
\end{proof}

\begin{definition}
    Let $\calA$ be a fermionic subalgebra of $\fMat(\ZZ^\dd)$
    whose fermionic commutant within $\fMat(\ZZ^\dd)$ is $\calB$.
    We say that $\calA$ is \emph{invertible} if there exists a constant
    spread $\ell > 0$ such that every $x\in \fMat(\ZZ^d)$
    can be written as a finite sum
    \begin{equation}
        x = \sum_i a_i b_i,
    \end{equation}
    for some $a_i\in \calA$, $b_i \in \calB$,
    and $\Supp(a_i), \Supp(b_i)\subseteq \Supp(x)^{+\ell}$.
\end{definition}

\begin{definition}
    A fermionic subalgebra $\calA$ of $\fMat(\ZZ^\dd)$
    is \emph{visibly central} if there exists a range $\ell > 0$
    such that for any $a\in \calA \setminus \CC\one$
    and any site~$s$ in the support of~$a$,
    there exists $w\in \calA \cap \fMat(\{s\}^{+\ell})$
    such that $\FermiCommutator{a}{w} \neq 0$.
\end{definition}

Clearly, every visibly central algebra is central.
This notion is closely related to
the Haag duality in algebraic approaches to quantum field theory 
because of~\ref{lem:VS_support} below,
and has appeared in~\cite{FreedmanHastings2019QCA}
under the name of ``visibly simple'' algebras.

\begin{lemma}
    \label{lem:invertible_VS}
    Let $\calA \subseteq \fMat(\ZZ^\dd)$ be an invertible
    subalgebra with spread~$\ell$, and let~$\calB$
    be its fermionic commutant within $\fMat(\ZZ^\dd)$.
    Then $\calA$ and $\calB$ are both visibly central of range $\ell$.
\end{lemma}

\begin{proof}
    Take any $a\in \calA\setminus \CC\one$ and $s\in \Supp(a)$.
    Since $\fMat(\{s\})$ is central simple,
    by~\ref{thm:FermiCenterCentrality} we can find $x\in \fMat(\{s\})$
    with $\FermiCommutator{a}{x}\neq 0$.
    Using the invertible property, we can write
    $x = \sum_i a_i b_i$, where $\Supp(a_i) \subseteq \{s\}^{+\ell}$
    and $b_i\in \calB$. Since $\calA$ and $\calB$ are graded,
    we can choose $a_i$ and $b_i$ to be homogeneous.
    Applying the formula~\eqref{eq:graded_derivation} to
    $\FermiCommutator{a}{\sum_i a_i b_i}$, we see that
    $\FermiCommutator{a}{a_i} \neq 0$ for some $a_i$.
    This $a_i$ satisfies the visible centrality.

    We follow the same argument to see that $\calB$ is visibly central.
    Take $b\in \calB \setminus \CC \one$
    and $t\in \Supp(b)$. Then there exists some
    $y\in \fMat(\{t\})$ with $\FermiCommutator{b}{y}\neq 0$.
    Writing $y = \sum_j a_j' b_j'$ with
    $\Supp(b_j') \subseteq \{t\}^{+\ell}$
    and $a_j'$, $b_j'$ homogeneous,
    we see that $\FermiCommutator{b}{b_j'}\neq 0$
    for some $b_j'$.
\end{proof}

\begin{lemma}\label{lem:VS_support}
    Let $\calA$ be a visibly central fermionic subalgebra of $\fMat(\ZZ^\dd)$
    with range $\ell$, and let $S\subseteq \ZZ^\dd$
    be any set of sites.
    Then any element $z\in \calA$ with 
    $\FermiCommutator{z}{a} = 0$ for all $a\in \calA\cap \fMat(S)$
    is supported on $(\ZZ^\dd \setminus S)^{+2\ell}$.
\end{lemma}
That is, a fermionic central element of a subalgebra of~$\calA$ 
can appear only at the boundary.

\begin{proof}
    If $\Supp(z)$ intersects $S\setminus (\ZZ^\dd\setminus S)^{+2\ell}$
    at a site $s\in S$, the visible centrality states that
    there is some $w\in \calA$ supported on $\{s\}^{+\ell}$
    such that $\FermiCommutator{z}{w}\neq 0$.
    However, the $\ell$-neighborhood $\{s\}^{+\ell}$
    is still contained in~$S$, so $w\in \calA \cap \fMat(S)$,
    contradiction.
\end{proof}

\begin{lemma}\label{lem:mul_injective}
    Let $\calA$ and $\calB$ be fermionic subalgebras of $\fMat(\ZZ^\dd)$,
    with $\FermiCommutator{\calA}{\calB} = 0$.
    If $\calA$ is visibly central, then the multiplication map
    \begin{equation}
        \mul: \calA \otimes \calB \ni \sum_{i} a_i\otimes b_i \mapsto
        \sum_i a_i b_i\in \calA \calB
    \end{equation}
    is an isomorphism of fermionic algebras,
    where $\otimes$ is the fermionic tensor product.
\end{lemma}

\begin{proof}
    First we check that $\mul$ is a
    fermionic algebra homomorphism.
    The multiplication map is by definition linear.
    Let $a_i,a_{i'}'\in \calA$, and $b_j, b_{j'}'\in\calB$
    be homogeneous elements, where $i,i',j,j'\in\{0,1\}$
    denote the grading.
    Then by the definition of fermionic tensor product,
    we have 
    $\mul((a_i\otimes b_j)^\dagger)=\mul((-1)^{ij}(a_i^\dagger\otimes b_j^\dagger))=(-1)^{ij}a_i^\dagger b_j^\dagger = b_j^\dagger a_i^\dagger = \mul(a_i\otimes b_j)^\dagger$
    and $(a_i\otimes b_j)(a_{i'}'\otimes b_{j'}')
    = (-1)^{i'j}(a_ia_{i'}'\otimes b_jb_{j'}')$,
    so 
    \begin{equation}
        \mul((a_i\otimes b_j)(a_{i'}'\otimes b_{j'}'))
        = (-1)^{i'j} a_ia_{i'}'b_jb_{j'}'
        = a_ib_j a_{i'}'b_{j'}'
        = \mul(a_i\otimes b_j)\mul(a_{i'}'\otimes b_{j'}').
    \end{equation}
    Hence $\mul$ is a homomorphism.
    
    Since $\mul$ is surjective onto $\calA\calB$
    by definition, we only need to show injectivity.
    Consider finite boxes $S(L) = [-L,L]^\dd \cap \ZZ^\dd$,
    and define $\calA_{2L} = \calA\cap \fMat(S(2L))$
    and $\calB_L = \calB\cap \fMat(S(L))$.
    These finite-dimensional fermionic algebras
    decompose into the direct sum of central
    (and hence simple by~\ref{thm:FiniteDimensionalCentralIsSimple})
    fermionic subalgebras by their canonical projectors:
    if $\FermiCent(\calA_{2L})_0 = \spn_\CC \{\pi_\mu\}$
    and $\FermiCent(\calB_L)_0 = \spn_\CC \{\tau_\nu\}$
    with projectors $\{\pi_\mu\}$ and $\{\tau_\nu\}$,
    then $\calA_{2L} \cong \bigoplus_{\mu} (\pi_\mu \calA_{2L})$
    and $\calB_L \cong \bigoplus_{\nu} (\tau_\nu \calB_L)$.
    Then
    \begin{equation}
        \calA_{2L}\otimes \calB_L \cong \bigoplus_{\mu, \nu}
        (\pi_\mu \calA_{2L}\otimes \tau_\nu \calB_L),
    \end{equation}
    where each $(\pi_\mu \calA_{2L}\otimes \tau_\nu \calB_L)$
    is central simple by \ref{thm:tensor_prod_simple}.
    Then $\mul$ is the sum of maps $\mul_{\mu, \nu,L}$,
    the restrictions of $\mul$ to
    $(\pi_\mu \calA_{2L}\otimes \tau_\nu \calB_L)$,
    and each $\mul_{\mu,\nu,L}$ is injective
    whenever it is nonzero because its domain is simple.
    Since $\calA$ is visibly central,
    $\pi_\mu$ is supported near the boundary of
    $S(2L)$ by \ref{lem:VS_support},
    while $\tau_\nu$ is supported in $S(L)$.
    For $L$ sufficiently large, $\pi_\mu$ and $\tau_\nu$
    have disjoint support, so
    $\mul_{\mu,\nu,L}(\pi_\mu \one
    \otimes \tau_\nu \one) = \pi_\mu\tau_\nu$
    is nonzero.
    Hence $\mul_{\mu,\nu,L}$ is injective for all
    sufficiently large~$L$.

    Now we show that $\mul$ is injective.
    Suppose $c = \sum_i a_i\otimes b_i$
    is in the kernel of $\mul$;
    then $c$ is in the kernel of $\mul$
    restricted to $\calA_{2L}\otimes \calB_L$
    for sufficiently large $L$, so we can write
    \begin{equation}
        \mul(c) = \sum_{\mu, \nu, i}\pi_\mu a_i \tau_\nu b_i = 0.
    \end{equation}
    Since $\pi_\mu \pi_{\mu'} = \delta_{\mu, \mu'} \pi_{\mu'}$
    and $\tau_\nu \tau_{\nu'} = \delta_{\nu, \nu'} \tau_{\nu'}$,
    multiplying by $\pi_{\mu'}\tau_{\nu'}$ yields
    $\sum_i \pi_{\mu'}a_i \tau_{\nu'}b_i = 0$,
    which is equivalent to $\mul_{\mu', \nu'}(c) = 0$.
    As $\mul_{\mu',\nu'}$ is injective,
    this forces $c = 0$.
\end{proof}

\begin{definition}\label{def:BoundedSpread}
    A graded homomorphism~$\phi$ from 
    a fermionic subalgebra~$\calA$ of~$\fMat(\ZZ^\dd, p)$ to a subalgebra of~$\fMat(\ZZ^\dd, p')$
    has \emph{spread~$\ell$}
    if $\Supp (\phi(a)) \subseteq \Supp(a)^{+\ell}$ for all~$a \in \calA$.
    We will say $\phi$ is \emph{bounded-spread} 
    if $\phi$ has spread~$\ell$ for some~$\ell \in [0,\infty)$.
\end{definition}

\begin{corollary}
    \label{cor:invertible_isomorphism}
    Let $\calA$ be a fermionic subalgebra of $\fMat(\ZZ^\dd)$
    with fermionic commutant~$\calB$.
    Then $\calA$ is invertible if and only if the
    multiplication map
    \begin{equation}
        \mul: \calA\otimes \calB\to \fMat(\ZZ^\dd)
    \end{equation}
    has a bounded-spread inverse.
\end{corollary}

\begin{proof}
    By \ref{lem:invertible_VS}, $\calA$ invertible
    implies $\calA$ is visibly central; then \ref{lem:mul_injective}
    gives $\mul$ injective.
    The definition of invertibility is equivalent
    to $\mul$ being surjective.
    On the other hand,
    if $\mul$ has a bounded-spread inverse,
    $\mul^{-1}$ exactly gives the decomposition
    for invertibility.
\end{proof}

\begin{proposition}\label{thm:InvertibleSubalgebraIsLocallyGenerated}
    Any invertible subalgebra $\calA$ of $\fMat(\ZZ^\dd)$
    of spread $\ell$ is central simple and
    $2\ell$-locally generated.
    In fact, any $a \in \calA$ is a sum of products of $2\ell$-local generators of~$\calA$
    that are each supported on~$\Supp(a)^{+\ell}$.
\end{proposition}

\begin{proof}
    First we observe that $\fMat(\ZZ^\dd)$ is central simple.
    If $x\in \FermiCent(\fMat(\ZZ^\dd))$,
    then certainly $x \in \FermiCent(\fMat(\Supp(x)))$, 
    so $x \in \CC\one$.
    If $\calI$ is a nonzero graded two-sided ideal
    of~$\fMat(\ZZ^\dd)$, we take any nonzero 
    homogeneous element $x\in \calI$.
    Choose any finite~$S$ containing $\Supp(x)$,
    and consider the two-sided ideal~$\calX$ of $\fMat(S)$
    generated by~$x$, which is graded as~$x$ is homogeneous.
    Since $\fMat(S)$ is simple,
    we must have $\calX = \fMat(S) \subseteq \calI$.
    This forces $\calI = \fMat(\ZZ^\dd)$,
    as $S$ can be arbitrarily large.

    If $\calA$ is an invertible subalgebra of $\fMat(\ZZ^\dd)$,
    from \ref{lem:invertible_VS} we know that
    $\calA$ is visibly central and hence central.
    Now suppose $\calJ$ is a nonzero graded two-sided
    ideal of~$\calA$.
    Let $\calB$ be the fermionic commutant of $\calA$
    within $\fMat(\ZZ^\dd)$; we claim that
    $\calJ \calB$ is a nonzero graded
    two-sided ideal of $\fMat(\ZZ^\dd)$.
    It is nonzero because it contains $\calJ \one = \calJ$,
    and it is graded since $\calJ$ and $\calB$ are individually graded.
    To see that it is a two-sided ideal,
    we need to prove that for homogeneous
    $y_j\in \calJ$, $z_k\in \calB$, and $x\in \fMat(\ZZ^\dd)$,
    $x y_j z_k$ and $y_j z_k x$ are both in $\calJ \calB$.
    By the invertibility, we can write
    $x = \sum_i a_i b_i$ with each $a_i$ and $b_i$ homogeneous.
    We have
    \begin{align}
        a_i b_i y_j z_k &= \pm (a_i y_j) (b_i z_k) \in \calJ\calB, \\
        y_j z_k a_i b_i &= \pm (y_j a_i) (z_k b_i) \in \calJ\calB, \nonumber
    \end{align}
    where the sign depends on the grading of the elements.
    Thus $\calJ\calB$ is a nonzero graded two-sided ideal of
    $\fMat(\ZZ^\dd)$, which forces $\calJ\calB = \fMat(\ZZ^\dd)$.
    Now since $\calB$ is visibly central by \ref{lem:invertible_VS},
    the result of~\ref{lem:mul_injective} implies that
    $\mul: \calJ \otimes \calB\to \fMat(\ZZ^\dd)$ is an isomorphism.
    But we already have (by \ref{cor:invertible_isomorphism})
    the isomorphism $\mul: \calA\otimes \calB\to \fMat(\ZZ^\dd)$,
    so $\calJ = \calA$ by~\ref{thm:CancellationPropertyForIsomorphisms}.
    Hence $\calA$ is simple.
    
    Now we show that $\calA$ is $2\ell$-locally generated.
    Let $A_s$ be the set of all $a_j$'s appearing in
    the $\calA\calB$-decomposition $x = \sum_j a_jb_j$
    for any single-site operator $x\in \fMat(\ZZ^\dd)$.
    By the definition of invertibility,
    each $a_j$ is supported on the $\ell$-neighborhood
    of a single site, and hence are $2\ell$-local.
    We will show that $\calA$ is generated by~$A_s$.
    Any $a\in \calA$ is a finite sum of finite products
    of single-site operators, and hence can be written
    \begin{equation}
        a = \sum_i \prod_j x_{ij}
    \end{equation}
    for single-site $x_{ij}$. 
    Clearly, these single-site $x_{ij}$ may be chosen to be supported on~$\Supp(a)$.
    Inputting the decomposition
    $x_{ij} = \sum_k a_{ijk} b_{ijk}$ for each
    single-site operator~$x_{ij}$ and expanding,
    we obtain a decomposition $a = \sum_\ell a_\ell'b_\ell'$,
    where $a_\ell'\in \braket{A_s}$ are supported on~$\Supp(a)^{+\ell}$.
    Given this decomposition, we can furthermore
    assume that the $b_\ell'$'s are linearly independent
    while retaining $a_\ell'\in \braket{A_s}$.
    Hence we have $\mul(a\otimes \one) = a =
    \mul(\sum_\ell a_\ell'\otimes b_\ell')$,
    so pulling back by~\ref{cor:invertible_isomorphism}
    gives $a\in \spn_{\CC}\{a_\ell'\} \subset \braket{A_s}$,
    as needed.
\end{proof}

\begin{proposition}\label{thm:BicommutantOfInvertible}
    If $\calA$ is an invertible subalgebra of $\fMat(\ZZ^\dd)$
    with fermionic commutant~$\calB$, then the fermionic
    commutant of~$\calB$ is~$\calA$.
    Therefore, a fermionic subalgebra of $\fMat(\ZZ^\dd)$
    is invertible if and only if its fermionic commutant
    is invertible.
\end{proposition}

\begin{proof}
    Let $\calB'$ be the fermionic commutant of~$\calB$.
    It is clear that $\calA \subseteq \calB'$.
    To show the reverse inclusion, take any $x\in \calB'$.
    By the invertibility, we can write
    $x = \mul(\sum_j a_j\otimes b_j)$,
    with each~$a_j$ and~$b_j$ homogeneous.
    In addition, we assume $a_j$ are linearly independent.

    Then, by the Leibniz rule~\eqref{eq:graded_derivation},
    for any homogeneous $b \in \calB$,
    we have $0 = \FermiCommutator{b}{x} = \sum_j \pm a_j \FermiCommutator{b}{b_j}
    = \mul(\sum_j \pm a_j \otimes \FermiCommutator{b}{b_j})$
    where the sign depends on the grading of the elements.
    By the injectivity of~$\mul$ and the linear independence of~$a_j$,
    we must have $\FermiCommutator{b}{b_j} = 0$ for all~$j$.
    Since $b\in \calB$ is arbitrary and $\calB$ is central,
    this forces $b_j\in \CC\one$,
    so $x = \sum_j a_j b_j \in \calA$.
\end{proof}

\begin{lemma}
    \label{lem:full_mat_tensor_factor}
    Let $\psi: \fMat(\ZZ^\dd, p)\to \fMat(\ZZ^\dd, q)$
    be a bounded-spread embedding of one fermionic
    algebra into another. If a subalgebra
    $\calM\subseteq \fMat(\ZZ^\dd, q)$ contains
    $\psi(\fMat(T,p))$ for some $T\subseteq \ZZ^\dd$,
    then the multiplication map
    \begin{equation}
        \mul: \psi(\fMat(T,p))\otimes
        \Commf(\psi(\fMat(T,p)),\calM)\to \calM
    \end{equation}
    is an isomorphism of fermionic algebras.
\end{lemma}

\begin{proof}
    Injectivity follows from \ref{lem:mul_injective}
    since $\psi(\fMat(T,p))$ is visibly central.
    To prove surjectivity, take any $x\in \calM$.
    Let $R\subseteq T$ be the set of sites
    $s$ such that there exists a single-site operator
    $y\in \fMat(\{s\},p(s))$ with
    $\FermiCommutator{\psi(y)}{x}\neq 0$.
    Let $S = \Supp(x)\cup \Supp(\psi(\fMat(R,p(s))))$;
    since $\psi$ is bounded-spread, both $R$ and $S$ are finite.

    The fermionic algebra $\calM' \coloneq \calM \cap \fMat(S,q)$
    is finite-dimensional semisimple, so it can be written
    as a direct sum $\calM' = \bigoplus_i \pi_i \calM'$,
    where the $\{\pi_i\}$ are its canonical central projectors.
    Since the inclusion $\psi(\fMat(R,p))\hookrightarrow \calM'$
    is unital, $\pi_i(\psi(\fMat(R,p)))\hookrightarrow \pi_i\calM'$
    is a unital inclusion of central fermionic algebras,
    so we can apply \ref{thm:Commutant-of-central-in-central}
    to get 
    \begin{equation}
        \pi_i\calM' \cong \pi_i(\psi(\fMat(R,p)))\otimes
        \FermiCommutant{\pi_i(\psi(\fMat(R,p)))}{\pi_i\calM'}
    \end{equation}
    for each $i$.
    Now since $\pi_i \one_{\psi(\fMat(R,p))}\neq 0$
    and $\psi(\fMat(R,p))$ is simple,
    multiplication by $\pi$ is injective on
    $\psi(\fMat(R,p))$: in other words,
    $\pi_i(\psi(\fMat(R,p)))\cong \psi(\fMat(R,p))$.
    Hence taking the direct sum over $i$ yields
    \begin{equation}
        \begin{aligned}
        \calM'&\cong \psi(\fMat(R,p))\otimes
        \bigoplus_i \FermiCommutant{\pi_i(\psi(\fMat(R,p)))}{\pi_i\calM'} \\
        &\cong \psi(\fMat(R,p))\otimes \FermiCommutant{\psi(\fMat(R,p))}{\calM'}.
        \end{aligned}
    \end{equation}
    By definition, $x\in \calM'$, so we have a decomposition 
    $x = \sum_i a_ib_i$ where $a_i\in \psi(\fMat(R,p))$ and
    $b_i\in \Commf(\psi(\fMat(R,p)),\calM')$
    are homogeneous and the $\{a_i\}$ are linearly independent.

    We claim that this is the desired decomposition.
    We already know that each $b_i$ commutes with
    $\psi(\fMat(R,p))$; we need to show that
    $\FermiCommutator{b_i}{\psi(\fMat(T\setminus R,p))} = 0$.
    To that end, take any homogeneous
    $z\in \fMat(T\setminus R,p)$;
    by the definition of $R$, $\FermiCommutator{x}{z} = 0$.
    Then using the linear independence of the $a_i$'s, 
    expanding each $\FermiCommutator{a_ib_i}{z}$
    by the Leibniz rule \ref{eq:graded_derivation} 
    implies that $\FermiCommutator{b_i}{z} = 0$ for all $i$.
    So each $b_i$ is in
    $\FermiCommutant{\psi(\fMat(T,p))}{\calM}$,
    as needed.
\end{proof}

\begin{corollary}
    \label{cor:ISAInverseSameSupport}
    Let $\calA$ be an invertible subalgebra of $\fMat(\ZZ^\dd)$
    supported only on a region $R\subseteq \ZZ^\dd$.
    Then $\calA$ has an inverse on $\Mat(R)$:
    there exists a subalgebra
    $\calB_R\subseteq \fMat(R)$ such that
    $\calA\otimes \calB_R \cong \fMat(R)$.
\end{corollary}

\begin{proof}
    Let $\calB = \FermiCommutant{\calA}{\fMat(\ZZ^\dd)}$;
    then by \ref{cor:invertible_isomorphism},
    $\calA\otimes \calB \cong \fMat(\ZZ^\dd)$.
    Since $\calA$ is supported on $R$,
    $\calB$ contains $\fMat(\ZZ^\dd\setminus R)$,
    so \ref{lem:full_mat_tensor_factor} gives
    \begin{equation}
        \fMat(\ZZ^\dd) \cong \calA\otimes \calB
        \cong \calA \otimes (\calB\cap \fMat(R))
        \otimes \fMat(\ZZ^\dd\setminus R).
    \end{equation}
    Taking the fermionic commutant of
    $\fMat(\ZZ^\dd\setminus R)$ on both sides,
    we see that we can take $\calB_R = \calB\cap\fMat(R)$.
\end{proof}

The following will be used importantly later,
but we state it here since it holds for all~$\dd$.

\begin{lemma}\label{thm:LocallyFactorizable}
    Let $\calA$ be a fermionic invertible subalgebra of~$\fMat(\ZZ^\dd)$ with spread~$\ell$.
    If $x \in \calA$ is supported on $L \cup R$ where $L$ and $R$ are separated by distance
    greater than~$2\ell$,
    then for any decomposition $x = \sum_i x^L_i x^R_i$
    where $\{x^L_i\}_i$ and $\{x^R_i\}_i$ are each linearly independent
    and $\Supp(x^L_i) \subseteq L$ and $\Supp(x^R_i) \subseteq R$
    we have $x^L_i, x^R_i \in \calA$ for all~$i$.
\end{lemma}
This is closely related to the ``additivity'' axiom 
in algebraic approaches to quantum field theory.
In~\cite{FreedmanHastings2019QCA} an algebra (in the ungraded setting) 
with such a property is called a ``locally factorizable'' algebra.
In~\cite{HarlowShaoSorceSrivastava2025} this property is called ``disjoint additivity.''
The grading plays no role in this lemma.
\begin{proof}
    Let $\calB$ be the fermionic commutant of~$\calA$.
    By~\ref{thm:BicommutantOfInvertible} we know $\calB$ is invertible with spread~$\ell$.
    From~\ref{thm:InvertibleSubalgebraIsLocallyGenerated} we know $\calB$ is $2\ell$-locally generated.
    Since the fermionic commutant of~$\calB$ is~$\calA$ by~\ref{thm:BicommutantOfInvertible},
    the membership to~$\calA$ of an arbitrary element $x \in \fMat(\ZZ^\dd)$ 
    is determined by the fermionic commutator between~$x$ and the local generators of~$\calB$.
    The support separation condition ensures that every local generator of~$\calB$
    may have a nontrivial commutator with at most one of $x^L_i$ and $x^R_i$.
    Since $x \in \calA$ has zero fermionic commutator with every local generator of~$\calB$,
    we conclude the proof.
\end{proof}

\subsection{Brauer groups}

\begin{definition}
    Two fermionic subalgebras $\calA_1\subseteq \fMat(\ZZ^\dd, p_1)$
    and $A_2\subseteq \fMat(\ZZ^\dd, p_2)$ are \emph{stably equivalent}
    if there exists
    a bounded-spread fermionic algebra isomorphism 
    \begin{equation}
        \phi: \calA_1\otimes \fMat(\ZZ^\dd, q_1)
        \xrightarrow{\;\sim\;}
        \calA_2\otimes \fMat(\ZZ^\dd, q_2)
    \end{equation}
    for some local assignments~$q_1,q_2$.
    The stable equivalence class of an invertible fermionic
    subalgebra~$\calA$, among all invertible fermionic
    subalgebras in~$\dd$ dimensions, is called
    the \emph{Brauer class} of~$\calA$ and denoted by $[\calA]$.
    A fermionic invertible subalgebra is \emph{Brauer trivial}
    if it is stably equivalent to $\CC = \fMat(\ZZ^d, p)$
    where $\fMat(\{s\}, p(s)) = \CC\one$ on every site~$s$.
\end{definition}

\begin{proposition-definition}
    The set of all Brauer classes of invertible
    fermionic subalgebras on $\ZZ^\dd$ is an
    abelian group under tensor product: we define
    \begin{equation}
        [\calA_1\subseteq \fMat(\ZZ^\dd, p_1)]
        + [\calA_2\subseteq \fMat(\ZZ^\dd, p_2)]\coloneqq
        [\calA_1\otimes \calA_2\subseteq \fMat(\ZZ^\dd, p)],
    \end{equation}
    where the local assignment $p$ is given
    by $\fMat(\{s\}, p(s)) =
    \fMat(\{s\},p_1(s))\otimes \fMat(\{s\}, p_2(s))$. 
    The inverse of $[\calA]$ is represented by
    the fermionic commutant of~$\calA$ within $\fMat(\ZZ^\dd)$
    in which $\calA$ is embedded.
\end{proposition-definition}

\begin{proof}
    First, we observe that for two fermionic algebras
    $\calA$ and $\calB$, there is a natural isomorphism
    $\calA\otimes \calB\cong \calB\otimes \calA$ by
    $a_i\otimes b_j\mapsto (-1)^{ij} b_j\otimes a_i$
    for homogeneous $a_i\in \calA$, $b_j\in \calB$,
    which is bounded-spread.
    Next we show that the tensor product
    $\calA_1\otimes \calA_2$ of two invertible
    subalgebras is invertible.
    Let $\calB_i = \FermiCommutant{\calA_i}
    {\fMat(\ZZ^\dd, p_i)}$
    for $i = 1,2$. By \ref{cor:invertible_isomorphism},
    we have an isomorphism
    $\calA_i\otimes\calB_i \cong \fMat(\ZZ^\dd, p_i)$
    with bounded-spread inverse, so we have
    \begin{equation}
        \begin{aligned}
        (\calA_1 \otimes \calA_2)\otimes (\calB_1\otimes\calB_2)
        & \cong (\calA_1\otimes\calB_1)\otimes(\calA_2\otimes\calB_2) \\
        &\cong \fMat(\ZZ^\dd,p_1)\otimes\fMat(\ZZ^\dd,p_2)
        = \fMat(\ZZ^\dd,p),
        \end{aligned}
    \end{equation}
    where all of the isomorphisms have bounded-spread
    inverses. Hence $\calA_1\otimes \calA_2$
    is invertible with inverse $\calB_1\otimes \calB_2$.

    It also follows that if
    $\calA_1\otimes \fMat(\ZZ^\dd, a_1)
    \cong \calA_1'\otimes \fMat(\ZZ^\dd,a_1')$
    and $\calA_2\otimes \fMat(\ZZ^\dd, a_2)
    \cong \calA_2'\otimes \fMat(\ZZ^\dd,a_2')$,
    then
    \begin{equation}
        \calA_1\otimes \calA_2\otimes \fMat(\ZZ^\dd,a) \cong
        \calA_1'\otimes \calA_2'\otimes \fMat(\ZZ^\dd,a'),
    \end{equation}
    where the local assignments $a$ and $a'$
    are the tensor product of those for $a_1, a_2$
    and $a_1', a_2'$.
    Hence the group operation is well-defined.
    The inverse being the fermionic commutant follows from~\ref{cor:invertible_isomorphism}.
\end{proof}

\subsubsection{Zero dimensions}

Since $\fMat(\ZZ^0,p)$ is always finite dimensional,
every invertible subalgebra in dimension zero is
finite-dimensional central.
Since we allow stabilization,
our Brauer group in dimension zero is therefore zero.
Note that the usual Brauer--Wall group of~$\CC$ is $\ZZ/2\ZZ$ 
(even or odd central simple graded $\CC$-algebras)
since there is no stabilization.

Although the zero dimensional Brauer group is zero trivially,
we do understand all finite-dimensional central simple algebras.
This will be elaborated further in~\S\ref{sec:1dQCAIndex}.

\subsubsection{One dimension}

\newcommand{\inta}[2]{\calA_{[#1,#2]}}
\newcommand{\intg}[2]{\calG_{[#1,#2]}}

Throughout this subsection,
$\calA$ stands for a fermionic invertible subalgebra of~$\fMat(\ZZ)$ with spread~$\ell > 0$.
We also write for any $a < b$
\begin{equation}
    \calA_{[a,b]} = \{ x \in \calA ~|~ \Supp(x) \subseteq [a,b] \} \, .
\end{equation}

\begin{lemma}\label{thm:LocalGenerators}
    For $b - a > 40\ell$, the algebra~$\calA_{[a,b]}$
    is $20\ell$-locally generated.
\end{lemma}

\begin{proof}
    Let $\calC = \calA_{[a,b]}$.
    Observe that $\calA_{[a + 16\ell, b - 16\ell]} \subseteq \calA_{[a + 12 \ell, b - 12\ell]}$,
    and by \ref{lem:VS_support} these two subalgebras have central elements at their boundaries, which are disjoint.
    Therefore, the conditions of~\ref{thm:FindYinbetweenXZ} are met, so
    we can find a central fermionic algebra $\calC_0$ such that 
    $\calA_{[a+16\ell, b-16\ell]}\subseteq \calC_0
    \subseteq \calA_{[a+12\ell,b-12\ell]} \otimes \FermionMinus{1}_s$.
    By~\ref{thm:Bicommutant-of-Central}, $\calC_0$ is a tensor factor: 
    $\calC \otimes \FermionMinus{1}_s \cong \calC_0 \otimes \calM$.

    First we show that $\calC\otimes \FermionMinus{1}_s$
    is $20\ell$-locally generated.
    By~\ref{thm:InvertibleSubalgebraIsLocallyGenerated}, 
    every element of~$\calC_0$ is a sum of products of $2\ell$-local elements from
    $\FermionMinus{1}_s$ and~$\calA_{[a+11\ell, b-11\ell]}$,
    which are members of~$\calC \otimes \FermionMinus{1}_s$.
    Then, it suffices for us to generate
    $\calM = \FermiCommutant{\calC_0}{\calC\otimes \FermionMinus{1}_s}$ by a $20\ell$-local set.
    Every element $w \in \calM$ has zero fermionic commutator with
    $\calA_{[a+16\ell, b - 16 \ell]}$.
    Letting $\FermionMinus{1}_s$ be generated by $\gamma$
    and applying~\ref{lem:VS_support},
    we obtain a decomposition 
    \begin{equation}
        w = \sum_i w^L_i w^R_i \otimes \one + \sum_j w^L_j w^R_j\otimes \gamma,
    \end{equation}
    where $\Supp(w^L_i) \subseteq [a,a+18\ell]$
    and $\Supp(w^R_i) \subseteq [b-18\ell,b]$ for each~$i$.
    We may further assume that sets
    $\{w^L_i\}_i$, $\{w^L_j\}_j$, $\{w^R_i\}_i$,
    and $\{w^R_j\}_j$ are each linearly independent;
    then \ref{thm:LocallyFactorizable} implies that all $w^L_i,w^L_j, w^R_i,w^R_j \in \calA_{[a,b]}$.
    (We do not claim that $w^L_i \in \calM$.)
    Declaring $\gamma$ and all such $w^L_{i,j}$ and $w^R_{i,j}$ as generators,
    we conclude that $\calC \otimes \FermionMinus{1}_s$ is $20\ell$-locally generated.

    To finish, let $g = \sum_i a_i \otimes b_i \in \calC \otimes \FermionMinus{1}_s$
    be a $20\ell$-local generator of~$\calC \otimes \FermionMinus{1}_s$.
    Then, we declare all such elements $a_i$ as generators for~$\calC$,
    each of which is $20\ell$-local.
    It follows that $\calC$ itself is $20\ell$-locally generated.
\end{proof}

\begin{lemma}\label{lem:IntervalAlgebraCenterHasTensorDecomposition}
    For any $a,b$,
    let 
    \begin{equation}
        \begin{aligned}
        \calZ^L_a &= \FermiCommutant{\calA_{[a,a+25\ell]}}{\calA_{[a,a+2\ell]}} \, , \\
        \calZ^R_b &= \FermiCommutant{\calA_{[b-25\ell,b]}}{\calA_{[b-2\ell,b]}}\, . 
        \end{aligned}\label{eq:LeftRightCenters}
    \end{equation}
    Then, for $b - a > 40\ell$,
    \begin{equation}
        \FermiCent(\calA_{[a,b]}) = \calZ^L_a \calZ^R_b \, .
    \end{equation}    
\end{lemma}
\begin{proof}
    By~\ref{lem:VS_support} the fermionic center~$\calC = \FermiCent(\calA_{[a,b]})$ 
    must be supported 
    on $L = [a,a+2\ell]$ union $R = [b-2\ell,b]$,
    and by~\ref{thm:Wedderburn-Fermionic} is spanned by canonical projectors~$\pi_\mu$,
    each of which is even.

    Let $\calD^L = (\calD^L)^\dagger = \calC \cap \calA_{[a,a+2\ell]}  \subseteq \calZ_a^L$ and
    $\calD^R = (\calD^R)^\dagger = \calC \cap \calA_{[b-2\ell,b]}  \subseteq \calZ_b^R$.
    Since $\calA_{[a,b]}$ is $20\ell$-locally generated by~\ref{thm:LocalGenerators}, 
    we know $\calZ_a^L \subseteq \calD^L$ and $\calZ_b^R \subseteq \calD^R$.
    So, $\calZ_a^L = \calD^L$ and $\calZ_b^R = \calD^R$.
    
    By definition, $\calD^L \calD^R \subseteq \calC = \FermiCent(\calA_{[a,b]})$.
    Since $L$ and $R$ are disjoint, 
    we have a Schmidt decomposition
    $\pi_\mu = \sum_i \pi^L_{\mu,i} \pi^R_{\mu,i}$ for any
    $\pi_\mu\in \calC$,
    where each Schmidt component must fermionically commute with every $20\ell$-local generator
    of $\calA_{[a,b]}$.
    It follows that each Schmidt component $\pi^L_{\mu,i}, \pi^R_{\mu,i}$ 
    is in the fermionic center.
    That is, $\calC \subseteq \calZ_a^L \calZ_b^R$,
    so $\calC = \calD^L \calD^R$.
\end{proof}

Using these preliminary lemmas, 
the remainder of the section will
be dedicated to proving:

\begin{theorem}
    \label{thm:1DBrauerGroupTrivial}
    For any invertible subalgebra~$\calA$ of~$\fMat(\ZZ,p)$ of spread~$\ell$,
    the stabilized algebra $\calA \otimes \fMat(\ZZ, s \mapsto \FermionMinus{1})$
    is bounded-spread isomorphic to~$\fMat(\ZZ,q)$
    for some local assignment ~$q$.
    The spread of the isomorphism can be taken to be
    at most~$50\ell$.
\end{theorem}

\begin{corollary}
    The Brauer group of invertible subalgebras
    of $\fMat(\ZZ)$ is zero.
\end{corollary}

Let $\calA$ be a \emph{stabilized} fermionic invertible subalgebra of
$\fMat(\ZZ)$ with spread $\ell$.
Here, by stabilization we mean that $\calA$ 
is a tensor product of a fermionic invertible subalgebra of~$\fMat(\ZZ,p)$
and an auxiliary algebra~$\fMat(\ZZ, s \mapsto \FermionMinus{1})$
that hosts one Majorana algebra~$\FermionMinus{1}$ on every site~$s$.
This auxiliary algebra will play an important role at an early step of the proof
and then stop playing any further role.
If $\calA$ was ungraded (i.e., $\calA$ is a subalgebra of~$\Mat(\ZZ)$),
then the stabilization is \emph{not} necessary.

We need to find a bounded-spread fermionic algebra
isomorphism from $\calA$ to a tensor product
of simple algebras.
To this end, the first step is to
construct an infinite sequence of
simple, fermionically commuting subalgebras of~$\calA$
starting from a source projector.
This is almost the same argument as the one in~\cite{FreedmanHastings2019QCA}
to show the triviality of qudit (bosonic) invertible subalgebras on a circle.
Our exposition differs in that 
we are equipped with the factorization of the center of
interval algebras~$\inta{a}{b}$ (\ref{lem:IntervalAlgebraCenterHasTensorDecomposition})
and the local generation property (\ref{thm:LocalGenerators}),
before we have the full structure theorem of $1\dd$ invertible subalgebras.

Working near the origin,
let $\pi$ be a minimal projector of $\calZ_0^L$
in \eqref{eq:LeftRightCenters}.
This $\pi$ will be fixed for the remainder of this subsection.
The introduction of $\pi$ allows us to
control the elements in the center of $\pi\inta{0}{s}$
on the left side of the interval:

\begin{lemma}
    For $s > 40\ell$, the fermionic center of
    $\pi\inta{0}{s}$ is $\pi \calZ_s^R$.
\end{lemma}

\begin{proof}
    Let $\pi z\in \FermiCent(\pi\inta{0}{s})$.
    For any $x\in \inta{0}{s}$, we have
    $\pi x\in \pi\inta{0}{s}$ and
    \begin{equation}
        0 = \FermiCommutator{\pi z}{\pi x}
        = \pi^2 \FermiCommutator{z}{x}
        = \pi \FermiCommutator{z}{x}
        = \FermiCommutator{\pi z}{x},
    \end{equation}
    so $\pi z$ fermionically commutes with $\inta{0}{s}$.
    Hence $\pi z\in \FermiCent(\inta{0}{s})=\calZ_0^L\calZ_s^R$
    by~\ref{lem:IntervalAlgebraCenterHasTensorDecomposition}.
    But since $\pi$ is a minimal projector in $\calZ_0^L$,
    $\pi z = \pi^2 z\in \pi \calZ_0^L\calZ_s^R
    = \pi \calZ_s^R$.
\end{proof}

We have a unital inclusion $\pi \inta{0}{s} \hookrightarrow \pi \inta{0}{s+3\ell}$
where the multiplicative identity is~$\pi$.
Using this inclusion, we wish to apply~\ref{thm:FindYinbetweenXZ} 
and find a central fermionic algebra.
The auxiliary algebra plays an important role here.
Since the auxiliary algebra is a fermionic tensor product of central algebras,
our minimial projector~$\pi$ or any element of~$\calZ^R_s$ has no support on the auxiliary algebra.
If we denote by~$\calA^0$ the invertible subalgebra before stabilization,
then we have $\inta{a}{b} = \calA^0_{[a,b]} \otimes \bigotimes_{s \in [a,b]} \FermionMinus{1}_s$.
The unital inclusion we are interested in factors as
\begin{align}
    \pi \inta{0}{s} &\hookrightarrow  
    \left(\pi \inta{0}{s+3\ell}^0 \otimes \bigotimes_{j=0}^{s+3\ell - 1} \FermionMinus{1}_j \right) \\
    &\hookrightarrow 
    \left(\pi \inta{0}{s+3\ell}^0 \otimes \bigotimes_{j=0}^{s+3\ell - 1} \FermionMinus{1}_j \right) \otimes \FermionMinus{1}_{s+3\ell} = \pi \inta{0}{s+3\ell}\, . \nonumber
\end{align}
Here, the first inclusion satisfies the trace condition of~\ref{thm:FindYinbetweenXZ}:
if $\pi_\mu \in \calZ^R_s$ and
$\tau_\nu \in \calZ^R_{s+3\ell}$ are minimal projectors,
then the canonical projectors of
$\pi \inta{0}{s}$ are $\pi \pi_\mu$ 
and those of $\pi \inta{0}{s+3\ell}$ are $\pi \tau_\nu$;
with the unique $\tr$ on $\fMat([0,s+3\ell])$,
we have 
$\tr( \pi \pi_\mu \cdot \pi \tau_\nu ) \tr(\pi) 
= \tr(\pi)^2 \tr(\pi_\mu) \tr(\tau_\nu)
= \tr(\pi \pi_\mu) \tr(\pi \tau_\nu)$
since $\pi_\mu$ and $\tau_\nu$ have disjoint support.
Therefore, we find a central fermionic algebra $\calC_s^\pi$
with unital inclusions
\begin{equation}
    \pi \inta{0}{s} \hookrightarrow \calC_s^\pi \hookrightarrow \pi \inta{0}{s+3\ell}.
\end{equation}
There may be more than one choice of~$\calC^\pi_s$ but we fix one for each~$s$.

For $b-40\ell > a > 40\ell$, we define
\begin{equation}
    \calG_{[a,b]}^\pi = \FermiCommutant{\calC_a^\pi}{\calC_b^\pi},
\end{equation}
which is central by \ref{thm:Commutant-of-central-in-central}
and satisfies $\calC^\pi_b \cong \calC_a^\pi \otimes
\intg{a}{b}^\pi$.
Furthermore, we define
\begin{equation}
    \label{eq:Gab_def}
    \calG_{[a,b]} = \{x^R \in\inta{a-3\ell}{b+3\ell} \mid 
    \pi x^R \in \calG^\pi_{[a,b]}\} .
\end{equation}

\begin{lemma}
    \label{lem:ConstructCentralG}
    We have  $\inta{a+5\ell}{b-5\ell}\subseteq \calG_{[a,b]}
    \subseteq \inta{a-3\ell}{b+3\ell}$
    and $\calG^\pi_{[a,b]} = \pi \calG_{[a,b]}$.
    Moreover, $\calG_{[a,b]}$ is central.
\end{lemma}

\begin{proof}
    By definition $\calG_{[a,b]}\subseteq \inta{a-3\ell}{b+3\ell}$. For the left inclusion,
    take any $x\in \inta{a+5\ell}{b-5\ell}$.
    Then $x$ has disjoint support from $\calC^\pi_a\subseteq \pi\inta{0}{a+3\ell}$,
    so $\pi x$ fermionically commutes with $\calC_a^{\pi}$
    and hence $\pi x \in \calG^{\pi}_{[a,b]}$.
    This proves that
    $\inta{a+5\ell}{b-5\ell}\subseteq \calG_{[a,b]}$.

    Next we show that $\pi\calG_{[a,b]} = \calG^\pi_{[a,b]}$.
    It is clear that $\pi \calG_{[a,b]} \subseteq \calG^\pi_{[a,b]}$. We need to show that
    every $\pi x \in \calG^\pi_{[a,b]} = \FermiCommutant{\calC_a^\pi}{\calC_b^\pi}$
    can be written as $\pi x = \pi x^R$
    for some $x^R$ supported on $[a-3\ell,b+3\ell]$.
    For any $y \in \inta{0}{a}$,
    we have $\pi y\in \pi \inta{0}{a}
    \subseteq \calC_a^\pi$ and
    \begin{equation}
        0 = \FermiCommutator{\pi x}{\pi y} =
        \pi^2 \FermiCommutator{x}{y} = \pi \FermiCommutator{x}{y}
        = \FermiCommutator{\pi x}{y},
    \end{equation}
    so $\pi x$ fermionically commutes with $\inta{0}{a}$.
    Now by \ref{lem:VS_support}
    the visible centrality of $\calA$ implies that
    $\pi x$ is supported on $L = [0,2\ell]$ union
    $R = [a-2\ell, b+3\ell]$, so we have a Schmidt
    decomposition $\pi x = \sum_i x_i^L x_i^R$
    with $x_i^L, x_i^R\in \calA$ by \ref{thm:LocallyFactorizable}.
    By \ref{thm:LocalGenerators}, $\inta{0}{a}$
    is $20\ell$-locally generated; each generator $g$
    cannot overlap with both $L$ and $R$,
    forcing $\FermiCommutator{g}{x_i^L} = 0$ for each $i$.
    Hence $x_i^L\in \FermiCent(\inta{0}{a})$.
    Multiplying by $\pi$, we have
    $\pi x = \pi^2x = \sum_i (\pi x_i^L)x_i^R$,
    where $\pi x_i^L \in \FermiCent(\pi\inta{0}{a}) =
    \pi \calZ_a^R$.
    As $\pi x_i^L$ is supported on the left side,
    we must have $\pi x_i^L\in \CC\pi$.
    Thus, for every $\pi x\in \calG^\pi_{[a,b]}$,
    there is some $x^R$ such that $\pi x = \pi x^R$
    and $x^R$ is supported on $R$.

    Finally, we prove that $\calG_{[a,b]}$ is central.
    Suppose $z\in \FermiCent(\calG_{[a,b]})$;
    then $\pi z\in \FermiCent(\calG^\pi_{[a,b]}) = \CC\pi$.
    Now since $z$ is supported on $[a-3\ell,b+3\ell]$,
    this forces $z\in \CC\one$.
\end{proof}

\begin{lemma}
    \label{lem:ProductOfGIntervals}
    For $c - 80\ell > b - 40\ell > a > 40\ell$,
    we have $\calG_{[a,c]} = \calG_{[a,b]}\calG_{[b,c]}$.
    Furthermore, $\intg{a}{b}$ and $\intg{b}{c}$
    fermionically commute.
\end{lemma}

\begin{proof}
    From \eqref{eq:Gab_def} and the definition
    of $\calG^\pi_{[a,b]}$, we immediately have
    $\intg{a}{b}\intg{b}{c}\subseteq \intg{a}{c}$.
    For the reverse inclusion, take
    a homogeneous element $x\in \intg{a}{c}$.
    Then $\pi x \in \calC_c^\pi \cong \calC_b^\pi \otimes \intg{b}{c}^\pi$ has a Schmidt decomposition
    \begin{equation}
        \pi x = \sum_i (\pi y_i)(\pi z_i),
    \end{equation}
    where $\pi y_i \in \calC_b^\pi$ and
    $\pi z_i \in \intg{b}{c}^\pi$ are homogeneous.
    Now since $\pi x$ fermionically commutes with
    $\calC_a^\pi$, the Leibniz formula~\eqref{eq:graded_derivation}
    implies that each $\pi y_i$ fermionically commutes with~$\calC_a^\pi$.
    Hence $\pi y_i\in \intg{a}{b}^\pi$.
    Finally, by \ref{lem:ConstructCentralG},
    there exist $y_i^R\in \intg{a}{b}$ and $z_i^R\in\intg{b}{c}$
    such that $\pi y_i^R = \pi y_i$ and $\pi z_i^R = \pi z_i$.
    We have $\pi x = \sum_i \pi y_i^R z_i^R$ where 
    $\pi$ has disjoint support from $x,y_i^R,z_i^R$.
    Therefore, $x = \sum_i y_i^R z_i^R$, as needed.

    To see that $\intg{a}{b}$ and $\intg{b}{c}$
    fermionically commute, take $x\in \intg{a}{b}$
    and $y\in \intg{b}{c}$. By definition,
    $\pi y$ fermionically commutes with $\calC_b^\pi \ni \pi x$,
    and so $0 = \FermiCommutator{\pi x}{\pi y} = \pi \FermiCommutator{x}{y}$.
    Since $\FermiCommutator{x}{y}$ has disjoint support
    from $\pi$, this implies $\FermiCommutator{x}{y} = 0$.
\end{proof}

These lemmas are summarized as the following.
\begin{proposition}\label{prop:IntervalContainsLocallyGeneratedSimpleAlgebra}
    Let $\calA$ be an invertible subalgebra of~$\fMat(\ZZ)$ of spread~$\ell$,
    stabilized by one-Majorana algebra if $\calA$ has nontrivial grading.
    For any interval~$[a,b]$ with $b - a > 80\ell$,
    there exists a simple subalgebra $\calG$ of~$\calA$ such that
    \begin{enumerate}
        \item $\calG$ is $50\ell$-locally generated and
        \item $\calA_{[a+5\ell, b-5\ell]} \subseteq \calG \subseteq \calA_{[a-3\ell,b+3\ell]}$.
    \end{enumerate}
\end{proposition}

\begin{proof}
    We take $\calG_{[a,b]}$ from~\ref{lem:ConstructCentralG},
    which shows that $\calG$ is finite-dimensional central 
    and hence simple by~\ref{thm:FiniteDimensionalCentralIsSimple},
    and also that $\calG$ is sandwiched by interval algebras.
    From~\ref{lem:ProductOfGIntervals},
    we see that $\calG_{[a,b]}$ is a product of two or more subalgebras,
    each of which is contained in an interval of length at most~$50\ell$.
    This shows the local generation claim.
\end{proof}

\begin{proof}[Proof of Theorem~\ref{thm:1DBrauerGroupTrivial}]
    Applying ~\ref{prop:IntervalContainsLocallyGeneratedSimpleAlgebra},
    we can find $50\ell$-locally generated simple 
    subalgebras $\calG_j = \calG_{[100j\ell+10\ell,100(j+1)\ell-10\ell]}$
    sandwiched between $\calA_{[100j\ell+15\ell], 100 (j+1)\ell - 15\ell]}$
    and
    $\calA_{[100j\ell + 7\ell,100(j+1)\ell - 7\ell]}$
    for any~$j \in \ZZ$.
    Let $\calG = \prod_j \calG_j$.
    Let $\calM = \Commf(\calG, \calA)$ be the commutant within~$\calA$.
    By~\ref{lem:VS_support}, the commutant~$\calM$ is supported on the union of
    separated intervals $[100j\ell - 20\ell, 100j\ell + 20\ell]$.
    By~\ref{thm:LocallyFactorizable}, every Schmidt component of any element of~$\calM$
    on one of the separated intervals is an element of~$\calA$,
    which is in fact an element of~$\calM$ because $\calG$ is locally generated.
    This means that $\calM$ is a product of fermionically commuting subalgebras~$\calM_j$
    on $[100j\ell - 20\ell, 100j\ell + 20\ell]$.
    If $z$ is a central element of~$\calM_j$,
    then it must be central in~$\calA$, implying that $\calM_j$ is central (and hence simple).

    Define a local operator algebra $\fMat(\ZZ,r)$
    by assigning $\fMat(\{100j\ell + 50\ell\},r(100j\ell + 50\ell)) \cong \calG_j$
    and assigning the trivial algebra~$\CC$ elsewhere.
    By~\ref{lem:full_mat_tensor_factor},
    we see that $\calM$ and $\calG$ generate all of~$\calA$.
    We conclude that $\calA$ is bounded-spread isomorphic to $\fMat(\ZZ,r) \otimes \fMat(\ZZ,q)$
    where $q$ assigns $\calM_j$ for $100j\ell \in \ZZ$ and $\CC$ elsewhere,
    which is the theorem.
\end{proof}

\subsection{Torus-trick for invertible subalgebras}
\label{subsec:TorusTrick}

In this subsection, we explore the generalizations
beyond $\ZZ^\dd$ to 
lattices over more general metric spaces,
which is the setting discussed in~\cite{FHH2019}.
We use the ``torus trick''~\cite{hastings2013classifying}
for invertible subalgebras to prove two results.
First, we show the Brauer triviality of invertible
subalgebras on one-dimensional simplicial complexes.
Then, we apply this result to show that
translation-invariant Brauer trivial invertible subalgebras
are trivial with uniformly bounded stabilization.

\begin{definition}\label{def:Lattice}
    A \emph{lattice} $\Lambda \subseteq X$ on a metric space~$X$
    is a locally finite subset,
    \emph{i.e.},
    every ball of~$X$ of finite radius
    contains only finitely many points of~$\Lambda$.
\end{definition}

To use the torus trick
for invertible subalgebras,
we first require a  ``filling-hole'' lemma:

\begin{lemma}\label{lem:fillingHoleISA}
    Let $\calA^\circ, \calB^\circ$ be fermionically commuting fermionic subalgebras of~$\fMat(\Lambda)$
    on a \emph{finite} lattice~$\Lambda$
    with the multiplication map $\mul : \calA^\circ \otimes \calB^\circ \to \fMat(\Lambda)$.
    Let $P\subseteq \Lambda$ be a ball of radius~$\ell$.
    Suppose that $\mul$ has a spread-$\ell$ partial inverse:
    \begin{equation}
        \fMat(\Lambda \setminus P) \to \calA^\circ \otimes \calB^\circ
    \end{equation}
   Then, there exists an invertible subalgebra~$\calA \subseteq \fMat(\Lambda) \otimes \FermionMinus{1}_{P}$ 
   of spread at most~$10\ell$
   such that: (i) $\calA$ contains~$\calA^\circ \cap \fMat(\Lambda \setminus P^{+10\ell})$, and
   (ii) $B^\circ$ fermionically commutes with~$\calA$. 
   Here, $\FermionMinus{1}_P$ is supported at the center site in~$P$.
\end{lemma}

\begin{proof}
    We first note that the visible centrality of~$\calA^\circ$ continues to hold away from~$P$:
    if $x \in \calA^\circ$ is supported outside~$P$,
    then a single-site operator that does not fermionically commute with~$x$
    has a spread-$\ell$ preimage under~$\mul$, and hence provides an element of~$\calA^\circ$
    that does not fermionically commute with~$x$.
    Then, for any region~$R \subseteq P^c$ (where \(P^c = \Lambda \setminus P\)),
    the fermionic center of~$\calA^\circ \cap \fMat(R)$
    must be supported near the boundary of~$R$.

    Consider two regions~$R = (P^{+4\ell})^c \subseteq R' = (P^{+\ell})^c$.
    The fermionic centers of finite-dimensional fermionic algebras 
    $\calA^\circ \cap \fMat(R)$ and $\calA^\circ \cap \fMat(R')$
    have disjoint support, and therefore we can apply~\ref{thm:FindYinbetweenXZ}
    with the stabilization~$\FermionMinus{1}_P$
    to obtain a simple fermionic algebra~$\calA$ with
    $\calA^\circ \cap \fMat(R)
    \subseteq \calA \subseteq
    (\calA^\circ \cap \fMat(R')) \otimes \FermionMinus{1}$.
    Note that $\FermiCommutator{\calA}{\calB^\circ} = 0$.

    It suffices to show that $\calA$ is invertible.
    Let $\calB\supseteq \calB^\circ$
    be the fermionic commutant of~$\calA$
    within $\fMat(\Lambda') = \fMat(\Lambda) \otimes \FermionMinus{1}_P$.
    Since $\Lambda$ is finite,
    we have $\calA \otimes \calB \cong \fMat(\Lambda')$;
    we need to check that this isomorphism has a
    spread-$10\ell$ inverse.
    Every element~$z \in \fMat(\Lambda')$ supported outside~$P^{+5\ell}$
    enjoys the same decomposition from the partial
    inverse of $\mul$
    using elements of~$\calA^\circ$ and~$\calB^\circ$.
    Every element~$y \in \fMat(\Lambda')$ supported inside~$P^{+6\ell}$
    decomposes according to $\calA \otimes \calB \cong \fMat(\Lambda')$,
    but all Schmidt $\calA$-components commute 
    with any element of~$\calA^\circ \cap \fMat((P^{+7\ell})^c)$
    so they are supported on~$P^{+9\ell}$.
    Since $y$ commutes with every single-site operator of~$\fMat((P^{+9\ell})^c)$,
    all Schmidt $\calB$-components are also supported on~$P^{+9\ell}$.
\end{proof}

\begin{definition}
    Consider a fermionic algebra $\fMat(\Lambda_X, p_X)$
    for a lattice $\Lambda_X$ on a metric space~$X$.
    Given another metric space~$Y$ and a map
    $\iota: Y\to X$, one can construct \emph{pull-back lattice}
    $\Lambda_Y = \iota^{-1}(\Lambda_X)$
    and the \emph{pull-back algebra} $\fMat(\Lambda_Y,p_Y)$
    defined by $p_Y(y) = p_X(\iota(y))$.
\end{definition}

We now dicuss the setup for the torus trick.
Let $\calA$ be an invertible subalgebra
of $\fMat(\Lambda_X)$ with spread $\ell$
for a lattice~$\Lambda_X$ on a metric space~$X$.
Let $B_{\dd}(L)$ denote the $\dd$-dimensional Euclidean
box of linear dimension $L$,
and let $T_{\dd}(L)$ denote the $\dd$-torus
$\RR^\dd/(L\ZZ)^\dd$.
For any $L\ge 1000\ell$,
there is an embedding
$\iota_T: B_{\dd}(L) \hookrightarrow T_{\dd}(L + 100\ell)$
such that $\dist(\iota(y),\iota(y')) = \dist(y,y')$
whenever $\dist(y,y')\le 100\ell$.
Given a fermionic algebra $\fMat(\Lambda_T)$
on the $\dd$-torus, $\iota_T$ yields a pull-back
algebra $\fMat(\Lambda_{B_\dd(L),T})$.

Let $L' = L + 100\ell$ if $\dd = 1$
or $L' = 7L$ if $\dd \ge 2$.
Suppose there is a subset $E_X\subseteq X$ isometric
to $B_\dd(L')$, and
that there is a metric-space isomorphism
$\iota_E: B_\dd(L)\to E_X^{-(L'-L)/2}$.\footnote{
    This isomorphism only exists if the dimensionality $\dd$
    is ``correct." If $\dd = 1$ but $X = \RR^2$,
    then $E_X^{-r} = \emptyset$ for any $-r < 0$.
}
Then $\iota_E$ yields a pull-back algebra
$\fMat(\Lambda_{B_\dd(L),E})$.
By our choice of $L'$, there exists an immersion
$\iota$ of the punctured $\dd$-torus $T_\dd^p$ in
$E_X$ such that the following diagram commutes:
\begin{equation}
    \begin{tikzcd}
        & E_X^{-(L'-L)/2} \arrow[rd, hook] & &  \\
        B_\dd(L) \arrow[ru, "\iota_E"] \arrow[rd, "\iota_T"'] & & E_X\cong B_\dd(L')\arrow[r, hook] & X \\
        & T_\dd^p(L+100\ell) \arrow[r,hook] \arrow[ru, "\iota"] & T_\dd(L+100\ell) & 
    \end{tikzcd}
    \label{eq:torus_setup}
\end{equation}
An example immersion is shown in \autoref{fig:torus_immersion}.
The images of $\iota_T$ and $\iota_E$ are shown in orange.

\begin{figure}
    \centering
    \begingroup
        \resizebox{0.9\textwidth}{!}{\definecolor{sgreen}{RGB}{125,247,125}
\definecolor{smagenta}{RGB}{255,102,255}
\definecolor{sorange}{RGB}{250,190,125}
\definecolor{sgray}{RGB}{135,135,135}

\def\aa{1.2}   
\def\rp{0.38}  
\def\ee{0.82}  

\newcommand{\idsquare}{%
  \draw[twoarrows] (-\aa, \aa) -- ( \aa, \aa);
  \draw[twoarrows] (-\aa,-\aa) -- ( \aa,-\aa);
  \draw[onearrow]  (-\aa,-\aa) -- (-\aa, \aa);
  \draw[onearrow]  ( \aa,-\aa) -- ( \aa, \aa);}

%
\def\gtop{(1.15,1.15)
  -- ( 4.95, 1.15) arc[start angle= 90,end angle=  0,radius=2.5cm]
  -- ( 7.45,-4.25) arc[start angle=  0,end angle=-90,radius=2.5cm]
  -- (-1.55,-6.75) arc[start angle=-90,end angle=-180,radius=2.5cm]
  -- (-4.05,-1.35) arc[start angle=180,end angle= 90,radius=2.5cm]
  -- (-1.15, 1.15)}
\def\gbot{%
  -- (-1.55,-1.15) arc[start angle= 90,end angle=180,radius=0.2cm]
  -- (-1.75,-4.25) arc[start angle=180,end angle=270,radius=0.2cm]
  -- ( 4.95,-4.45) arc[start angle=270,end angle=360,radius=0.2cm]
  -- ( 5.15,-1.35) arc[start angle=  0,end angle= 90,radius=0.2cm]
  -- ( 1.15,-1.15)}
\def\mout{(-1.15,-1.15)
  -- (-1.15,-1.55) arc[start angle=180,end angle=270,radius=2.5cm]
  -- ( 2.25,-4.05) arc[start angle=-90,end angle=  0,radius=2.5cm]
  -- ( 4.75, 1.55) arc[start angle=  0,end angle= 90,radius=2.5cm]
  -- ( 1.35, 4.05) arc[start angle= 90,end angle=180,radius=2.5cm]
  -- (-1.15, 1.15)}
\def\minn{%
  -- ( 1.15, 1.55) arc[start angle=180,end angle= 90,radius=0.2cm]
  -- ( 2.25, 1.75) arc[start angle= 90,end angle=  0,radius=0.2cm]
  -- ( 2.45,-1.55) arc[start angle=  0,end angle=-90,radius=0.2cm]
  -- ( 1.35,-1.75) arc[start angle=-90,end angle=-180,radius=0.2cm]
  -- ( 1.15,-1.15)}
\tikzset{dimline/.style={{Stealth[length=1.4mm,width=1.4mm]}-%
                         {Stealth[length=1.4mm,width=1.4mm]},line width=0.8pt}}

\begin{tikzpicture}[
  line width=1.1pt,
  onearrow/.style={
    decoration={markings,
      mark=at position 0.551 with {\arrow{Stealth[length=3.4mm,width=3.6mm]}}},
    postaction={decorate}},
  twoarrows/.style={
    decoration={markings,
      mark=at position 0.500 with {\arrow{Stealth[length=3.4mm,width=3.6mm]}},
      mark=at position 0.601 with {\arrow{Stealth[length=3.4mm,width=3.6mm]}}},
    postaction={decorate}},
  hatch/.style={pattern=north east lines,pattern color=black}]

\begin{scope}[scale=1.4]
  \idsquare
  \path[hatch] (-\rp,-\rp) rectangle (\rp,\rp);
  \draw        (-\rp,-\rp) rectangle (\rp,\rp);
\end{scope}

\node[scale=2.4] at (2.7,0) {$=$};

\begin{scope}[shift={(5.4,0)},scale=1.4]
  \fill[sorange]  (-\ee,-\ee) rectangle ( \ee, \ee);
  \fill[sgreen]   (-\aa,-\ee) rectangle (-\ee, \ee);
  \fill[sgreen]   ( \ee,-\ee) rectangle ( \aa, \ee);
  \fill[smagenta] (-\ee, \ee) rectangle ( \ee, \aa);
  \fill[smagenta] (-\ee,-\aa) rectangle ( \ee,-\ee);
  \foreach \sx in {-1,1}{\foreach \sy in {-1,1}{%
    \path[hatch] ({\sx*\ee},{\sy*\ee}) rectangle ({\sx*\aa},{\sy*\aa});
    \draw ({\sx*\aa},{\sy*\ee}) -- ({\sx*\ee},{\sy*\ee}) -- ({\sx*\ee},{\sy*\aa});}}
  \idsquare
  \draw[dimline] (-\ee,0) -- (\ee,0)
    node[midway,fill=sorange,inner sep=1pt,scale=1.25] {$L$};
\end{scope}

\node[scale=1.5] at (5.4,2.35) {$T_2^p$};

\draw[-{Stealth[length=3mm,width=3mm]}] (7.6,0) -- (9.4,0)
  node[midway,above=1pt,scale=1.5] {$\iota$};

\begin{scope}[shift={(13.26,0)},scale=0.42,line join=round,line cap=round]
  \draw (-8.05,-8.05) rectangle (8.05,8.05);
  \node[anchor=north west,scale=1.5] at (-7.75,7.75) {$E_X$};
  \fill[sgreen] \gtop -- (-1.15,-1.15) \gbot -- cycle;
  \draw \gtop;
  \draw (-1.15,-1.15) \gbot;
  \fill[smagenta] \mout -- (1.15,1.15) \minn -- cycle;
  \draw \mout;
  \draw (1.15,1.15) \minn;
  \fill[sgray] (2.45,-1.15) rectangle (4.75,1.15);
  \draw[line cap=butt] (2.45,-1.15) rectangle (4.75,1.15);
  \fill[sorange] (-1.16,-1.16) rectangle (1.16,1.16);
  \draw[dimline] (-1.15,0) -- (1.15,0)
    node[midway,fill=sorange,inner sep=1pt,scale=1.25] {$L$};
  \draw[dimline] (-8.05,-8.6) -- (8.05,-8.6)
    node[midway,below=1pt,scale=1.5] {$7L$};
\end{scope}

\end{tikzpicture}}
    \endgroup
    \caption{Immersion of the punctured $2$-torus in $\RR^2$.}
    \label{fig:torus_immersion}
\end{figure}
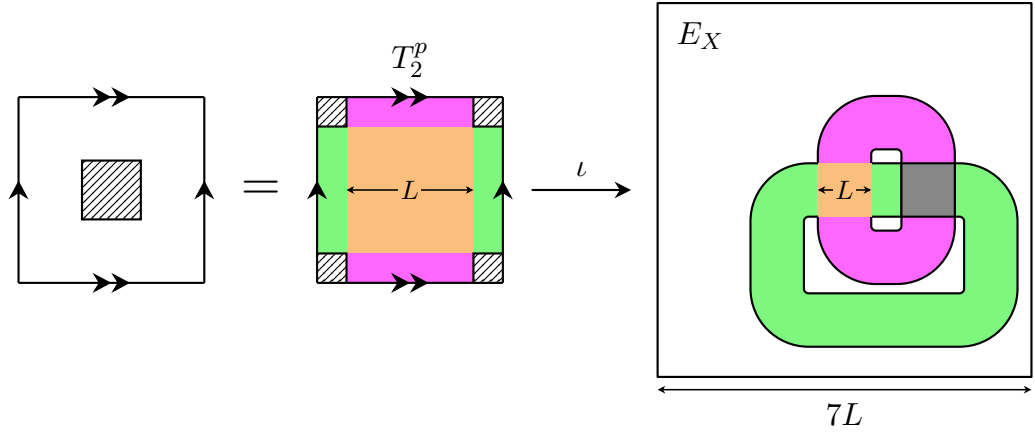

\begin{lemma}\label{lem:TorusTrick}
    Let $\calA$ be an invertible subalgebra with spread
    $\ell$ on $\fMat(\Lambda_X)$ for some metric space~$X$
    equipped with the maps in~\ref{eq:torus_setup}.
    Then there exists an algebra $\fMat(\Lambda_{T_\dd})$
    on $T_\dd(L+100\ell)$ and an invertible subalgebra
    $\calD\subseteq \fMat(\Lambda_{T_\dd})$ such that
    the pull-backs $\fMat(\Lambda_{B_d(L),T})$
    and $\fMat(\Lambda_{B_\dd(L),E})$ are exactly equal,
    and the pull-back of $\calD\cap \fMat(\im \iota_T)$
    via~$\iota_T$ inside $\fMat(\Lambda_{B_d(L),T})$
    exactly equals the pull-back
    of $\calA\cap \fMat(\im \iota_E)$ via~$\iota_E$ inside 
    $\fMat(\Lambda_{B_\dd(L),E})$.
\end{lemma}

\begin{proof}
    Let $Y = T_\dd^p(L+100\ell)$ be the punctured torus.
    Define an algebra $\fMat(\Lambda_Y)$
    on the punctured torus as the pull-back of
    $\fMat(E_X)$ via $\iota$.
    For any single-site operator~$y \in \fMat(\Lambda_Y)$ at~$s \in \Lambda_Y^{-10\ell}$,
    the invertible subalgebra~$\calA$
    gives a local decomposition $x = \sum_i a_i b_i$
    where $x \in \fMat(\{\iota(s)\})$ is a copy of~$y$
    and $a_i \in \calA$ and $b_i$ commutes with~$\calA$
    supported on a ball of radius~$\ell$ centered at~$\iota(s)$.
    By the definition of the pull-back, $a_i$ and $b_i$ have corresponding elements in~$\fMat(\Lambda_Y)$
    supported on a ball of radius~$\ell$ centered at~$s$.
    Here, we use the fact that the immersion $\iota$
    only increases distances locally:
    $\dist(\iota(y),\iota(y'))\ge \dist(y,y')$
    for any points $y$,$y'$ with $\dist(y,y')\le 100\ell$.
    
    Thus, we have constructed fermionic subalgebras~$\calA^\circ_Y$ and~$\calB^\circ_Y$
    of~$\fMat(\Lambda_Y)$ that allow for a partial inverse of the multiplication map
    in the sense of~\ref{lem:fillingHoleISA}.
    Because the diagram \ref{eq:torus_setup} commutes,
    $\calA_Y^\circ\cap \fMat(\im \iota_T)$ exactly
    corresponds to $\calA\cap \fMat(\im \iota_E)$.
    Now heal the puncture of the torus by~\ref{lem:fillingHoleISA}
    to obtain the invertible subalgebra~$\calD$ on
    $\fMat(\Lambda_{T_\dd})\coloneq \fMat(\Lambda_Y)\otimes \FermionMinus{1}_P$.
\end{proof}

\begin{remark}
    \label{rem:UnwrapTorus}
    Given an invertible subalgebra on the torus,
    one can consider its pull-back
    via the universal covering map $\RR^\dd\to T_\dd$
    to obtain a translation-invariant
    invertible subalgebra on~$\RR^\dd$.
    In particular, starting with an invertible subalgebra
    $\calA$ on $\ZZ^\dd$, one can use the torus trick
    to obtain a translation-invariant invertible subalgebra
    on $\ZZ^\dd$ that agrees with $\calA$ on a large region.
\end{remark}

\subsubsection{Invertible subalgebras on one-dimensional
simplicial complexes}

We use the torus trick to extend the
Brauer triviality result for invertible subalgebras
on $\ZZ$ to those on one-dimensional simplicial complexes.

\begin{definition}
    A one-dimensional simplicial complex~$X$ with a metric is \emph{$\ell$-coarse}
    if
    \begin{itemize}
        \item every $0$-cell of~$X$ is attached to finitely many $1$-cells of~$X$,
        \item the distance between any pair of distinct $0$-cells is $>10\ell$, and
        \item if $E$ is a $1$-cell of~$X$, then $E^{-\ell}$
        is separated from all other~$1$-cells by distance~$\ge \ell$.
   \end{itemize}
   We simply say $X$ is coarse if there exists $\ell > 0$ such that $X$ is $\ell$-coarse.
\end{definition}

\begin{example}
    The canonical example of a coarse $1\dd$ simplicial complex is~$\RR$
    where 0-cells are the integral points and 1-cells are the unit length intervals.
    Here, the parameter~$\ell$ can be chosen to be $0.01$.
    The set of all integral points, $\ZZ \subset \RR$, is a lattice on~$\RR$.
\end{example}
    
\begin{example}
    A slightly nontrivial example~$X$ is the union of all vertical and horizontal lines
    that cross points of integer coordinates in~$\RR^2$.
    The $0$-cells are all the points~$(x,y) \in \ZZ^2 \subset \RR^2$,
    and the $1$-cells are line segments of unit length connecting $(x,y)$ and $(x+1,y)$ or $(x,y+1)$
    with $x,y \in\ZZ$.
    A lattice is
    $\Lambda = \{ (x , y) \in X \mid y \in \ZZ, x \in \frac 1 {1000} \ZZ\}
    \cup \{ (x , y) \in X \mid x \in \ZZ, y \in \frac 1 {1000} \ZZ\}$.
\end{example}

This example shows that a lattice on a coarse one-dimensional simplicial complex
does not have to look like a line.
It is more appropriate to think of 
the $1$-skeleton of some high dimensional simplicial (or CW)
complex without any point of too high density of cells,
and to think of a locally finite collection of points on it.

\begin{proposition}\label{thm:ISAon1DSkeleton}
    There exist constants $C > c > 1$ such that the following is true.
    Let $\ell > 1$.
   Let $\Lambda$ be a lattice on a $C\ell$-coarse one-dimensional simplicial complex~$X$.
   Then, any ungraded invertible subalgebra of~$\Mat(\Lambda,p)$ of spread~$\ell$ 
   is isomorphic to $\Mat(\Lambda,q)$ for some~$q$ with spread bounded by~$c\ell$.
   Similarly, for any invertible subalgebra~$\calA$ of~$\fMat(\Lambda,p)$ of spread~$\ell$,
   we have an isomorphism between $\calA \otimes \fMat(\Lambda,s \mapsto \FermionMinus{1})$ 
   and $\fMat(\Lambda,q)$ with spread bounded by~$c\ell$.
\end{proposition}

The argument is essentially the same as in
the proof of the Brauer triviality theorem~\ref{thm:1DBrauerGroupTrivial}
using~\ref{prop:IntervalContainsLocallyGeneratedSimpleAlgebra}.
However, we will use~\ref{lem:TorusTrick} 
for the proof of~\ref{thm:ISAon1DSkeleton}
largely because we do not wish to whitebox the argument for~\ref{thm:1DBrauerGroupTrivial}.

\begin{proof}
    By~\ref{lem:TorusTrick} and \ref{rem:UnwrapTorus},
    we have an invertible subalgebra~$\calD_E$ over~$\ZZ$
    that is bounded-spread isomorphic to $\calA$
    in the middle~$E^\circ$ of any given edge~$E$.
    By the Brauer triviality result~\ref{thm:1DBrauerGroupTrivial},
    we know~$\calD_E$ is trivial,
    and thus we find a locally generated central simple subalgebra~$\calG_{E^\circ}$
    that includes all local generators of~$\calA$ in the middle~$E^\circ$ of each~$E$.
    By~\ref{lem:full_mat_tensor_factor},
    the subalgebra~$\prod_E \calG_{E^{\circ}}$ is a tensor factor of~$\calA$.
    Let $\calM$ be the commutant of~$\prod_E \calG_{E^\circ}$ within~$\calA$.
    Every Schmidt component of an element of~$\calM$
    is an element of~$\calA$ by~\ref{thm:LocallyFactorizable},
    and in turn is an element of~$\calM$ since each $\calG_{E^\circ}$ is locally generated.
    Therefore, $\calM$ factorizes into local algebras supported near the vertices of~$X$.
\end{proof}

\subsubsection{Translation-invariant Brauer trivial invertible subalgebras}

As a corollary of the above,
we prove a result about translation-invariant
Brauer trivial invertible subalgebras.

\begin{theorem}
    \label{thm:TranslationInvariantISAHasBoundedAncillas}
    Let $\calA$ be a Brauer trivial invertible subalgebra of
    $\fMat(\ZZ^2,p)$ with a translation-invariant
    set of generators. Then there exist uniformly
    bounded local assignments $q$ and $r$ such that
    $\calA\otimes \fMat(\ZZ^2,q)$ is
    bounded-spread isomorphic to $\fMat(\ZZ^2,r)$.
\end{theorem}

\newcommand{\Tint}{T_{\mathrm{int}}}

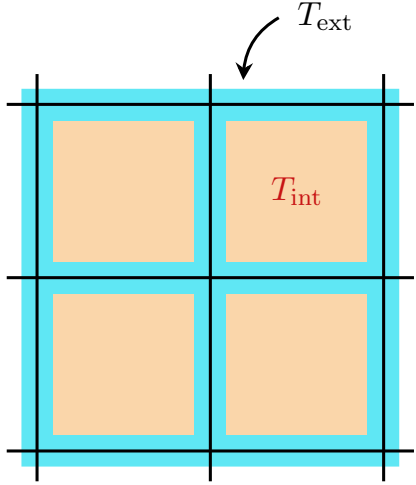
\begin{figure}
    \centering
    \begingroup
        \resizebox{0.4\textwidth}{!}{\definecolor{scyan}{RGB}{95,231,244}
\definecolor{sorange}{RGB}{250,214,170}
\definecolor{sred}{RGB}{200,20,20}

\def\ov{0.35}   
\def\bl{0.18}   
\def\hs{0.80}   
\def\os{0.81}   

\begin{tikzpicture}[line width=1pt]

\fill[scyan,even odd rule]
  (-\bl,-\bl) rectangle (4+\bl,4+\bl)
  \foreach \x in {1,3}{\foreach \y in {1,3}{%
      (\x-\hs,\y-\hs) rectangle (\x+\hs,\y+\hs)}};

\foreach \x in {1,3}{\foreach \y in {1,3}{%
    \fill[sorange] (\x-\os,\y-\os) rectangle (\x+\os,\y+\os);}}

\foreach \c in {0,2,4}{%
  \draw (\c,-\ov) -- (\c,4+\ov);
  \draw (-\ov,\c) -- (4+\ov,\c);}

\node[anchor=west] at (2.86,5.02) {$T_{\text{ext}}$};
\draw[-{Stealth[length=2mm,width=2mm]}]
  (2.79,5.0) to[bend right=30] (2.38,4.28);
\node[sred] at (3,3) {$T_{\text{int}}$};

\end{tikzpicture}}
    \endgroup
    \caption{Copies of $[0,1000\ell]\times [0, 1000\ell]$
    tiling the plane.}
    \label{fig:torus_skeleton}
\end{figure}

\begin{proof}
    Since $\calA$ is Brauer trivial, there is
    a bounded-spread isomorphism $\calA\otimes 
    \fMat(\ZZ^2,q_A)\cong \fMat(\ZZ^2,r_A)$.
    Let $\ell$ be the spread of this isomorphism.
    Take $\fMat(\ZZ^2, q_A)$ supported on
    $B = [0, 1000\ell]\times [0, 1000\ell]$ and copy it
    across the entire plane: in other words,
    let $\fMat(\ZZ^2, q_T)$ be the algebra
    obtained by setting
    \begin{equation}
        \fMat(\{s = (x,y)\}, q_T(s)) =
        \fMat(\{(x\text{ mod } 1000\ell, y \text{ mod }
        1000\ell)\}, q_A).
    \end{equation}
    Define $\fMat(\ZZ^2, r_T)$ similarly.
    Since $\calA$ is translation-invariant,
    we can copy the restricted map
    $\phi: \fMat(B^{-\ell}, r_A)\to (\calA\cap \fMat(B,p))
    \otimes \fMat(B,q_A)$
    across the plane.
    The copies of $B^{-\ell}$ are the orange patches
    in \autoref{fig:torus_skeleton}, the union of which
    we call the interior $\Tint$.
    Thus, we have a map $\tilde{\phi}:
    \fMat(\Tint, r_T)\to \calA\otimes \fMat(\ZZ^2,q_T)$.

    By \ref{lem:full_mat_tensor_factor}, $\im\phi$
    is a tensor factor of $\calA \otimes \fMat(\ZZ^2,q_T)$,
    so we can write
    $\tilde{\phi}(\fMat(\Tint, r_T)) \otimes \calC
    \cong \calA\otimes \fMat(\ZZ^2,q_T)$
    for some $\calC$ supported on a $1\dd$ skeleton
    $T_{\text{ext}} = (\Tint^{-\ell})^c$
    (blue in \autoref{fig:torus_skeleton}).
    By \ref{cor:ISAInverseSameSupport},
    $\calC$ is a genuine invertible subalgebra on
    the $1\dd$ skeleton,
    which is a one-dimensional simplicial complex satisfying
    the conditions of \ref{thm:ISAon1DSkeleton}.
    Applying the theorem, we see that
    after stabilizing with $\FermionMinus{1}$
    per site, $\calC$ is bounded-spread isomorphic
    to some $\fMat$ on the $1\dd$ skeleton.
    Hence $\calA\otimes \fMat(\ZZ^2,q_T)\otimes
    \fMat(\ZZ^2, \FermionMinus{1})$
    is bounded-spread isomorphic to some $\fMat(\ZZ^2)$,
    implying that $\calA$ is Brauer trivial with
    uniformly bounded stabilization.
\end{proof}

\section{Automorphisms of Brauer trivial fermionic algebras}

In this section, all algebra homomorphisms and inclusions are unital.

\begin{definition}
    A \emph{(fermionic) QCA} is a bounded-spread automorphism (see~\ref{def:BoundedSpread})
    of a fermionic algebra~$\fMat(\ZZ^\dd, p)$.
\end{definition}

\begin{lemma}
    \label{lem:inverse_QCA_spread}
    If $\alpha$ is a fermionic QCA with spread $\ell$,
    then $\alpha^{-1}$ is a QCA with spread $\ell$.
\end{lemma}

\begin{proof}
    Applying \ref{lem:supp_by_commutation},
    the condition for an automorphism $\alpha$ to be a QCA is
    $\FermiCommutator{\alpha(x)}{y}=0$ for all
    $x, y\in \fMat(\ZZ^\dd)$ whose supports are separated by 
    distance $\ell$. Applying $\alpha^{-1}$, this is equivalent
    to $\FermiCommutator{x}{\alpha^{-1}(y)} = 0$ for
    such $x$ and $y$. Hence $\alpha^{-1}$ is a QCA with spread $\ell$.
\end{proof}

\subsection{Circuits}

Obvious fermionic QCAs are built out of automorphisms of local algebras.
Here we recall this class of QCAs and (re)prove basic facts.

\begin{example}\label{ex:swap}
    Consider $\SWAP: \FermionMinus{1}^{\otimes 2} \to \FermionMinus{1}^{\otimes 2}$
    defined by~$\gamma_1 \mapsto \gamma_2$ and $\gamma_2 \mapsto \gamma_1$.
    Since $\FermionMinus{1}^{\otimes 2}$ is even and hence every automorphism is inner,
    $\SWAP$ must be a conjugation by a unitary.
    Indeed, it is implemented by an odd unitary~$u = (\gamma_1 + \gamma_2)/\sqrt{2}$.
\end{example}

\begin{lemma}\label{thm:SWAPisEvenOrOdd}
    Let $\calF$ be a finite-dimensional central fermionic algebra.
    Define
    \begin{equation}
        \SWAP: \calF \otimes \calF \ni : x_a \otimes y_b \mapsto (-1)^{ab} y_b \otimes x_a \in \calF \otimes \calF \label{eq:DefSWAP}
    \end{equation}
    for all homogeneous elements~$x_a, y_b \in \calF$ where $a,b \in \ZZ_2$ denote grading.
    Then, $\SWAP$ is the conjugation by a unitary~$w \in \calF \otimes \calF$.
    The unitary $w$ is even if and only if~$\calF$ is even.
\end{lemma}

\begin{proof}
    It is routine to check that $\SWAP$ is a graded homomorphism;
    the only nontrivial piece\footnote{It is instructive to note that this step fails without $(-1)^{ab}$ in~\eqref{eq:DefSWAP}.}
    is $\SWAP((x_a \otimes y_b)(x'_{a'} \otimes y'_{b'})) 
    = (-1)^{ba'} \SWAP(x_a x'_{a'} \otimes y_b y'_{b'})
    = (-1)^{ba'+(a+a')(b+b')} y_b y'_{b'} \otimes x_{a} x'_{a'}
    = \SWAP(x_a \otimes y_b) \SWAP(x'_{a'} \otimes y'_{b'})$.
    Since $\calF \otimes \calF$ is always even (\ref{ex:TableOfTensorProduct})
    and $\SWAP$ is an automorphism of an ungraded central simple $*$-algebra,
    we find by~\ref{thm:AutoIsInnerInMat} 
    a unitary~$w \in \calF \otimes \calF$ implementing~$\SWAP$.
    We may choose~$w$ to be homogeneous by~\ref{thm:AutomorphismByUnitaryConjugationCanAlwaysBeHomogeneous}.
    If $\calF$ is even, the fermion parity unitary of~$\calF \otimes \calF$ is
    the product of the fermion parity unitaries of the factors,
    and, conjugating the parity unitary by \(w\), we see that $w$ is even;
    if $\calF$ is odd, the fermion parity unitary of~$\calF \otimes \calF$ is
    the product of the odd central generators of the factors,
    and hence $w$ is odd.
\end{proof}

Many local automoprhisms may be combined to give a QCA as follows.

\begin{remark}\label{rem:TensorProductOfMorphism}
    For two graded homomorphisms 
    $\alpha: \calA \to \calA'$
    and $\beta : \calB \to \calB'$,
    we can form its tensor product
    \begin{equation}
        \gamma = \alpha \otimes \beta : 
        \calA \otimes \calB \ni x \otimes y \mapsto \alpha(x) \otimes \beta(y) \in \calA' \otimes \calB'
    \end{equation}
    It is routine to check that $\gamma$ is a graded homomorphism:
    $\gamma((x_a \otimes y_b)(x'_i \otimes y'_j))
    =\gamma(x_a x'_i \otimes y_b y'_j)(-1)^{bi}
    =\alpha(x_a x'_i) \otimes \beta(y_b y'_j) (-1)^{bi}
    =\alpha(x_a)\alpha(x'_i) \otimes \beta(y_b)\beta(y'_j) (-1)^{bi}
    =(\alpha(x_a) \otimes \beta(y_b) )(\alpha(x'_i) \otimes \beta(y'_j))$.
    More generally, 
    given a collection of automorphisms~$\alpha_s$, one for every site of some lattice~$\Lambda$,
    we have an automorphism~$\alpha = \bigotimes_{s \in \Lambda} \alpha_s$ of~$\fMat(\Lambda)$,
    where,
    unlike in the definition of infinite tensor product algebras,
    infinitely many $\alpha_s$ may be non-identity.
\end{remark}

\begin{definition}
    A \emph{depth-1 circuit} of local automorphisms
    is a product $\alpha = \prod_s \alpha_s$
    of locally finite set of pairwise commuting automorphisms~$\alpha_s$,
    all of which have a support diameter uniformly bounded across the lattice.
    Its action as an automorphism is defined for every local operator
    by ignoring those whose support do not overlap with that of the local operator.
    A \emph{circuit of local automorphisms}
    is a finite composition of depth-1 circuit of local automorphisms.
    The number of depth-1 circuits in such a composition 
    is the \emph{depth} of the circuit.
    Circuits are \emph{equivalent} 
    if the composite automorphisms are the same.
\end{definition}

We do not consider circuits of infinite depth in this paper,
although such infinite compositions can sometimes makes sense
as a limit~\cite{Haah2023invertible,Ranard2022}.
Following conventional usage,
we often write ``finite depth circuits of local automorphisms''
to emphasize that there are finitely many layers of local automorphisms.

\begin{remark}
    ``Circuits'' are most commonly discussed with unitary gates,
    where the unitaries act on a local operator algebra by conjugation.
    Since we have noted in~\ref{thm:AutomorphismByUnitaryConjugationCanAlwaysBeHomogeneous}
    that every local automorphism is the conjugation by some homogeneous local unitary,
    it is tempting to identify a circuit of local unitary elements
    with a circuit of local automorphisms.
    However, the local unitary corresponding to a local automorphism may be odd,
    and odd operators of disjoint support anticommute
    while corresponding local automorphisms commute.
    Hence, the relation between circuits of local automorphisms and circuits of unitary gates
    requires more careful treatment.
    We elaborate on this subtlety in~\cite{paper2}.
\end{remark}

Circuits of local automorphisms are closely related
to the commutator subgroup of the (stable) group of all QCA~\cite{FHH2019}.
The following lemmas provide the basic connection.

\begin{lemma}\label{lem:CircuitGroupIsNormal}
    The group of all finite depth circuits of local automorphisms
    is a normal subgroup of the group of all QCAs on~$\ZZ^\dd$.
\end{lemma}

This is a standard result.

\begin{proof}
    It suffices to show that for any QCA~$\alpha$ of spread \(\ell\),
    the conjugation $\alpha \circ (\prod_s \sigma_s) \circ \alpha^{-1}$
    of a depth-$1$ circuit is again a circuit.
    Let $\mu_s = \alpha \circ \sigma_s \circ \alpha^{-1}$. 
    Observe that $\{ \mu_s \}$ consists of pairwise commuting automorphisms.
    If $\sigma_s$ is supported on~$R$, then $\mu_s$ is supported on~$R^{+\ell}$
    where $\ell$ is the spread of~$\alpha$.
    Hence, $\{\mu_s\}$ is locally finite and the support diameter is uniformly bounded.
\end{proof}

\begin{proposition}\label{thm:QCAQCAInverseIsCircuit}
    For any lattice~$\Lambda$ and a QCA~$\alpha$ on~$\fMat(\Lambda)$,
    the product
    $\alpha \otimes \alpha^{-1}$ is a finite depth circuit of local automorphisms
    with the individual spread depending only on the spread~$\ell$ of~$\alpha$.
\end{proposition}

This is a generalization of the ungraded result in~\cite{Arrighi2007}.
It turns out that the grading is immaterial.

\begin{proof}
    Let $\sigma_s$ be $\SWAP$ 
    acting on a pair~$s$ of corresponding sites in the doubled system~$\fMat(\ZZ^\dd)^{\otimes 2}$. 
    First, we observe that a depth-1 circuit~$\mathsf S = \bigotimes_s \sigma_s$ 
    of local swaps on the doubled system gives 
    \begin{equation}
        \alpha \otimes \alpha^{-1} = 
        (\alpha \otimes \id) \mathsf S (\alpha^{-1} \otimes \id) \mathsf S^{-1}.
        \label{eq:alphaAlphaSwap}
    \end{equation}
    Indeed,
    \begin{align}
        x_i \otimes y_j
        &\xmapsto{~\mathsf S^{-1}~}
        (-1)^{ij} y_j \otimes x_i \nonumber\\
        &\xmapsto{~\alpha^{-1} \otimes \id~}
        (-1)^{ij}\alpha^{-1}(y_j) \otimes x_i\\
        &\xmapsto{~\mathsf S~}
        x_i \otimes \alpha^{-1}(y_j)\nonumber\\
        &\xmapsto{~\alpha \otimes \id~} \alpha(x_i) \otimes \alpha^{-1}(y_j) \, .\nonumber
    \end{align}
    Next, as in the ungraded case~\cite{Arrighi2007},
    we observe that the composition 
    $(\alpha \otimes \id) \mathsf S (\alpha^{-1} \otimes \id)$
    is a circuit by~\ref{lem:CircuitGroupIsNormal}.  
\end{proof}

\begin{corollary}\label{cor:QCAmodCircuitsIsAbelian}
    For any lattice~$\Lambda$
    and any QCA~$\alpha,\beta$ on~$\fMat(\Lambda, p)$,
    the stabilized commutator~$(\alpha\circ\beta\circ\alpha^{-1}\circ\beta^{-1}) \otimes \id$
    on~$\fMat(\Lambda, p^2)$
    is a finite depth circuit.
\end{corollary}

This is another standard result.

\begin{proof}
    Let $\sim$ denote the equivalence up to a circuit.
    We do not have to specify whether the circuit is precomposed or postcomposed
    because the circuit group is normal (\ref{lem:CircuitGroupIsNormal}).
    Now, $(\alpha\circ\beta) \otimes \id \sim \alpha \otimes \beta
    \sim (\beta \circ \alpha) \otimes \id$
    because $\beta^{-1} \otimes \beta$ and $\beta \otimes \beta^{-1}$
    are both circuits by~\ref{thm:QCAQCAInverseIsCircuit}.
\end{proof}

The stabilization in~\ref{cor:QCAmodCircuitsIsAbelian}
can often be removed~\cite{FHH2019}.
We state a special version that we use later.

\begin{lemma}[adapted from {\cite[\S2]{FHH2019}}]\label{lem:AncillaRemoval}
    Let $p,m$ be local assignments,
    where $m$ assigns a fixed local algebra on every site.
    Consider a circuit of local automorphisms that implements a QCA~$\alpha \otimes \id$
    on $\fMat(\ZZ^\dd, p m) \otimes \fMat(\ZZ^\dd, m^k)$ for some $k \ge 1$,
    where the individual local automorphisms may act nontrivially on~$\fMat(\ZZ^\dd, m^k)$.
    Then, $\alpha$ without stabilization is a circuit of local automorphisms,
    where the individual local automorphisms act by~$\id$ on~$\fMat(\ZZ^\dd, m^k)$.
\end{lemma}

\begin{proof}
    See \autoref{fig:ancilla-removal} and \cite[\S2]{FHH2019}.
\end{proof}

\begin{figure}
    \centering
    \begingroup
        \hspace*{-10ex}
        \resizebox{0.95\textwidth}{!}{\input{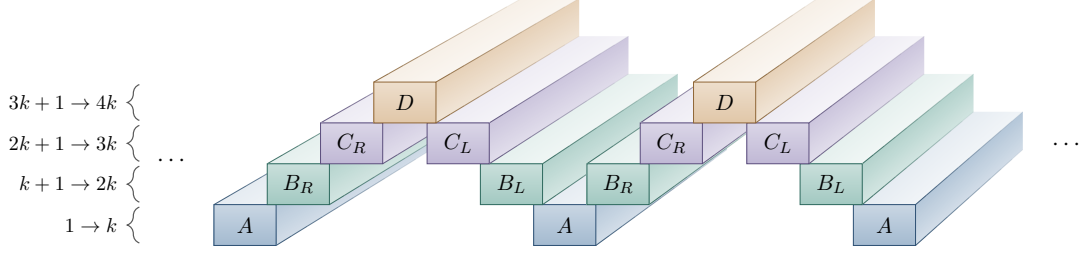}}
    \endgroup
    \caption{Ancilla (fermion) removal.
    Any finite depth circuit of local automorphisms can be rewritten with $4k$ layers $U_1, U_2,\ldots, U_{4k}$,
    where each $U_j$ is a circuit supported on the union~$R_j$ of far separated regions
    (strips in $2\dd$).
    Each $R_j$ takes up a fraction at most 
    $1/k$ of all sites.
    Any (fermionic) ancilla used by layer~$j_a$ 
    with~$j_a \in [1,k]$
    can be borrowed from a nearby ancilla in
    $R_{2k+1} \cup \cdots \cup R_{3k}$ and immediately returned.
    Each layer just needs $O(1)$ strips
    to borrow ancillas from,
    but we give a big enough pool of ancillas 
    by taking $k$ large.
    The ancillas used by layers $j_b \in [k+1,2k]$
    are borrowed from~$R_{3k+1} \cup \cdots \cup R_{4k}$;
    those used by layers $j_c \in [2k+1,3k]$
    are from $R_1 \cup \cdots \cup R_k$;
    finally, those used by layers $j_d \in [3k+1,4k]$
    are from $R_{k+1} \cup \cdots \cup R_{2k}$.
    }
    \label{fig:ancilla-removal}
\end{figure}

\subsection{Shifts}

Another example QCA is to permute simple factors within a bounded distance.
We consider a generalization of permutations
and show that they are virtually equivalent to permutations.

\begin{definition}
    A QCA $\alpha$ on a lattice~$\Lambda$ 
    is a \emph{shift} if for any site $s\in \Lambda$,
    there is a factorization $\fMat(\{s\},p(s)) = \bigotimes_j \calA_j$
    of the single-site algebra
    into simple fermionic algebras $\calA_j$ such that
    the image $\alpha(\calA_j)$ is supported on a single site
    for each $j$.
    If every such factorization consists of primitive algebras,
    then we say the shift is \emph{primitive}.
    A shift is \emph{monolayer} if each single-site algebra
    $\fMat(\{s\},p(s))$ is mapped onto a single-site algebra.
\end{definition}

With stabilization every shift becomes primitive.

\begin{lemma}\label{lem:ShiftsAreCircuitEquivToPrimitive}
    Let \(\alpha : \fMat(\Lambda,p) \to \fMat(\Lambda,p)\) be a shift QCA.
    Then there is a local assignment~\(p'\)
    such that \(\alpha \otimes \id : \fMat(\Lambda,pp') \to \fMat(\Lambda,p p')\)
    is circuit equivalent to a primitive shift. 
    The assignment \(p'\) is uniformly bounded if \(p\) is uniformly bounded.
\end{lemma}

\begin{proof}
    Let \(\fMat(\{s\},p(s)) = \bigotimes_{j} \calA_{j_s}'\) be the decomposition of a single-site algebra into simple factors promised by the definition of a shift.
    Note that we also have a decomposition of every site as \(\fMat(\{s\},p(s)) = \bigotimes_{j} \alpha(\calA_{j_t}')\) by \ref{thm:Commutant-of-central-in-central} and the fact that \(\alpha\) is an automorphism. 
    Here, \(\calA_{j_t}'\) is one of the simple algebras for a nearby site \(t\).
    
    Define \(p'\) as a product of two Majorana algebras
    \(\FermionMinus{1}^{\otimes 2} = \FermionMinus{1}_s \otimes \FermionMinus{1}_t\)
    for each \(\calA'_{j_s}\),
    placing one factor on the same site~$s$ as~\(\calA'_{j_s}\)
    and the other on the same site~$t$ as~\(\alpha(\calA'_{j_s})\).
    Let \(\delta_{j_s}: \calA_{j_s}' \otimes \FermionMinus{1}_s
    \xrightarrow{\cong} \calA_{j_s} \otimes \FermionMinus{1}_s\)
    be an isomorphism from the simple factors and one stabilizing Majorana
    to a tensor product of primitive algebras.
    At the image site, similarly define
    \(\delta'_{j_t}: \alpha(\calA_{j_s}') \otimes \FermionMinus{1}_t
    \xrightarrow{\cong} \calA_{j_t} \otimes \FermionMinus{1}_t\)
    by \(\delta'_{j_s} =  (\beta \otimes \iota)\circ\delta_{j_s} \circ (\alpha^{-1}\otimes\iota^{-1})\),
    where \(\iota: \FermionMinus{1}_s \to \FermionMinus{1}_t\)
    and \(\beta: \calA_{j_s} \to \calA_{j_t}\) are arbitrary isomorphisms.
    In particular, \(\delta'_{j_t} \circ (\alpha \otimes \iota) \circ \delta_{j_s}^{-1} = \beta \otimes \iota\)
    is a tensor product of isomorphisms on the primitive algebras.
    
    Defining \(\Delta^{(\prime)} = \bigotimes_{s,j_s} \delta^{(\prime)}_{j_s}\) (taking the product across the whole lattice), we observe that
    \begin{equation}
        \tilde{\alpha} = \Delta' \circ (\alpha \otimes \mathsf{S}) \circ \Delta^{-1}
    \end{equation}
    is a primitive shift, where \(\mathsf{S}\) is a depth-1 circuit of swaps between the (disjoint) pairs of stabilizing Majorana algebras,
    \begin{align}
        \mathsf{S} = \bigotimes_{s,j_s} \SWAP_{j_s},
        \quad
        \SWAP_{j_s} : \FermionMinus{1}_s \otimes\FermionMinus{1}_t &\to \FermionMinus{1}_s \otimes\FermionMinus{1}_t \\
        x \otimes y &\mapsto (-1)^{|x| \cdot |y|} \iota^{-1}(y) \otimes \iota(x). \nonumber
    \end{align}
    (Here, \(|x|\) is the grading of the homogeneous element \(x\).)
    Indeed, for each pair of primitive algebras \(\calA_{j_s} \otimes \FermionMinus{1}_s\), the factorization of \(\delta'_{j_t} \circ (\alpha \otimes \iota) \circ \delta_{j_s}^{-1} = \beta \otimes \iota\) ensures that every primitive factor is simply shifted to another site.
    The remaining \(\FermionMinus{1}_t\) algebras are also just swapped back to site \(s\).
    The primitive shift \(\tilde{\alpha}\) differs from \(\alpha\) by a tensor product with a depth-1 circuit, and multiplication on the left and right by depth-1 circuits.
    Thus, it is finite-depth circuit equivalent to \(\alpha\).
\end{proof}

\begin{lemma}\label{lem:PrimitiveShiftToPermutation}
    For any lattice~$\Lambda$ and any primitive shift QCA~$\alpha$ on~$\fMat(\Lambda,p)$,
    there exists a circuit~$\delta$ of on-site automorphisms (depth~$1$ and spread~$0$)
    such that $\delta \circ \alpha$ is a permutation of primitive algebra factors.
\end{lemma}

\begin{proof}
    Since $\alpha$ is a primitive shift,
    we have a preferred factorization $\fMat(\{s\}, p(s) ) = \bigotimes \calA_{s,j}$
    for each single-site algebra at~$s \in \Lambda$
    such that  $\alpha(\calA_{s,j})$ is supported on a single site.
    Then, the preimage under~$\alpha$
    of the single-site algebra~$\bigotimes_j \calA_{s,j}$ at~$s$
    must be generated by 
    a mutually fermionically commuting collection
    of primitive algebras~$\alpha(\calA_{s',j'(s')})$
    with $s'$ near~$s$.
    By the unique factorization property~\ref{thm:StabilizedThenDecomposed},
    there exists a single-site automorphism at~$s$ 
    that sends $\alpha(\calA_{s',j'(s')})$ to~$\calA_{s,i(s')}$ for each~$s'$.
    Let $\delta$ be the circuit consisting of all these single-site automorphisms.
\end{proof}

The next lemma shows that
any primitive shift may be represented by
a permutation on a minimal system.

\begin{lemma}\label{lem:AncillaReduction}
    Let $p$ be an arbitrary local assignment on a lattice~$\Lambda$.
    Let $m$ be a uniform local assignment $\Lambda \ni s \mapsto \calM$
    for a fixed primitive algebra~$\calM$.
    For an integer $k \ge 1$,
    any primitive shift~$\alpha$ on~$\fMat(\Lambda, p m^{k+1})$
    is a primitive shift~$\beta \otimes \id$
    on~$\fMat(\Lambda, p m^k)\otimes \fMat(\Lambda, m)$
    followed by a circuit of local automorphisms.
\end{lemma}

\begin{proof}
    (To a permutation of~$\calM$'s)
    From~\ref{lem:PrimitiveShiftToPermutation},
    we have a circuit $\delta_1$ of spread~$0$ such that 
    $\alpha_1 = \delta_1 \circ \alpha$ is a permutation of primitive algebras,
    which is accompanied by a preferred factorization of~$pm^{k+1}$ into primitive factors.
    We then isolate a permutation among the tensor powers of~$\calM$.
    In this isolated permutation, every site has at least~$k+1$ factors of~$\calM$.
    Hence, $\alpha_1 = \delta_1 \circ \alpha$ induces a permutation on
    a vertex set $V = \{ (s,j) \mid s \in \Lambda,\, j \in \{0,1,2,\ldots,k,\ldots,j_\mathrm{max}(s)\}$
    with a bounded spread on the $\Lambda$-component.
    We are going to use the same letter~$\alpha_1$ to denote this permutation.
    Define $B = \Lambda \times \{0\} \subseteq V$ and $A = V \setminus A$.
    The vertex set~$V$ maps onto the set~$\Lambda$ of all sites.
   
    (To simple cycles)
    Decompose $\alpha_1$ into disjoint cycles.
    We say a cycle is \emph{simple} if it visits any site $s \in \Lambda$ at most once;
    a cycle can visit a site~$s$ at most $j_\mathrm{max}(s)+1$ times.
    We claim that there exists a depth-1 circuit~$\delta_2$ of on-site permutations
    such that $\alpha_2 = \delta_2 \circ \alpha_1$ consists of simple cycles only.
    Suppose $\tau$ is a nonsimple cycle that visits a site~$s$ more than once:
    \begin{equation}
        \cdots \to (s',i') \to (s,i) \to (s'',i'') \to \cdots \to (s''',i''') \to (s,j) \to (s'''',i'''') \to \cdots
    \end{equation}
    where $\tau$ does not visit $s$ between $(s,i)$ and $(s,j)$.
    Let $\pi$ be the transposition of~$(s,i)$ and~$(s,j)$.
    Then, $\pi \circ \tau$ is a product of two disjoint cycles
    \begin{equation}
        \big(\cdots \to (s',i') \to (s,j) \to (s'''',i'''') \to \cdots \big)
        \cdot
        \big( (s,i) \to (s'',i'') \to \cdots \to (s''',i''') \to (s,i) \big),
    \end{equation}
    the first of which visits~$s$ one fewer time and the second visits~$s$ only once.
    Repeatedly applying this on every site, we see that every cycle becomes simple.
    Therefore, $\delta_2$ is constructed.
    
    (Decoupling~$B$ from~$A$)
    We claim that there exists a depth-1 circuit~$\delta_3$ of on-site permutations
    such that $\alpha_3 = \delta_3 \circ \alpha_2 \circ \delta_3^{-1}$
    has no cycle that traverses between $B$ and $A$.
    The conjugation by~$\delta_3$ is a permutation 
    of the on-site index~$i \in \{0,1,2,\ldots,k,\ldots,j_\mathrm{max}(s)\}$.
    Imagine that we travel along a simple cycle starting with a $B$-vertex,
    and whenever we encounter a $A$-vertex $(s,i)$,
    we swap~$i$ with~$0$.
    Then, this simple cycle become entirely supported on~$B$.
    Once all $B$-vertices are visited by some simple cycle,
    which is entirely supported on~$B$,
    no other simple cycle can traverse between~$B$ and~$A$.
    Therefore, $\alpha_3$ is a tensor product of disjoint permutations
    $\alpha_3^\calA$ on~$A$ and $\alpha_3^\calB$ on~$B$.

    (Tossing cycles from~$B$ to~$A$)
    The vertex subset~$A$ contains
    a copy of $B = \Lambda \times \{ 0 \}$;
    we choose $\Lambda \times \{1\}$ as a preferred copy.%
    \footnote{This is where we use the assumption~$k \ge 1$.}
    Let $\alpha_3^{\calB,\mathrm{copy}}$ be the copy of~$\alpha_3^\calB$
    permuting vertices of the preferred copy~$\Lambda \times \{1\}$.
    As a QCA, $\alpha_3^{\calB,\mathrm{copy}}$ acts on~$\fMat(A)$.
    Let $\delta_4 = \alpha_3^{\calB,\mathrm{copy}}\otimes(\alpha_3^\calB)^{-1}$.
    We see that $\alpha_4 = \delta_4 \circ \alpha_3$ is a permutation
    that acts by identity on~$\fMat(B)$.
    By~\ref{thm:QCAQCAInverseIsCircuit},
    we know that $\delta_4$ is a circuit of local automorphisms.
    
    Define $\beta \otimes \id = \alpha_4$ on~$\fMat(A) \otimes \fMat(B)$.
    Unrolling all deformations above,
    we have $(\beta \otimes \id) = \delta_4 \delta_3 \delta_2 \delta_1 \alpha \delta_3^{-1}$
    or
    $\alpha = \big(\delta_1^{-1} \delta_2^{-1} \delta_3^{-1} \delta_4^{-1} (\beta \otimes \id) \delta_3 (\beta \otimes \id)^{-1} \big) (\beta\otimes\id) $.
    By~\ref{lem:CircuitGroupIsNormal},
    the circuit~$\delta_3$ conjugated by~$\beta$ is a circuit,
    and we complete the proof.
\end{proof}

\subsection{QCA-algebra correspondence}

In this subsection we are going to establish the correspondence between
QCA on $\dd$-dimensional lattice and $(\dd-1)$-dimensional invertible subalgebras.
A useful equivalence relation to this end is the following.

\begin{definition}
    Let $\alpha$ be a QCA on $\fMat(\ZZ^\dd,p)$ and $\beta$
    be a QCA on $\fMat(\ZZ^\dd,q)$.
    For any $\ell > 0$, let 
    $S_-(n,\ell) = (\ZZ\cap (-\infty,n-\ell)) \times \ZZ^{\dd-1}$
    and $S_+(n,\ell) = (\ZZ\cap (n+\ell,\infty)) \times \ZZ^{\dd-1}$.
    We say that $\alpha$ \emph{blends into} $\beta$
    along the first axis if there exists $\ell > 0$
    (unrelated to the spread of the two QCAs),
    a position of the interface $n \in \ZZ$,
    and a spread-$\ell$ blending QCA $\gamma$ of $\fMat(\ZZ^\dd,r)$
    such that
    \begin{align}
    r(s) = p(s) \text{ and } \gamma(x) = \alpha(x) 
    &\text{ for all }
    x \in \fMat(S_-(n,\ell),p),
    s \in S_-(n,\ell),\\
    r(s) = q(s) \text{ and } \gamma(x) = \beta(x) 
    &\text{ for all }
    x \in \fMat(S_+(n,\ell),q),
    s \in S_+(n,\ell). \nonumber
    \end{align}
\end{definition}

\begin{lemma}\label{lem:BlendingIsEquivalence}
    Along a fixed axis,
    a QCA $\alpha$ blends into a QCA $\beta$
    with blending interface at~$n$,
    if and only if $\alpha$ blends into~$\beta$
    with blending interface anywhere,
    if and only if $\beta$ blends into 
    $\alpha$ with blending interface anywhere.
\end{lemma}

\begin{proof}
    The proof is exactly the same as in~\cite[Lemma~3.12]{Haah2023invertible},
    except that we use~\ref{thm:QCAQCAInverseIsCircuit}
    instead of the ungraded result of~\cite{Arrighi2007}.
\end{proof}

\begin{corollary}
    \label{cor:QCA_group}
    Given an axis, blending is an equivalence relation,
    modulo which the set of all $\dd$-dimensional QCAs
    form an abelian group under tensor product,
    with the identity class $[\id]$ represented by the identity QCA
    and the inverse of $[\alpha]$ represented by $\alpha^{-1}$.
\end{corollary}

\begin{proof}
    Clear by \ref{lem:BlendingIsEquivalence} 
    and the fact that $\alpha \otimes \alpha^{-1}$
    blends into the identity (\ref{thm:QCAQCAInverseIsCircuit}).
\end{proof}

\begin{theorem}\label{thm:QCAToBrauer}
    The abelian group of all blending equivalence classes of
    $\dd$-dimensional fermionic QCA (blending along a fixed axis),
    modulo those represented by shift QCA, is isomorphic
    to the Brauer group of $(\dd-1)$-dimensional
    fermionic invertible subalgebras.
\end{theorem}
\noindent
This parallels the established correspondence~\cite{Haah2023invertible}
for the ungraded case, and we follow the same argument.

For any $n\in \ZZ$, we define
\begin{equation}
    \begin{aligned}
        H(n) &= (\ZZ\cap (-\infty, n])\times \ZZ^{\dd-1} &\text{(a half space)},\\
        T(n,\ell) &= (\ZZ\cap (-\infty, n-\ell])\times \ZZ^{\dd-1} &\text{(a slightly smaller half space)},\\
        S(n,\ell) &= (\ZZ\cap [n-\ell+1,n+\ell])\times\ZZ^{\dd-1} &\text{(a slab)} .\\
    \end{aligned}
    \label{def:regions}
\end{equation}

\begin{lemma}
    \label{lem:boundary_alg_tensor_factor}
    Let $\alpha$ be a QCA with spread $\ell > 1$ on $\fMat(\ZZ^\dd)$.
    Define
    \begin{equation}
        \calB(n,\ell) = \alpha(\fMat(H(n)))\cap \fMat(S(n,\ell)).
    \end{equation}
    Then $\calB(n,\ell)$ is a fermionic subalgebra of $\fMat(S(n,\ell))$
    such that the multiplication map
    \begin{equation}
        \mul: \fMat(T(n,\ell))\otimes\calB(n,\ell)\to
        \alpha(\fMat(H(n)))
    \end{equation}
    is an isomorphism of fermionic algebras.
\end{lemma}

\begin{proof}
    We suppress the arguments $n$ and $\ell$ for clarity.
    By the definition of QCA spread,
    $\alpha(\fMat(H))\subseteq \fMat(T \sqcup S)$.
    Using \ref{lem:inverse_QCA_spread}, $\alpha^{-1}$
    is also a QCA with spread $\ell$, so
    $\alpha^{-1}(\fMat(T))\subseteq \fMat(H)$.
    Applying $\alpha$ gives $\fMat(T)\subseteq \alpha(\fMat(H))$,
    so we are done by \ref{lem:full_mat_tensor_factor}.
\end{proof}

\begin{definition}
    The subalgebra $\calB(n,\ell)$ defined in 
    \ref{lem:boundary_alg_tensor_factor} is called
    a \emph{boundary algebra} of $\alpha$ on the positive
    first axis.
\end{definition}

\begin{lemma}
    Every boundary algebra of a fermionic QCA on $\ZZ^\dd$
    is a fermionic invertible subalgebra in $\dd-1$ dimensions.
    If the QCA has spread at most $\ell$,
    then the invertible subalgebra has spread at most $2\ell$.
\end{lemma}

\begin{proof}
    We continue with the notation of 
    \ref{def:regions} and \ref{lem:boundary_alg_tensor_factor}.
    Let $\calA$ be the fermionic commutant of $\calB$
    within $\fMat(S)$; we must show that every $y\in \fMat(S)$
    admits an $\calA\calB$ decomposition.
    Since $\alpha^{-1}$
    is a QCA of spread $\ell$ by \ref{lem:inverse_QCA_spread},
    $\alpha^{-1}(y)$ is supported
    on the wider slab $S^{+\ell}$.
    Using $\fMat(\ZZ^\dd) = \fMat(H)\otimes \fMat(\ZZ^\dd\setminus H)$,
    we have a decomposition
    $\alpha^{-1}(y) = \sum_i b_i'\otimes a_i'$,
    where $b_i'\in \fMat(H)$ and $a_i'\in \fMat(\ZZ^\dd\setminus H)$.
    Furthermore we choose the decomposition
    so that $\{a_i'\}$ and $\{b_i'\}$
    form linearly independent sets.
    Now define $a_i = \alpha(a_i')\in\alpha(\fMat(\ZZ^\dd\setminus H))$
    and $b_i = \alpha(b_i')\in \alpha(\fMat(H))$,
    which are supported on $\Supp(y)^{+2\ell}$ and
    yield the decomposition
    \begin{equation}
        y = \sum_i b_i\otimes a_i\in \alpha(\fMat(H))\otimes
        \alpha(\fMat(\ZZ^\dd\setminus H)) \cong \fMat(\ZZ^\dd).
    \end{equation}
    By definition, $a_i$ fermionically commutes
    with all of $\calB$, so we just need to show
    that $a_i$ and $b_i$ are supported on $S$.
    
    Define the region to the right of the slab $S$
    as $R(n,\ell) = (\ZZ\cap [n+\ell+1, \infty))\times \ZZ^{\dd-1}$.
    The locality of the QCA implies that
    $\Supp(b_i)\subseteq T\sqcup S$ and
    $\Supp(a_i)\subseteq S\sqcup R$;
    we also have $\fMat(T)\subseteq \alpha(\fMat(H))$
    and $\fMat(R)\subseteq \alpha(\fMat(\ZZ^\dd\setminus H))$.
    Since $y\in \fMat(S)$, any $t\in \fMat(T)$ fermionically
    commutes with $y$; 
    however, by~\eqref{eq:graded_derivation}
    we also have
    \begin{equation}
        0 = \FermiCommutator{t}{y} = \sum_i 
        \FermiCommutator{t\otimes \one}{b_i \otimes a_i}
        = \sum_i \FermiCommutator{t}{b_i}\otimes a_i,
    \end{equation}
    which forces $\FermiCommutator{t}{b_i} = 0$
    as the $a_i$'s are linearly independent.
    Similarly, any $r\in \fMat(R)$ fermionically commutes
    with $y$, which by the same argument forces
    $\FermiCommutator{r}{a_i} = 0$ for all $i$.
    Now by \ref{lem:supp_by_commutation},
    $a_i$ and $b_i$ must be supported on $S$, as needed.
\end{proof}

\begin{lemma}
    \label{lem:BoundaryAlgebraShift}
    Let $\calB(n,\ell)$ denote the boundary algebra
    as in \ref{lem:boundary_alg_tensor_factor}.
    Then we have
    \begin{equation}
    \begin{aligned}
        \calB(n,\ell+1) &\cong \fMat(\{n-\ell\}\times \ZZ^{\dd-1})
        \otimes \calB(n,\ell), \\
        \calB(n+1,\ell+1) &\cong \alpha(\fMat(\{n+1\}\times \ZZ^{\dd-1})) \otimes \calB(n,\ell).
    \end{aligned}
    \end{equation}
\end{lemma}

\begin{proof}
    Since $\fMat(\{n-\ell\}\times\ZZ^{\dd-1})\subseteq \calB(n,\ell+1)$,
    applying \ref{lem:full_mat_tensor_factor} gives
    \begin{equation}
        \calB(n,\ell+1) \cong \fMat(\{n-\ell\} \times\ZZ^{\dd-1})\otimes
        (\calB(n,\ell+1)\cap \fMat((\ZZ\setminus \{n-\ell\})\times \ZZ^{\dd-1}))
    \end{equation}
    The second tensor factor is exactly $\calB(n,\ell)$.

    For the second part, note that $\calB(n+1,\ell+1)$ contains
    $\alpha(\fMat(\{n+1\}\times\ZZ^{\dd-1}))$, so
    applying~\ref{lem:full_mat_tensor_factor} again
    yields
    \begin{equation}
    \begin{aligned}
        \alpha^{-1}(\calB(n+1,\ell+1)) &\cong
        \fMat(\{n+1\}\times\ZZ^{\dd-1}) \\
        & \quad \otimes
        (\alpha^{-1}(\calB(n+1,\ell+1))\cap \fMat((\ZZ\setminus \{n+1\})\times\ZZ^{\dd-1})).
    \end{aligned}
    \end{equation}
    By the definition of the boundary algebra,
    $\calB(n+1,\ell+1)\cap \alpha(\fMat((\ZZ\setminus \{n+1\})\times\ZZ^{\dd-1})) = \calB(n,\ell)$,
    so the second tensor factor is exactly $\alpha^{-1}(\calB(n,\ell))$.
    This proves the lemma.
\end{proof}

\begin{corollary}
    \label{lem:boundary_alg_well_defined}
    Let $\alpha$ be a $\dd$-dimensional QCA with spread $\ell$.
    Then for any $n_1, n_2\in \ZZ$ and $\ell_1, \ell_2\ge \ell$,
    the boundary algebras $\calB(n_1, \ell_1)$ and $\calB(n_2,\ell_2)$
    (defined as in \ref{lem:boundary_alg_tensor_factor})
    are stably equivalent.
\end{corollary}

\begin{proof}
    This immediately follows from \ref{lem:BoundaryAlgebraShift},
    since $\alpha(\fMat(\{n+1\}\times\ZZ^{\dd-1}))$ is
    isomorphic to some $\fMat(\ZZ^{\dd-1})$.
\end{proof}

Thus, the stable equivalence class of the boundary algebra
does not depend on the choice of $n$ and $\ell$.

\begin{lemma}
    \label{lem:boundary_algebra_commutant}
    Continuing with the notation of 
    \ref{lem:boundary_alg_tensor_factor},
    we have
    \begin{equation}
        \calA(n,\ell)\calB(n,\ell) = \fMat(S(n,\ell)),
    \end{equation}
    where $\calA(n,\ell) = \alpha(\fMat(\ZZ^\dd\setminus H(n)))\cap
    \fMat(S(n,\ell))$.
\end{lemma}

\begin{proof}
    Applying $\alpha^{-1}$, it is clear that $\calA(n,\ell)$ is
    the fermionic commutant of $\calB(n,\ell)$ 
    within $\fMat(S(n,\ell))$,
    so the result follows by \ref{cor:invertible_isomorphism}.
\end{proof}

\begin{lemma}
    \label{lem:alg_to_QCA}
    Given any invertible subalgebra $\calB$ in $\fMat(\ZZ^{\dd-1},p)$,
    there exists a local operator algebra $\fMat(\ZZ^\dd)$
    with a QCA $\alpha$, such that $\calB$ is stably equivalent
    to its boundary algebra. If $\calB$ has spread at most $\ell \ge 1$,
    then $\alpha$ has spread at most $\ell$.
\end{lemma}

\begin{proof}
    Construct the local operator algebra in $\dd$ dimensions
    by copying $\fMat(\ZZ^{\dd-1},p)$ at each point
    of a new axis $\ZZ$; we will denote the copy at point $j\in \ZZ$ by
    $\fMat(\ZZ^{\dd-1}_j) \coloneq \fMat(\{j\}\times \ZZ^{\dd-1})$.
    Let $\calB_j$ be the copy of $\calB$ in $\fMat(\ZZ^{\dd-1}_j)$,
    and let $\calA_j = \Commf(\calB, \fMat(\ZZ^{\dd-1}_j))$.
    Then for any $x^{(j)}\in \fMat(\ZZ^{\dd-1}_j)$, we can use
    the invertible property (\ref{cor:invertible_isomorphism}) to write
    $x^{(j)} = \sum_i a_i^{(j)}b_i^{(j)}$, where 
    $a_i^{(j)}\in \calA_j$ and $b_i^{(j)}\in\calB_j$ are supported
    on the $\ell$-neighborhood of $\Supp(x^{(j)})\subseteq
    \{j\}\times\ZZ^{\dd-1}$.
    Consider the maps
    \begin{equation}\begin{aligned}
        \alpha_j: & \fMat(\ZZ^{\dd-1}_j)\ni x^{(j)} =
        \sum_i a_i^{(j)}b_i^{(j)}\mapsto \sum_i a_i^{(j)}
        \otimes b_i^{(j+1)} 
        \in \fMat(\ZZ^{\dd-1}_j)\otimes \fMat(\ZZ^{\dd-1}_{j+1}) \\
        \beta_j: & \fMat(\ZZ^{\dd-1}_j)\ni 
        x^{(j)} = \sum_i a_i^{(j)}b_i^{(j)}\mapsto \sum_i a_i^{(j)}
        \otimes b_i^{(j-1)}
        \in \fMat(\ZZ^{\dd-1}_j)\otimes \fMat(\ZZ^{\dd-1}_{j-1}).
    \end{aligned}\end{equation}
    To see that these are well-defined 
    fermionic algebra homomorphisms,
    note that $\alpha_j$ is just the composition
    \begin{equation}
        \fMat(\ZZ_j^{\dd-1}) \xrightarrow{\mul^{-1}}
        \calA_j\otimes \calB_j \longrightarrow
        \calA_j\otimes \calB_{j+1}\lhook\joinrel\longrightarrow
        \fMat(\ZZ^{\dd-1}_j)\otimes \fMat(\ZZ^{\dd-1}_{j+1}),
    \end{equation}
    where $\mul^{-1}$ exists by \ref{cor:invertible_isomorphism}.
    To construct QCAs on $\fMat(\ZZ^\dd)$,
    we define
    \begin{equation}\begin{cases}
        \alpha: \fMat(\ZZ^\dd)& \to\fMat(\ZZ^\dd), \quad 
        x = \prod_j x^{(j)}\mapsto \prod_j \alpha_j(x^{(j)}) \\
        \beta: \fMat(\ZZ^\dd)& \to\fMat(\ZZ^\dd), \quad 
        x = \prod_j x^{(j)}\mapsto \prod_j \beta_j(x^{(j)}),
    \end{cases}\end{equation}
    where the maps are extended to all of $\fMat(\ZZ^\dd)$
    using $\CC$-linearity.
    Since $\calA$ and $\calB$ fermionically commute,
    the images of $\alpha_j$ and $\alpha_{j'}$ fermionically
    commute for all $j\neq j'$.
    Hence $\alpha$ is a well-defined fermionic algebra
    homomorphism by \ref{lem:TensorProdUniversalProperty},
    and likewise for $\beta$.
    Moreover, $\beta \circ \alpha = \id = \alpha \circ \beta$:
    \begin{equation}\begin{cases}
        x^{(j)} 
        &\xmapsto{\quad \alpha \quad } \sum_i a_i^{(j)} b_i^{(j+1)}
        \xmapsto{\quad \beta \quad } \sum_i \beta_j(a_i^{(j)})\beta_{j+1}(	b_i^{(j+1)}) 
        = \sum_i a_i^{(j)} b_i^{(j)} = x^{(j)},\\
        x^{(j)} 
        &\xmapsto{\quad\beta\quad } \sum_i a_i^{(j)} b_i^{(j-1)}
        \xmapsto{\quad\alpha\quad } \sum_i \alpha_j(a_i^{(j)})\alpha_{j-1}(	b_i^{(j-1)}) 
        = \sum_i a_i^{(j)} b_i^{(j)} = x^{(j)}.
    \end{cases}\end{equation}
    Therefore, $\alpha$ is an automorphism of $\fMat(\ZZ^\dd)$;
    by construction $\alpha(x)$ is supported on $\Supp(x)^{+\ell}$,
    so it is a QCA with spread $\ell$.

    Finally, we note that $\calB$ is stably equivalent to a boundary
    algebra of $\alpha$: indeed, 
    $\alpha$ sends
    $\calM = \fMat((\ZZ\cap (-\infty,0])\times\ZZ^{\dd-1})$
    to $\calM\otimes \calB$.
\end{proof}

\begin{lemma}
    \label{lem:trivial_alg_blend_shift}
    A QCA with a Brauer trivial boundary algebra on the positive
    first axis blends into a bilayer shift QCA along the
    first axis.
\end{lemma}

\begin{proof}
    We continue using the notation of 
    \ref{def:regions} and \ref{lem:boundary_alg_tensor_factor}.
    By compressing the lattice, we may assume
    that $\ell = 1$, so that the image of a single-site
    operator under the QCA
    $\alpha:\fMat(\ZZ^\dd,p)\to\fMat(\ZZ^\dd,p)$ is supported
    on neighboring sites (with respect to $\ell_\infty$-distance
    on $\ZZ^\dd$).
    The assumption $\ell = 1$ is not a restriction,
    as we can always lump sites together until it holds, then un-lump them if desired.

    Let $\calA = \calA(n,\ell=1)$ be defined as in
    \ref{lem:boundary_algebra_commutant},
    which is the fermionic commutant of~$\calB = \calB(n,\ell=1)$ within
    $\fMat(S(n,\ell=1)) = \fMat(\{n,n+1\}\times \ZZ^{\dd-1})$.
    We have $\calA \otimes \calB \cong \fMat(\ZZ^{\dd-1})$ by~\ref{cor:invertible_isomorphism}.
    Since $\calB$ is by assumption Brauer trivial,
    we have a stable isomorphism $\calA \otimes \calB \sim \calA$, 
    which in turn gives a bounded-spread isomorphism of fermionic algebras
    \begin{equation}
        \phi: \fMat(\ZZ^{\dd-1},q)\xrightarrow{\quad \cong \quad}
        \calA(n,\ell=1)\otimes \fMat(\ZZ^{\dd-1},p')
    \end{equation}
    for some local assignments $q$ and $p'$.

    Now we define two monolayer shift QCAs $\beta_1$ and $\beta_2$.
    Since we are distinguishing the first axis,
    we will use $s=(x,y)\in \ZZ^\dd$ to denote a lattice site,
    where $x\in \ZZ$ and $y\in \ZZ^{\dd-1}$.
    A local assignment $p$
    contains the data of the isomorphism class
    of the central fermionic algebra at each site;
    for a site $s\in \ZZ^\dd$, let $\fMat(p(s))$ denote
    this local algebra.
    \begin{itemize}
        \item For $\beta_1$, define the local operator algebra 
        by $\fMat(\{(x,y)\},p_1(x,y)) = \fMat(q(y))$,
        copying $\fMat(\ZZ^{\dd-1},q)$ at each site on the first axis.
        The action of $\beta_1$ is given by translation
        to the left by one unit along the first axis.
        
        \item For $\beta_2$, define the local operator algebra 
        by $\fMat(\{(x,y)\},p_2(x,y)) =
        \fMat(p(n+1,y))\otimes \fMat(p'(y))$.
        The action of $\beta_2$ is given by translation to the
        right by one unit along the first axis.
    \end{itemize}
    The bilayer shift QCA is $\beta = \beta_1\otimes \beta_2$.

    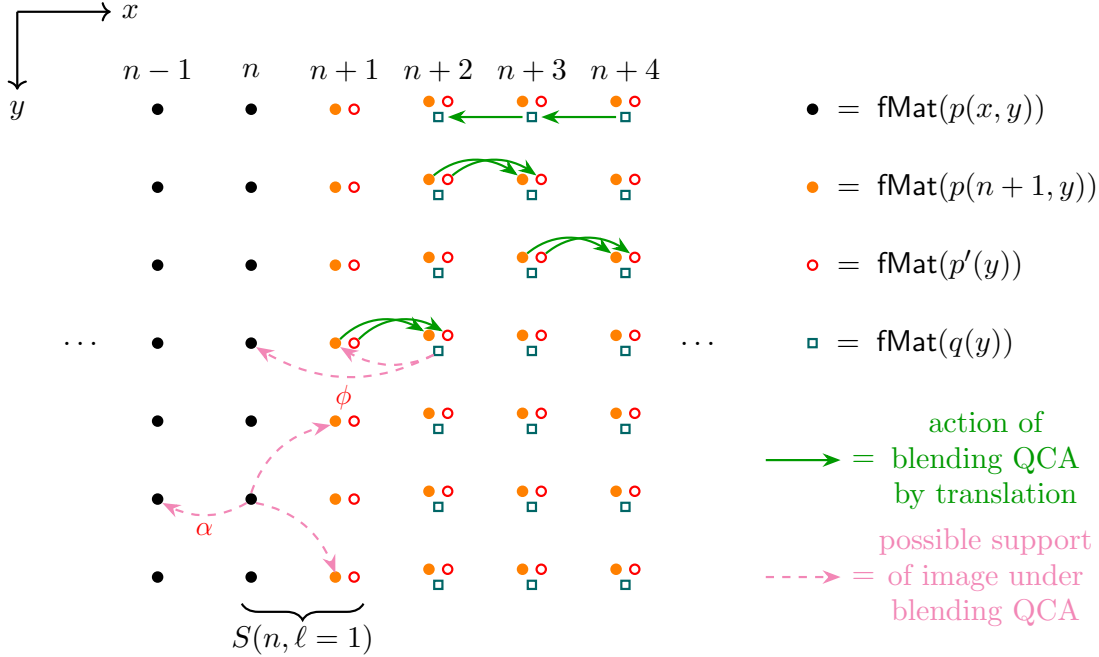
\begin{figure}[b]
        \centering
        \begingroup
            \resizebox{0.9\textwidth}{!}{\begin{tikzpicture}[
    x=1.2cm, y=-1cm, 
    dot/.style={circle, draw=black, thick, fill=black, inner sep=1.2pt},
    orange_dot/.style={circle, draw=orange, thick, fill=orange, inner sep=1.2pt},
    red_circle/.style={circle, draw=red, thick, inner sep=1.2pt},
    green_sq/.style={rectangle, draw=teal!80!black, thick, inner sep=1.5pt},
    action_arrow/.style={-Stealth, green!60!black, bend left=40, thick},
    straight_arrow/.style={-Stealth, green!60!black, thick},
    support_arrow/.style={-Stealth, magenta!60, dashed, bend left=30, thick}
]

    \draw[thick, ->] (-0.5, -1.25) -- (0.5, -1.25) node[right] {$x$};
    \draw[thick, ->] (-0.5, -1.25) -- (-0.5, -0.25) node[below] {$y$};

    \foreach \col [count=\i] in {n-1, n, n+1, n+2, n+3, n+4} {
        \node at (\i, -0.5) {$\col$};
    }
    \node at (0.2, 3) {$\dots$};
    \node at (6.8, 3) {$\dots$};

    \foreach \y in {0,...,6} {
        \node[dot] at (1, \y) {};
        \node[dot] at (2, \y) {};

        \foreach \x in {3,...,6} {
            \begin{scope}[shift={(\x, \y)}]
                \ifnum\x=3
                    \node[red_circle] at (0.1, 0) {};
                    \node[orange_dot] at (-0.1, 0) {};
                \fi
                \ifnum\x>3
                    \node[red_circle] at (0.1, -0.1) {};
                    \node[orange_dot] at (-0.1, -0.1) {};
                    \node[green_sq] at (0, 0.1) {};
                \fi
            \end{scope}
        }
    }

    \draw[straight_arrow] (5.9, 0.1) to (5.1, 0.1);
    \draw[straight_arrow] (4.9, 0.1) to (4.1, 0.1);
    \draw[action_arrow] (3.95, 0.85) to (4.85, 0.85);
    \draw[action_arrow] (4.15, 0.85) to (5.05, 0.85);
    \draw[action_arrow] (4.95, 1.85) to (5.85, 1.85);
    \draw[action_arrow] (5.15, 1.85) to (6.05, 1.85);
    \draw[action_arrow] (2.95, 2.95) to (3.85, 2.85);
    \draw[action_arrow] (3.15, 2.95) to (4.05, 2.85);

    \draw[support_arrow] (3.95, 3.15) to (2.95, 3.05);
    \draw[support_arrow] (3.95, 3.15) to node[midway, below, scale=0.7, red] {\Large $\phi$} (2.05, 3.05);
    \draw[support_arrow] (1.95, 5.05) to node[midway, below, scale=0.7, red]{\Large $\alpha$} (1.05, 5.05);
    \draw[support_arrow] (2, 4.92) to (2.85, 4.05);
    \draw[support_arrow] (2.05, 5.05) to  (2.9, 5.95);

    \draw [decorate, decoration={brace, amplitude=6pt, mirror, raise=4pt}, thick]
        (1.9, 6.2) -- (3.2, 6.2) node [black, midway, yshift=-0.6cm] {$S(n, \ell=1)$};

    \begin{scope}[shift={(8, 0)}]
        \node[dot, label=right:{$\;=\; \mathsf{fMat}(p(x,y))$}] at (0,0) {};
        \node[orange_dot, label=right:{$\; = \;\mathsf{fMat}(p(n+1,y))$}] at (0,1) {};
        \node[red_circle, label=right:{$\; = \;\mathsf{fMat}(p'(y))$}] at (0,2) {};
        \node[green_sq, label=right:{$\; = \;\mathsf{fMat}(q(y))$}] at (0,3) {};
        
        \draw[action_arrow, bend right=0] (-0.5, 4.5) -- (0.3, 4.5) 
            node[right, green!60!black] {$=$};
        \node[green!60!black, text width=3cm, align=center] at (1.85,4.5)
            {action of blending QCA by translation};
        
        \draw[support_arrow, bend left=0] (-0.5, 6) -- (0.3, 6) 
            node[right, magenta!60]{$=$};
        \node[magenta!60, text width=3cm, align=center] at (1.85,6)
            {possible support of image under blending QCA};
    \end{scope}

\end{tikzpicture}}
        \endgroup
        \caption{Blending QCA~$\gamma$ from~$\alpha$ on the left with a trivial Brauer boundary algebra~$\calB$
        into bilayer shift QCA~$\beta = \beta_1 \otimes \beta_2$ on the right.
        The construction is transcription of the following line of thought.
        The green squares are mapped according to~$\phi$ onto the fermionic commutant~$\calA$ of~$\calB$,
        where $\phi$ is provided by the Brauer triviality.
        Then, $\phi$  take up a part of the codomain
        so any other potential image must live elsewhere.
        In particular, the on-site algebras at $x = n,n+1$ must map out of that region,
        so we send them out to the right.
        This incurs cascading surplus and deficit of local algebras,
        leading to a chain of local shifts, resulting in~$\beta_1$ and~$\beta_2$.
        }
        \label{fig:blending_QCA}
    \end{figure}
    
    Finally, we construct a QCA
    $\gamma:\fMat(\ZZ^\dd, r)\to \fMat(\ZZ^\dd,r)$
    blending $\alpha$ into $\beta$.
    We define the local operator algebra by
    \begin{equation}
        \fMat(r(x,y)) = \begin{cases*}
            \fMat(p(x,y)) & $x \le n$ \\
            \fMat(p(n+1,y))\otimes \fMat(p'(y)) & $x = n+1$ \\
            \fMat(p(n+1,y))\otimes \fMat(p'(y))
            \otimes \fMat(q(y)) & $x \ge n+2$,
        \end{cases*}
    \end{equation}
    as shown in Fig.~\ref{fig:blending_QCA}.
    The action of $\gamma$ is determined by the action
    on each of these tensor factors:
    \begin{itemize}
        \item For $x\le n$ (black dots in Fig.~\ref{fig:blending_QCA}),
        $\gamma$ follows the action of $\alpha$.
        The image of any operator supported
        on $x\le n$ is supported on $x\le n$
        and the $\fMat(p(n+1,y))$ component of $x = n+1$.
        \item $\gamma$ translates
        the copies of $\fMat(p(n+1,y))\otimes \fMat(p'(y))$
        for $x\ge n+1$ (orange and red)
        to the right by one lattice site.
        \item $\gamma$ translates the copies of
        $\fMat(q(y))$ for $x\ge n+3$ (green) to the left
        by one lattice site.
        \item Finally, $\gamma$ sends each $\fMat(q(y))$
        at $x=n+2$ (green) to its image under $\phi$,
        supported on $x \in \{n, n+1\}$.
    \end{itemize}
    The images of each of these tensor factors fermionically
    commute, so by \ref{lem:TensorProdUniversalProperty}
    this prescription defines a homomorphism of fermionic algebras
    $\fMat(\ZZ^\dd,r)\to \fMat(\ZZ^\dd,r)$.
    Since $\fMat(\ZZ^\dd,r)$ is simple, $\gamma$ is automatically
    injective. All that remains is to show surjectivity.
    
    Every operator $o$ supported on $x \le n-1$ is covered by
    $\gamma(\alpha^{-1}(o))$, since $\alpha^{-1}(o)$ is supported
    on $x \le n$ (using $\ell = 1$).
    All operators supported on $x\ge n+2$ are in the image
    of the translations.
    Finally, looking at the slab $S$, we have
    \begin{equation}
        \fMat(x = \{n,n+1\},r) = \calB(n,\ell=1)\calA(n,\ell=1)\otimes
        \fMat(\ZZ^{\dd-1},p').
    \end{equation}
    By definition $\calB$ is in the image of~$\alpha$ applied to
    $x\le n$, while $\calA\otimes \fMat(\ZZ^{\dd-1},p')$
    is in the image of~$\phi$.
    Hence $\gamma$ is an isomorphism.
    It is bounded-spread by construction, 
    so it is a valid blending QCA.
\end{proof}

\begin{proof}[Proof of Theorem~\ref{thm:QCAToBrauer}]
    Let $G_{\text{QCA}}$ be the abelian group of blending
    equivalence classes of QCAs (along a fixed axis)
    from \ref{cor:QCA_group}, and let $G_{\text{Br}}$
    be the Brauer group in $\dd-1$ dimensions.
    Consider the map $\phi$ taking a QCA in $G_{\text{QCA}}$
    to the Brauer class of its boundary algebra;
    this is well-defined by \ref{lem:boundary_alg_well_defined},
    and it clearly preserves the tensor product operation.
    Furthermore, if $\gamma$ is a QCA which blends
    the QCAs $\alpha$ and $\beta$,
    then we see
    by varying the location of the boundary on $\gamma$
    (\ref{lem:boundary_alg_well_defined})
    that
    the boundary algebra of $\alpha$ is stably equivalent
    to the boundary algebra of $\beta$.
    Hence $\phi: G_{\text{QCA}}\to G_{\text{Br}}$
    is a group homomorphism.

    The boundary algebras of shift QCAs are clearly Brauer
    trivial, and by \ref{lem:trivial_alg_blend_shift}
    any QCA with trivial boundary algebra blends with a
    bilayer shift QCA.
    Hence $\ker \phi$ is exactly the subgroup of shift QCAs.
    Finally, $\phi$ is surjective because of~\ref{lem:alg_to_QCA},
    so $G_{\text{Br}}$ is isomorphic to $G_{\text{QCA}}$
    modulo shifts.
\end{proof}

\subsection{Structure of fermionic QCAs in low dimensions}

In this subsection, we prove that after stabilization,
every fermionic QCA is trivial in dimension~$1$ and~$2$:

\begin{theorem} \label{thm:fQCAStructure1d2d}
    Suppose $\dd = 1$ or $\dd=2$.
    Let $m : \ZZ^\dd \ni s \mapsto \FermionMinus{1}$ be a uniform local assignment 
    of one Majorana algebra on every site.
    For every fermionic QCA $\alpha$ 
    on a local operator algebra $\fMat(\ZZ^\dd,p)$, 
    the stabilized QCA
    $(\alpha \otimes \id) : \fMat(\ZZ, p m) \to \fMat(\ZZ, p m)$
    is a circuit of local automorphisms followed by a primitive shift. 
    
    If a local operator algebra $\fMat(\ZZ^\dd,p)$ 
    contains a Majorana operator on every site,
    the stabilization by~$m$ is not necessary.
\end{theorem}

To prove this,
we will adopt and extend the argument in~\cite[Appendix~A]{clifqca1}
to the fermionic setting.

\begin{proposition}
    \label{lem:MatchPrimitiveAlgebras}
    Let $\Lambda$ be a possibly infinite set of sites,
    where $\fMat(\{s\},q(s))$ is primitive for each $s\in \Lambda$.
    Suppose that there are primitive subalgebras~$\calD_i$ of~$\fMat(\Lambda,q)$,
    that are mutually fermionically commuting.
    If the $\calD_i$'s generate the full algebra~$\fMat(\Lambda,q)$,
    then there is an algebra automorphism~$\beta$ of~$\fMat(\Lambda,q)$
    such that
    for every $i$, $\calD_i = \beta(\fMat(\{s_i\}))$ 
    for some $s_i \in \Supp(\calD_i)$.
\end{proposition}

This is a generalization of~\cite[Lemma A.3]{clifqca1} 
to the fermionic setting
and also to an infinite number of sites.
The assumption that $\calD_i$ is finite dimensional
implies that each $\calD_i$ is finitely supported
since $\fMat(\Lambda,q)$ consists of finitely supported operators.

\begin{proof}
    We use Marshall Hall, Jr.'s marriage theorem~\cite[Theorem 5.1.2]{MHallJrBook1988}:%
    \emph{
        for a possibly infinite bipartite graph
        where every left node has finitely many neighbors on the right,
        if for any~$n \ge 0$ every finite set of $n$ left nodes
        has at least $n$ distinct neighbors on the right,
        then there exists a subset of edges that gives an injection
        from the left nodes to the right.
    }
    
    Consider a bipartite graph $G = (A \sqcup B, E)$
    with vertex sets $A = \{\calD_i\}$
    and $B = \{s\in \Lambda\}$,
    and an edge $(\calD_i, s) \in E$ 
    between~$\calD_i$ and site~$s$ 
    if and only if 
    $\calD_i \cong \fMat(\{s\})$ and $s\in \Supp(\calD_i)$.
    The statement of the lemma is equivalent to finding
    a perfect matching in~$G$; given a matching, we can
    take $\beta$ to be a map sending each $\fMat(\{s\})$
    to its corresponding~$\calD_i$.
    
    Since every $\calD_i$ is finitely supported,
    every $A$-vertex has finitely many $B$-neighbors.
    Let $I_s$ be the minimal set of indices~$i$
    such that $\fMat(\{s\}) \subseteq \calD(s) = \prod_{i \in I_s} \calD_i$.
    Since $\fMat(\{s\})$ has finitely many generators,
    each of which is generated by finitely many~$\calD_i$'s by assumption,
    the index set~$I_s$ is finite.
    Any~$\calD_i$ with $i \notin I_s$ 
    fermionically commutes with $\calD(s)$ by assumption,
    and hence cannot have support on~$s$.
    Therefore, every $B$-vertex has finitely
    many $A$-neighbors.
    We claim that the marriage
    condition holds in both directions.
    Then, we have injections between~$A$ and~$B$,
    by which we obtain a perfect matching
    (the Schr\"oder--Bernstein theorem).
    
    First, take an arbitrary set 
    of $k$ primitive subalgebras
    $\calD_{i_1},\dots, \calD_{i_k}$;
    their tensor product is embedded in $\fMat(S, q)$,
    where $S = \bigcup_m \Supp(\calD_{i_m})$.
    By uniqueness in~\ref{thm:StabilizedThenDecomposed},
    the single-site algebras making up $\fMat(S,q)$
    must include $k$ primitive algebras corresponding
    to the same isomorphism classes as $\{\calD_{i_m}\}$.
    This is the marriage condition 
    in the $A \to B$ direction.
    
    For the other direction, take an arbitrary set of
    $k$ sites $S = \{s_1,\dots, s_k\}$.
    Since the $\calD_i$'s generate $\fMat(\Lambda)$,
    $\fMat(S)$ is embedded in $\prod_{i\in T} \calD_i$,
    where $T = \bigcup_{s \in S} I_s$ is the set of indices $i$ such that
    $\Supp(\calD_i)\cap S\neq \emptyset$.
    Applying the uniquness in \ref{thm:StabilizedThenDecomposed}
    again, $\{\calD_i \mid i\in T\}$ must contain
    $k$ primitive algebras corresponding to 
    $\fMat(\{s\})$ for $s\in S$.
    This proves the marriage condition in the $B \to A$ direction.
\end{proof}

Now we can decompose the action of 
a $\dd$-dimensional fermionic QCA~$\alpha$ on a slab
into a $(\dd-1)$-dimensional QCA and a shift on that strip,
whenever all $(\dd-1)$-dimensional invertible subalgebras are Brauer trivial.

\begin{lemma}
    \label{lem:Decompose1DStrip}
    Suppose $\dd =1$ or $\dd=2$.
    Let $A = [0, 1000\ell]\times \ZZ^{\dd-1}$ be a vertical strip of the lattice.
    Let $p' : A^{+\ell} \ni s \mapsto \FermionMinus{1}^{\otimes 2}$
    be a uniform local assignment.
    Then, there exists a $(\dd-1)$-dimensional fermionic QCA
    $\beta: \fMat(A^{+\ell}, pp')\to\fMat(A^{+\ell},pp')$
    with spread at most $50\ell$ such that:
    for any $s\in A$, $\fMat(\{s\},pp'(s))$ can be decomposed
    into a tensor product of primitive fermionic algebras where
    the image of each factor under
    $\beta^{-1}\circ (\alpha\otimes \id)$
    is supported on a single site.
    In other words, $\beta$ agrees with
    $\alpha\otimes \id$ on $A$ up to a shift.
\end{lemma}

This lemma assumes $\dd =1,2$
because we will use the fact that the Brauer group in dimension~$\dd-1$
is trivial.

\begin{proof}
    By the above results for boundary algebras,
    we have the decomposition
    \begin{equation}
        \FermiCommutant{\alpha(\fMat(A,p))}{\fMat(A^{+\ell},p)}
        = \calA_L\otimes \calA_R,
    \end{equation}
    where $\calA_L$, $\calA_R$ are $(\dd-1)$-dimensional invertible subalgebras
    such that $\Supp(\calA_L)\subseteq [-\ell,\ell]\times \ZZ^{\dd-1}$ and
    $\Supp(\calA_R)\subseteq [999\ell,1001\ell]\times \ZZ^{\dd-1}$.
    These algebras are the fermionic commutants of the boundary algebras
    on the left and right
    (see \ref{lem:boundary_algebra_commutant}).
    
    Let $\fMat(A^{+\ell},p_1)$ be the algebra with $\FermionMinus{1}$
    at each site of~$A^{+\ell}$.
    If we use this algebra as stabilization,
    then by \ref{thm:1DBrauerGroupTrivial} if $\dd=2$
    or by definition if~$\dd=1$,
    we know that $\calA_L$ and $\calA_R$ factorize into a tensor product
    of $50\ell$-local central fermionic algebras.
    Furthermore, the image $(\alpha\otimes \id)(\fMat(\{s\},pp_1(s)))$
    can be written as a tensor product of primitive algebras for each $s\in A$
    (by~\ref{thm:StabilizedThenDecomposed}).
    Hence,
    \begin{equation}
        \fMat(A^{+\ell}, pp_1) \cong \calA_L \otimes \alpha(\fMat(A,p))
        \otimes \calA_R \otimes \fMat(A^{+\ell}, p_1)
        \cong \bigotimes_j \calC_j \otimes \bigotimes_{s\in A} \calD_s, \label{eq:stripDecomposition}
    \end{equation}
    where the $\calC_j$ are local central fermionic algebras
    coming from the stabilized $\calA_L$ and $\calA_R$,
    and $\calD_s = (\alpha\otimes \id)(\fMat(\{s\},pp_1(s))$ with~$s \in A$.
    
    We wish to ``disentangle'' all these local central fermionic algebras
    via~\ref{lem:MatchPrimitiveAlgebras}.
    The only requirement for~\ref{lem:MatchPrimitiveAlgebras}
    that we do not immediately have in~\ref{eq:stripDecomposition} 
    is the condition that the local central algebras be primitive.
    This condition is easily met by stabilizing with
    another product of one-Majorana algebras
    $\fMat(A^{+\ell},p_2) \cong \fMat(A^{+\ell}, p_1)$.\footnote{Closer examination 
        of the proof of~\ref{thm:1DBrauerGroupTrivial} tells us that 
        the second stabilization~$p_2$ is unnecessary,
        but we quote~\ref{thm:1DBrauerGroupTrivial} in a blackbox manner.
    }
    Decomposing into primitive factors, 
    we do not worsen the locality of the central fermionic algebras in~\ref{eq:stripDecomposition}.
    We then apply~\ref{lem:MatchPrimitiveAlgebras} to 
    the collection of primitive factors, denoted by~$\calD'_a$, 
    of stabilized~$\calD_s$ and~$\calC_j$
    viewed as subalgebras of the tensor product of all the primitive factors of $p p_1 p_2$,
    to obtain an
    automorphism~$\beta$ of~$\fMat(A^{+\ell},p p_1 p_2)$
    with spread at most~$50\ell$.
    
    By construction, $\alpha \otimes \id$ maps
    a primitive factor of~$p p_1 p_2$ 
    to a primitive factor~$\calD'_a$.
    Under~$\beta^{-1}$, this primitive factor~$\calD'_a$ is sent to a single site.
    Hence, $\beta$ agrees with $\alpha\otimes\id$ up to a shift.
\end{proof}

We extend the result on slabs to the entire $\dd$-dimensional lattice.
This will be almost the structure theorem~\ref{thm:fQCAStructure1d2d}.

\begin{lemma}\label{lem:fQCAStructure1d2dWithBiggerAncilla}
    Theorem~\ref{thm:fQCAStructure1d2d} is true
    with $p' : \ZZ^\dd \ni s \mapsto \FermionMinus{1}^{\otimes 10}$
    in place of~$m: \ZZ^\dd \ni s \mapsto \FermionMinus{1}$.
\end{lemma}

\begin{proof}
    Divide the plane into vertical strips
    $A_i = [2000\ell i, (2000i+1000)\ell]\times \ZZ^{\dd-1}$
    and $B_i = ((2000i+1000)\ell, (2000i+2000)\ell)\times \ZZ^{\dd-1}$.
    Applying \ref{lem:Decompose1DStrip} to each $A_i$,
    we obtain a $(\dd-1)$-dimensional QCA $\beta_i$ acting on a stabilized
    algebra in a slightly wider strip $A_i^{+\ell}$.
    Let $\beta = \prod_i \beta_i$, and suppose
    $\beta$ acts on the algebra $\fMat(\ZZ^\dd, pq_1)$
    where $q_1$ denotes the stabilization used in~\ref{lem:Decompose1DStrip}.
    By construction,
    \begin{equation}
        \eta\coloneq\beta^{-1}\circ(\alpha\otimes\id)
    \end{equation}
    acts by shift on each $A_i$,
    and has spread at most $51\ell$.
    The key property%
    \footnote{In the language of~\cite{clifqca1}, $\eta(\fMat(A_i))$ has ``sharp'' support.}
    of $\eta$ is that it sends primitive algebras
    in the decomposition of $\fMat(\{s\},pq_1)$ for $s\in A_i$ to
    primitive algebras on single sites in $A_i^{+\ell}$.

    Let $q_2$ be the local assignment with
    $\FermionMinus{1}$ at each site in the $B_i$'s,
    so that $\fMat(\{s\},pq_1q_2)$ can be decomposed into
    primitive algebras for all $s\in B_i$.
    Now focus on a single $B_i$, and consider
    \begin{equation}  \begin{aligned}
    \calB_i &= (\eta\otimes\id)(\fMat(B_i, pq_1q_2))\\
        &= \FermiCommutant{(\eta\otimes\id)
        (\fMat(A_i\cup A_{i+1},pq_1q_2)}
        {~~\fMat(B_i^{+51\ell}, pq_1q_2)}.
    \end{aligned}\end{equation}
    By the key property above, the image of
    $\fMat(A_i\cup A_{i+1})$ under~$\eta$ 
    is the tensor product of
    primitive algebras on single sites.
    Hence the commutant is a tensor product of a
    central fermionic algebras on each $s\in B_i^{+51\ell}$.
    Then if $q_3$ is the local assignment with $\FermionMinus{1}$
    at each site of $B_i^{+51\ell}$, then
    $\calB_i \otimes \fMat(B_i^{+51\ell},q_3)$
    can be decomposed as the tensor product of primitive
    algebras on each site.
    
    On the other hand, we have $\calB_i = \bigotimes_j \calD_j$,
    where the $\calD_j$'s are the images under~$\eta$ of primitive algebras
    in the decompositions of the single-site algebras
    $\fMat(\{s\},pq_1q_2)$ for $s\in B_i$.
    Thus, we can apply \ref{lem:MatchPrimitiveAlgebras}
    to obtain a bounded-spread automorphism $\gamma_i$ of
    $\calB_i\otimes \fMat(B_i^{+51\ell},q_3)$
    such that $\gamma_i$ agrees with $\eta\otimes\id$
    on $B_i$ up to a primitive shift.
    
    Letting $\gamma = \prod_i \gamma_i$, we see that
    $\gamma^{-1}\circ (\beta^{-1}\otimes \id)\circ (\alpha\otimes\id)$
    is an overall primitive shift on $\fMat(\ZZ^\dd, p p')$,
    where $p'$ includes all above stabilizations.
    In other words, we have decomposed $\alpha\otimes\id$
    as the product of two layers of disjoint $(\dd-1)$-dimensional QCAs
    and a primitive shift.
    
    To finish, we note that
    each $\beta_i$ and $\gamma_i$ is a $(\dd-1)$-dimensional fermionic QCA.
    If $\dd = 1$, then these are local automorphisms,
    and we conclude the proof for $\dd=1$ case.
    If $\dd=2$, since we have just proved the theorem in $\dd=1$ case,
    we decompose $\beta_i$ and $\gamma_i$ into a circuit composed with a primitive shift.
    The local automorphisms in this decomposition
    may have to be further decomposed,
    in which case the depth of the circuit will scale with the width
    of the strips, but this is finite bounded only by~$\ell$.
    Thus, the two-dimensional QCA $\alpha\otimes\id$ is the composition
    of a circuit and a primitive shift.
\end{proof}

\begin{proof}[Proof of the structure Theorem \ref{thm:fQCAStructure1d2d}]
    The final structure theorem differs from~\ref{lem:fQCAStructure1d2dWithBiggerAncilla}
    in that the ancilla system is $m^{10}$ instead of just~$m$.
    Let $\zeta$ denote the identity QCA on~$\fMat(\ZZ^\dd, m)$.
    From~\ref{lem:fQCAStructure1d2dWithBiggerAncilla}
    we have $\alpha \otimes \zeta^{\otimes 10} = \delta \circ \beta$
    where $\delta$ is a circuit and $\beta$ is a primitive shift.
    With~\ref{lem:AncillaReduction} inductively applied,
    we find a circuit~$\delta'$
    such that $\beta = \delta' \circ (\beta' \otimes \zeta^{\otimes 9})$
    where $\beta'$ is a primitive shift on $\fMat(\ZZ^\dd, pm)$.
    So, $\alpha \otimes \zeta^{\otimes 10} = \delta\circ \delta' \circ (\beta' \otimes \zeta^{\otimes 9})$
    and hence
    $((\alpha \otimes \zeta)\circ (\beta')^{-1}) \otimes \zeta^{\otimes 9} = \delta \circ \delta'$
    is a circuit that acts by identity on $\fMat(\ZZ^\dd, m^9)$.
    The ancillla removal~\ref{lem:AncillaRemoval}
    gives an equivalent circuit~$\delta''$ where all individual local automorphisms act
    by identity on~$\fMat(\ZZ^\dd, m^9)$.

    If $p$ already contains a Majorana operator on every site,
    then we may write $p = qm$, and we have a circuit~$\delta''$ and a primitive shift~$\beta'$
    acting on $\fMat(\ZZ^\dd, pm) = \fMat(\ZZ^\dd, qm^2)$.
    By~\ref{lem:AncillaReduction} we replace~$\beta'$
    with another primitive shift~$\beta''$ acting on~$\fMat(\ZZ^\dd, p)$
    at the cost of changing the circuit to~$\delta'''$.
    By~\ref{lem:AncillaRemoval} we remove one copy of~$m$ from~$\delta'''$.
\end{proof}

\subsection{Index theory in one dimension}
\label{sec:1dQCAIndex}

Here, we reproduce the well-known index theory of $1\dd$ QCA~\cite{GNVW,fermionGNVW2}
for completeness.

\begin{definition}
    Let $\alpha$ be a QCA acting on $\fMat(\ZZ,p)$
    with spread $\ell$, and
    let $\calB(n,\ell)$ be the boundary algebra of $\alpha$
    as defined in \ref{lem:boundary_alg_tensor_factor}.
    The \emph{index} of the QCA $\alpha$ is given by
    \begin{equation}
        \label{eq:IndexDefinition}
        \ind \alpha = \sqrt{\frac{\dim_{\CC}\calB(n,\ell)}{\dim_{\CC}\fMat([n-\ell+1,n],p)}} \, .
    \end{equation}
\end{definition}

\begin{lemma}
    \label{lem:IndexWellDefined}
    The index is well-defined: 
    it is invariant if we increase $\ell$
    or change the location~$n$ of the boundary.
\end{lemma}

\begin{proof}
    If we increase $\ell\to \ell+1$,
    the denominator in \ref{eq:IndexDefinition}
    will increase by a factor of $\dim_\CC \fMat(\{n-\ell\})$.
    By \ref{lem:BoundaryAlgebraShift}, we have
    $\calB(n,\ell+1)\cong \fMat(\{n-\ell\})\otimes \calB(n,\ell)$,
    so under $\ell\to \ell+1$
    the numerator in \ref{eq:IndexDefinition}
    also increases by a factor of $\dim_\CC \fMat(\{n-\ell\})$.
    Hence the index is invariant if we increase the
    spread $\ell$.

    Next, we consider the transformation
    $(n,\ell)\to (n+1,\ell+1)$. The denominator in
    \ref{eq:IndexDefinition} will increase by a factor of
    $\dim_\CC \fMat(\{n+1\})$. 
    By \ref{lem:BoundaryAlgebraShift}, we have
    $\calB(n+1,\ell+1)\cong \alpha(\fMat(\{n+1\}))\otimes
    \calB(n,\ell)$,
    so the numerator in \ref{eq:IndexDefinition}
    also increases by a factor of $\dim_{\CC} \fMat(\{n+1\})$
    under the transformation $(n,\ell)\to (n+1,\ell+1)$.
    Therefore, the overall index is invariant under changes
    in $n$ and $\ell$.
\end{proof}

\begin{corollary}\label{lem:blendingImpliesSameIndex}
    If the $1\dd$ fermionic QCAs $\alpha$ and $\beta$
    admit a blending, then $\ind \alpha = \ind \beta$.
\end{corollary}

\begin{proof}
    Let $\gamma$ be a QCA that blends $\alpha$ and $\beta$.
    By \ref{lem:IndexWellDefined}, we can compute $\ind \gamma$
    far in the region where $\gamma$ agrees with $\alpha$,
    or far in the region where $\gamma$ agrees with $\beta$.
    Hence $\ind \alpha = \ind \gamma = \ind \beta$.
\end{proof}

\begin{example}
    For bosonic (qudit) QCA~\cite{GNVW},
    the index is valued in~$\QQ^\times$,
    and every rational index is realized by some shift QCA.
    For fermionic QCA, the index group is enlarged
    by one extra generator~$\sqrt{2}$~\cite{fermionGNVW2}.
    This irrational index is realized by the Majorana shift QCA:
    assign $\FermionPlus{1}{1} = \langle \gamma_{2s-1}, \gamma_{2s}\rangle$
    on every site~$s$ of the $1\dd$ lattice~$\ZZ$,
    where each $\gamma_j$ is an odd self-adjoint unitary,
    and consider a fermionic QCA $\alpha: \gamma_j \mapsto \gamma_{j+1}$ for all~$j$.
    This QCA has spread~$1$, and the boundary algebra is~$\FermionMinus{1}$.
\end{example}

\begin{lemma}
    \label{lem:IndexDefinedFromRight}
    The index can equally well be defined from the right
    half line: if $\calA(n,\ell)$ is defined as in
    \ref{lem:boundary_algebra_commutant}, then
    \begin{equation}
        \label{eq:IndexDefinitionOpposite}
        \ind \alpha = \sqrt{\frac{\dim_\CC \fMat([n+1,n+\ell],p)}{\dim_\CC \calA(n,\ell)}}
    \end{equation}
\end{lemma}

\begin{proof}
    By \ref{lem:boundary_algebra_commutant}, we have
    \begin{equation}
        \calA(n,\ell)\otimes\calB(n,\ell)\cong
        \fMat(S(n,\ell))\cong
        \fMat([n-\ell+1,n])\otimes\fMat([n+1,n+\ell]).
    \end{equation}
    Taking dimensions, we see that the expressions in
    \ref{eq:IndexDefinition} and \ref{eq:IndexDefinitionOpposite}
    are equal.
\end{proof}

\begin{proposition}
    \label{lem:SameIndexQCABlend}
    Any two $1\dd$ fermionic QCAs $\alpha$ and $\beta$
    with the same index admit a blending.
    If $\alpha$ and $\beta$ both have spread at most $\ell$,
    then the blending QCA $\gamma$ can be chosen to have
    spread at most $2\ell$.
    If $\alpha$ and $\beta$ are on the same algebra~$\fMat(\ZZ,p)$,
    then a blending QCA~$\gamma$ exists on~$\fMat(\ZZ,p) \otimes \FermionMinus{1}_0$
    where the stabilization is at the origin site.
\end{proposition}

Hence the index completely classifies $1\dd$ fermionic QCAs
up to blending equivalence.

\begin{proof}
    Suppose $\alpha$ acts on $\fMat(\ZZ, p)$ and
    $\beta$ acts on $\fMat(\ZZ,q)$.
    Let $\calB_\alpha(n,\ell)$ and $\calB_\beta(n,\ell)$
    be the boundary algebras for the QCAs $\alpha$ and $\beta$, respectively.
    Define $\calA_\alpha(n,\ell)$ and $\calA_\beta(n,\ell)$
    as in \ref{lem:boundary_algebra_commutant}.
    By \ref{lem:IndexDefinedFromRight}, the condition $\ind \alpha = \ind \beta$
    implies that (choosing the boundary at $n = 0$)
    \begin{equation}
        \dim_\CC \calB_\alpha(0,\ell)\cdot \dim_\CC \calA_\beta(0,\ell)
        = \dim_\CC\fMat([-\ell+1,0],p)\cdot
        \dim_\CC\fMat([1,\ell],q).
    \end{equation}
    Stabilizing by~$\FermionMinus{1}$,
    we have from~\ref{thm:StabilizedThenDecomposed} a fermionic algebra isomorphism
    \begin{equation}
        \phi: \calB_\alpha(0,\ell)\otimes \calA_\beta(0,\ell)
        \otimes \FermionMinus{1} \xrightarrow{\;\sim\;} 
        \fMat([-\ell+1,0],p)\otimes 
        \fMat([1,\ell],q)\otimes \FermionMinus{1}.
    \end{equation}
    To construct the blending QCA,
    consider the algebra isomorphism
    \begin{equation}
        \begin{aligned}
            \delta: \fMat(\ZZ,r) &= \fMat((-\infty, 0],p) \otimes\FermionMinus{1}\otimes \fMat([1,\infty),q) \\
            & \xrightarrow{\;\sim\;}
            \alpha^{-1}(\fMat((-\infty,-\ell])) \otimes
            \alpha^{-1}(\calB_\alpha(0,\ell))
            \otimes \FermionMinus{1} \\
            &\qquad \qquad \otimes
            \beta^{-1}(\calA_\beta(0,\ell))\otimes 
            \beta^{-1}(\fMat([\ell+1,\infty))),
        \end{aligned}
    \end{equation}
    where the ancillary algebra~$\FermionMinus{1}$ 
    is located at the origin of the lattice~$\ZZ$.
    We define a blending QCA acting on $\fMat(\ZZ,r)$
    as a composition $\gamma = \eta \circ \delta$, 
    where $\eta$ is defined by sending pure tensors
    \begin{equation}
        \eta: x\otimes x'\otimes w\otimes y'\otimes y
        \mapsto \alpha(x)\otimes \phi\Big(\alpha(x')\otimes \beta(y')\otimes w\Big) \otimes \beta(y).
    \end{equation}
    The image of $\eta$ lives in the algebra
    \begin{equation}
    \begin{aligned}
        \fMat(\ZZ,r) &\cong 
        \fMat((-\infty,-\ell])\otimes \Big(\fMat([-\ell+1,0],p) \\
        &\qquad \otimes \fMat([1,\ell],q)\otimes \FermionMinus{1}\Big) 
        \otimes \fMat([\ell+1,\infty))\,.
    \end{aligned}
    \end{equation}
    By \ref{lem:TensorProdUniversalProperty}, $\eta$ is
    a homomorphism of fermionic algebras.
    It is by definition injective and surjective.
    The composition~$\eta \circ \delta$ 
    manifestly has spread at most $2\ell$.
    Thus $\gamma = \eta \circ \delta$ is the desired blending QCA.
\end{proof}

\begin{proposition}
    \label{lem:CompositionQCAIndexAdds}
    For two QCAs $\alpha$ and $\beta$
    acting on $\fMat(\ZZ,p)$, we have
    $\ind(\alpha\beta) = \ind(\alpha \otimes \beta) = \ind(\alpha) \ind(\beta)$.
\end{proposition}

\begin{proof}
    It is clear from the definition of the index that
    $\ind(\alpha\otimes\beta) = \ind(\alpha) \ind(\beta)$.
    By~\ref{thm:QCAQCAInverseIsCircuit}, $\beta\otimes\beta^{-1}$
    is a circuit of local automorphisms.
    By dropping some of the local automorphisms, 
    we see that $\beta\otimes\beta^{-1}$ blends with the identity.
    Therefore, $(\alpha \beta \otimes \id)(\beta^{-1} \otimes \beta) = \alpha \otimes \beta$
    blends with $\alpha \beta \otimes \id$.
    Using the invariance of the index under blending (\ref{lem:blendingImpliesSameIndex}),
    we conclude the proof.
\end{proof}

\begin{remark}
    It is therefore appropriate to think of the group of all indices of fermionic QCA
    as the Grothendieck group of the rank tuples.
\end{remark}

\section{Nonuniform local algebras}

In prior sections, we frequently had to stabilize by tensoring with small fermionic algebras.
However, we only ever introduced a bounded density of new degrees of freedom.
Whenever the original algebra enjoyed a uniform bound on the local algebra dimension, so too will the stabilized algebra.

In this section, we explore the additional power granted by tensoring with fermionic algebras with no uniform bound on the local dimension.
A notable consequence of this is that all shift QCAs on~$\ZZ^\dd$
can be expressed as circuits of local automorphisms in any dimension \(\dd \geq 2\).
In particular, all QCAs in \(\dd=2\) are circuits.
Indeed, this follows from the main result of this section:

\begin{theorem}[QCA with a hole is a circuit]\label{thm:QCAWithAHoleIsTrivial}
    If \(\alpha : \fMat(\ZZ^\dd,p) \to \fMat(\ZZ^\dd,p)\)
    is a \mbox{spread-\(\ell\)} QCA
    which acts trivially in a ball \(B \subseteq \ZZ^\dd\) of radius \(r \geq 2^{4\dd+5} \ell\),
    then there is a (generally nonuniform) local assignment \(p'\)
    such that \(\alpha\otimes \id : \fMat(\ZZ^\dd,p p') \to \fMat(\ZZ^\dd,p p')\)
    is a circuit of local automorphisms.
\end{theorem}

Several immediate corollaries with nonuniform stabilization allowed
are explored below:
the circuit class of any QCA is locally computable~(\ref{cor:CircuitClassIsLocallyComputable});
any shift QCA in $\dd \ge 2$ is a circuit~(\ref{thm:ShiftsTrivialForDdGeq2});
and
any $\dd =2$ fermionic QCA is a circuit~(\ref{thm:2DfQCAAreTrivialWithNonuniformStabilization}).

Two subsequent subsections constitute the proof of~\ref{thm:QCAWithAHoleIsTrivial}.
First,
we establish a range-reduction lemma:
the spread of a QCA can be reduced by the application of a finite depth circuit
and the introduction of ancillas,
a procedure known as \emph{pleating} and \emph{ironing}
which we adapt from~\cite{FHH2019}.
Next,
we set up an inductive procedure
by which the ``hole'' in \(\alpha\) posited by \ref{thm:QCAWithAHoleIsTrivial}
can be made larger and larger by the application of finite depth circuits,
in a process we call \emph{hemming}.
Hemming increases the spread of the QCA,
but interleaving with pleating and ironing results
in a circuit expression for \(\alpha \otimes \id\).

\begin{remark}
The fermionic grading plays no particular role in this section.
Results are phrased for fermionic QCAs, but all calculations and arguments
are semantically identical to the ungraded case.
Furthermore, when applied to Clifford QCA~\cite{clifqca1,clifQCA},
this shows that the translation-invariance assumption
in the classification result~\cite{clifQCAclassification}
may be dropped.
\end{remark}

\begin{corollary}[Circuit/Brauer class is locally computable]
    \label{cor:CircuitClassIsLocallyComputable}
    Let \(\alpha : \fMat(\ZZ^\dd,p) \to \fMat(\ZZ^\dd,p)\) be a spread-\(\ell\) QCA.
    Then determining the finite depth circuit equivalence class
    (with nonuniform ancilla dimensions in space) of~\(\alpha\)
    only requires knowing \(\alpha|_B : \fMat(B,p) \hookrightarrow \fMat(B^{+\ell},p)\),
    the restriction of~\(\alpha\) to any ball \(B \subseteq \ZZ^\dd\)
    of radius \(r \geq 2^{4\dd+5} \ell \).

    Similarly, if \(\calA\) is an invertible subalgebra of spread \(\ell\),
    then determining the Brauer class of \(\calA\)
    only requires knowledge of \(\calA_B = \calA \cap \fMat(B)\).
\end{corollary}

This means that the finite depth circuit equivalence class of a QCA
and the Brauer class of an invertible subalgebra
are both \emph{locally computable}~\cite{GNVW}.
That is, the class can be determined, in principle,
just from the restriction of the QCA or invertible subalgebra
to a sufficiently large ball.

\begin{proof}
    Given \(\alpha|_B\), arbitrarily extend it to a QCA \(\tilde{\alpha} : \fMat(\ZZ^\dd,q) \to \fMat(\ZZ^\dd,p)\), where \(\tilde{\alpha}|_B = \alpha|_B\).
    This is possible as \(\alpha|_B\) was originally found from truncating a QCA, so there is at least one way to make the extension, namely as \(\tilde{\alpha} = \alpha\).
    Then \ref{thm:QCAWithAHoleIsTrivial} implies that \(\tilde{\alpha}\) and \(\alpha\) are in the same circuit class:
    add ancillas \(q'\) to \(q\) and \(p'\) to \(p\) such that \(q q' = p p'\), and observe that \((\alpha^{-1} \otimes \id)(\tilde{\alpha}\otimes \id)\) acts trivially on \(B\), and thus is a circuit after stabilization.
    We thus learn the circuit class of \(\alpha\) by computing that of \(\tilde{\alpha}\).

    For an invertible subalgebra, similarly extend \(\calA_B\) to a full \(\tilde{\calA}\), and then use that the Brauer class of \(\tilde{\calA}\) is determined by the circuit class of the QCA \(\tilde{\alpha}\) from \ref{lem:alg_to_QCA}.
    Were we to compute the similar QCA \(\alpha\) using the actual \(\calA\) instead of our arbitrary extension, it would still agree with \(\tilde{\alpha}\) on a cylinder \(B \times \ZZ\).
    Thus, we reduce to the case of circuit classes of QCAs.
\end{proof}

\begin{remark}
    Combining \ref{thm:QCAWithAHoleIsTrivial}
    with the torus trick of~\cite{hastings2013classifying}
    shows that every QCA is circuit equivalent to a QCA
    which is periodic along every axis.
    With the torus trick for invertible subalgebras~(\ref{lem:TorusTrick}),
    we also see that every invertible subalgebra is circuit equivalent to
    another that is periodic along every axis.
\end{remark}

\begin{corollary}\label{thm:ShiftsTrivialForDdGeq2}
    Let \(\dd \geq 2\) and \(\alpha : \fMat(\ZZ^\dd,p) \to \fMat(\ZZ^\dd,p)\) be a shift QCA.
    Then there is a possibly nonuniform stabilization
    such that \(\alpha\otimes \id\) is equal to a finite depth circuit of local automorphisms.
\end{corollary}
\begin{proof}
    We find a circuit QCA \(\beta\) such that \(\beta \alpha\) acts as the identity in a sufficiently large ball \(B\) (radius \(r \geq 2^{4\dd+5} \ell\)), and then apply \ref{thm:QCAWithAHoleIsTrivial} to conclude that \(\beta \alpha\), and hence \(\alpha\), is a finite depth circuit.

    Up to a finite depth circuit, we can consider \(\alpha\) to be primitive by \ref{lem:ShiftsAreCircuitEquivToPrimitive}.
    By \ref{lem:PrimitiveShiftToPermutation} this primitive shift is simply a permutation of the single-site primitive factors of \(\fMat(\ZZ^\dd, p) = \bigotimes_j \calA_j\).
    Associate a directed graph \(G_\alpha\) to this permutation in the obvious way: the vertex set \(V = \{j\}\) indexes the collection of single-site primitive factor algebras \(\calA_{j}\), and an edge \((j,j')\) is present if \(\alpha(\calA_{j}) = \calA_{j'}\). 
    The graph \(G_\alpha\) is then a collection of disjoint cycles:
    a cycle is a directed one-dimensional graph without boundary,
    either an infinite line or a finite circle.
    
    Let \(B\) be a ball of radius \(r \geq 2^{4\dd+5} \ell\),
    and consider a connected component \(\gamma\) of \(G_\alpha \cap B^{+\ell}\). 
    Every loop in the \(\dd\)-dimensional ball is the boundary of a $2$-disk in the ball,
    and for \(\dd \geq 2\) every path from one point on the surface~$\partial B$
    to another is homologous to a path between the same end-points
    which is entirely along the surface~$\partial B$.
    Thus, we can find a finite simplicial 2-complex (a disk) \(\Sigma\)
    such that \(\gamma + \partial \Sigma \subseteq B^{+\ell} \setminus B\)
    is a (possibly empty) path along the surface with steps of length \(\leq \ell\).
    In the expression \(\gamma + \partial \Sigma\),
    we regard \(\gamma\) and \(\Sigma\) as \(\ZZ\)-chains
    in the cell complex for \(B^{+\ell}\)
    (with vertex points separated out into primitive factor algebras).

    Subdivide \(\Sigma\) until it has 2-cells of diameter at most \(\ell\).%
    \footnote{To maintain the lattice as \(\ZZ^\dd\), one can move each vertex of the sub-triangulation to the nearest integer point, so the diameter of a 2-cell is now at most \(\ell+1\).}
    To each triangle \(\sigma = [i,j,k]\),
    associate a shift \(u_\sigma = \SWAP_{ij}\SWAP_{ki}\)
    acting on algebras isomorphic to those shifted by \(\alpha\),
    which must be added as ancillas if not already present.
    Then the product \(\beta_\gamma = \prod_{\sigma \in \Sigma} u_\sigma\)
    is a shift around \(\partial \Sigma\),
    and \(\beta_\gamma \alpha\)
    moves the component of the shift associated to \(\gamma\)
    to the surface \(B^{+\ell} \setminus B\).
    The product \(\prod_{\sigma \in \Sigma} u_\sigma\)
    is a finite depth circuit as the triangulation is finite.

    Repeat this for every connected component of \(G_\alpha \cap B^{+\ell}\), of which there can only be finitely many, and thus obtain a finite circuit \(\beta = \prod_{\gamma} \beta_\gamma\) such that \(\beta \alpha\) acts as the identity in \(B\) and has spread \(\ell\). Then apply \ref{thm:QCAWithAHoleIsTrivial}.
\end{proof}

\begin{remark}
    A similar proof (replacing paths and 2-cells with higher dimensional cells) shows that any \(\dd\)-dimensional (fermionic) QCA given by a product of embedded QCAs of dimension \(k < \dd\) is a finite depth circuit after possibly nonuniform stabilization.
\end{remark}

\begin{corollary}[\(2\dd\) QCA are circuits]\label{thm:2DfQCAAreTrivialWithNonuniformStabilization}
    Let \(\alpha : \fMat(\ZZ^2,p) \to \fMat(\ZZ^2,p)\) be a QCA. Then there is a possibly nonuniform local assignment \(p'\) such that \(\alpha\otimes\id : \fMat(\ZZ^2,pp') \to \fMat(\ZZ^2,pp')\) is a finite depth circuit.
\end{corollary}
\begin{proof}
    By \ref{thm:fQCAStructure1d2d}, every fermionic QCA on \(\ZZ^2\) is finite depth circuit equivalent to a primitive shift after a uniformly bounded stabilization.
    By \ref{thm:ShiftsTrivialForDdGeq2}, every shift in \(\dd = 2\) becomes a finite depth circuit after a possibly nonuniform stabilization.
\end{proof}

\subsection{Spread reduction by pleating and ironing}\label{subsec:PleatingIroning}

The following two subsections present the proof of \ref{thm:QCAWithAHoleIsTrivial}.
The first step of the proof is to establish a method of \emph{spread reduction}---that is,
to construct a finite depth circuit \(\beta\)
such that \(\beta \circ \alpha\) has much smaller spread than \(\alpha\).

In the spread reduction procedure, it is convenient to introduce more general lattices (see \ref{def:Lattice}) over \(\RR^\dd\), so that we can achieve spreads that are less than \(1\).
In the following, we always take \(\Lambda \subset \RR^\dd\) to be a lattice in \(\RR^\dd\).
The use of an arbitrary lattice is harmless: to recover an integer lattice, we can simply round the points of \(\Lambda\) to integers to find a QCA defined on \(\ZZ^\dd\) with slightly larger spread.

The construction proceeds iteratively over the coordinate axes of \(\RR^\dd\).
As such, it is useful to introduce notation for spreads along different axes.

\begin{definition}[Anisotropic spread]
    For \(\mathbf{L} \in [0,\infty)^\dd\) and \(S \subseteq \RR^\dd\), define the \(\mathbf{L}\)-growth of \(S\) by
    \begin{equation}
        S^{+\mathbf{L}} =
        \{t = (t_1,\ldots,t_\dd) \in \RR^\dd \mid \text{there is }s\in S
        \text{ such that }|t_k-s_k| \leq L_k \text{ for all }k\},
    \end{equation}
    where \(t_k,s_k,L_k\) are the \(k\)th components of \(t,s,\mathbf{L}\).
    If a QCA \(\alpha\) on \(\fMat(\ZZ^\dd)\) obeys \(\Supp(\alpha(x)) \subseteq \Supp(x)^{+\mathbf{L}}\) for all \(x \in \fMat(\ZZ^\dd)\), we say that \(\alpha\) has anisotropic spread \(\mathbf{L}\).
\end{definition}

The simplest spread reduction construction operates on one axis of the whole lattice, as follows.

\begin{lemma}\label{lem:PleatingSpreadReductionOneDimension}
    Given any QCA \(\alpha\) on \(\fMat(\Lambda,p)\) with anisotropic spread \(\mathbf{L}\),
    an axis \(k\in \{1, ...,\dd\}\), and an odd integer scaling factor \(\lambda \in 2\NN+1\),
    there is a finite depth circuit \(\beta\) on \(\fMat(\Lambda',pp')\) (where \(\Lambda' \supseteq \Lambda\))
    such that \(\beta \circ (\alpha \otimes \id)\) has anisotropic spread 
    \begin{equation}
        \mathbf{L}' = \left(2L_1, \ldots, 2 L_{k-1}, \frac{4L_k}{\lambda} , 2 L_{k+1},\ldots, 2L_\dd\right).
    \end{equation}
    Furthermore, \(\beta\) consists of
    a layer of single-site local automorphisms
    followed by a layer of commuting local automorphisms,
    the support of which have a width at most \(10L_k\) in the \(k\)th axis
    and at most \(2L_{j} \) in all other axes $j \neq k$.
    In particular, \(\beta\) has anisotropic spread
    \((2L_1, ..., 10L_k, ..., 2L_\dd)\).
\end{lemma}

Note that the spread of \(\beta\) is bounded independent of \(\lambda\).

\begin{proof}
    Without loss of generality, we take \(k = 1\).
    The claim is trivial unless \(\lambda \geq 5\) and \(L_1 > 0\), so we restrict to this case.

    The spread-reduced QCA \(\beta \circ (\alpha \otimes \id)\) is constructed by
    \emph{pleating} and \emph{ironing}~\cite{FHH2019}.
    See~\autoref{fig:PleatCompress}.
    Consider the QCA
    \begin{equation}
        \gamma = \underbrace{\alpha \otimes \alpha^{-1} \otimes \alpha \otimes \alpha^{-1} \otimes  \cdots \otimes \alpha}_{\lambda \text{ factors}},
    \end{equation}
    consisting of an alternating product of \((\lambda+1)/2\) copies of \(\alpha\) and \((\lambda-1)/2\) copies of \(\alpha^{-1}\).
    We can regard any two adjacent factors as being realized as a circuit as in \ref{thm:QCAQCAInverseIsCircuit}.
    Inspecting the proof shows that this circuit consists of two types of local automorphisms:
    \(\{\sigma_s\}_{s\in \ZZ^\dd}\),
    where \(\sigma_s\) is a swap between factor algebras supported on site \(s\);
    and \(\{\mu_s = (\alpha \otimes \id) \circ \sigma_s \circ (\alpha^{-1} \otimes \id)\}_{s\in \ZZ^\dd}\),
    all of which commute and have support contained within \(\{s\}^{+\mathbf{L}}\).
    Upon removing local automorphisms from this circuit,
    the anisotropic spread of the new circuit cannot increase above \(2\mathbf{L}\),
    as this bounds the support of any local automorphism in the circuit.

    \begin{figure}
        \centering
        \begingroup
        \resizebox{0.9\textwidth}{!}{\input{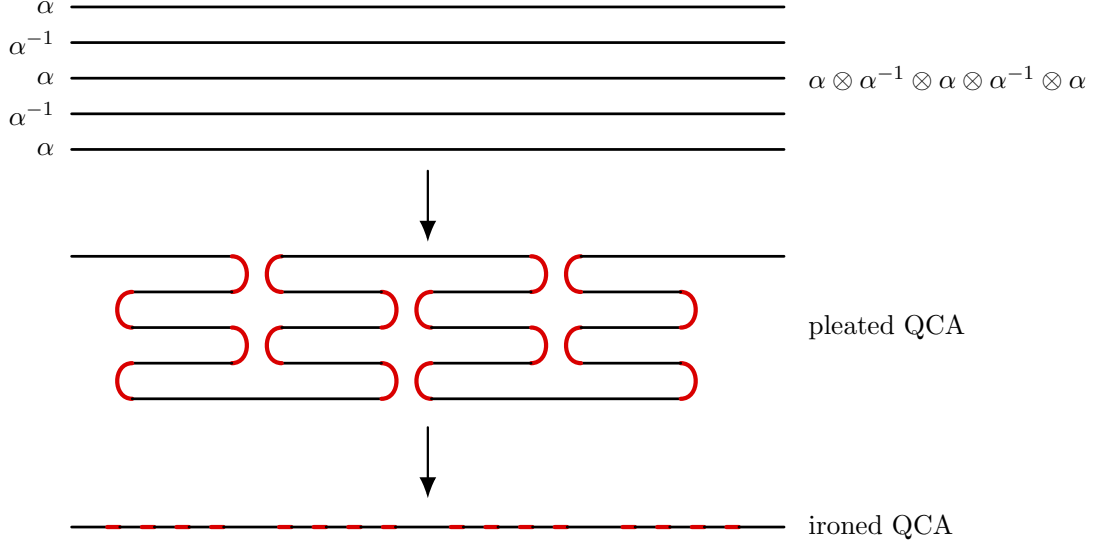}}
        \endgroup
        \caption{The \emph{pleating and ironing} construction used in \ref{lem:PleatingSpreadReductionOneDimension}.
        Any QCA \(\alpha\) is equivalent to \(\alpha \otimes (\alpha^{-1}\otimes \alpha)^{\otimes n}\) (top panel with \(n=2\) layers).
        By removing local automorphisms in the circuit for \(\alpha \otimes \alpha^{-1}\) or \(\alpha^{-1} \otimes \alpha\), adjacent pairs of layers can be blended to the identity in alternating ways, creating a \emph{pleating} of \(\alpha\)~\cite{FHH2019}.
        The red segments indicate positions where the action of the pleated QCA differs from a tensor product of \(\alpha\) or \(\alpha^{-1}\).
        This pleating can then be flattened and compressed (or \emph{ironed out}), simply by rearranging local algebras in space.
        The result is a QCA \(\beta\circ (\alpha \otimes \id)\),
        where \(\beta\) is a circuit,
        such that \(\beta\circ (\alpha \otimes \id)\)
        has a reduced spread along the direction of the pleat.}
        \label{fig:PleatCompress}
    \end{figure}

    We construct a new QCA \(\gamma'\)
    by removing local automorphisms from \(\gamma\)
    in a periodic pattern along the first axis of \(\ZZ^\dd\).
    Near the cuts \(\{s \in \RR^\dd\mid s_1 \in 10 L_1 \ZZ \}\),
    we interpret \(\gamma\) as
    \begin{equation}\label{eqn:PleatPair1}
        \gamma = \alpha \otimes (\alpha^{-1}\otimes \alpha)^{\otimes (\lambda-1)/2},
    \end{equation}
    and remove any of the local automorphisms in the \(\alpha^{-1} \otimes \alpha\) factors that cross the cut.
    Then, along the cuts \(\{s \in \RR^\dd\mid s_1 \in 10 L_1 (\ZZ + \tfrac{1}{2}) \}\), we interpret \(\gamma\) as
    \begin{equation}\label{eqn:PleatPair2}
        \gamma = (\alpha\otimes \alpha^{-1})^{\otimes (\lambda-1)/2} \otimes \alpha,
    \end{equation}
    and similarly remove any local automorphism in \(\alpha \otimes \alpha^{-1}\) that crosses these cuts.
    The resulting \(\gamma'\) is visualized in \autoref{fig:PleatCompress}, where it appears as a ``pleated'' version of \(\alpha\).
    By the previous paragraph, \(\gamma'\) has anisotropic spread at most \(2 \mathbf{L}\), and is related to \(\gamma\) by a depth-1 circuit of \(\mu_s\), which also has anisotropic spread at most \(2 \mathbf{L}\).

    The pleated structure of \(\gamma'\) means in particular that local operators \(x\) spread under \(\gamma'\) in the meandering connectivity illustrated in \autoref{fig:PleatCompress}.
    Further than \(2L_1\) away from the cuts (the cuts are spaced by \(5 L_1\), so there is a nonempty segment at least \(2L_1\) from any cut), \(\gamma'\) agrees with \(\gamma\), which does not spread operators between layers.
    Within \(2L_1\) of the cuts, \(\gamma'\) only has local automorphisms
    which are supported on the pairs of layers specified
    in~\eqref{eqn:PleatPair1} and~\eqref{eqn:PleatPair2},
    so operators only spread between these disjoint pairs, and only the unpaired layer spreads operators across the cut.
    Thus, we can rearrange the local algebras introduced in the construction in space, unwinding the meander of \autoref{fig:PleatCompress} to form an \emph{ironed flat} lattice \(\Lambda'\) with a spacing along the first axis that is \(1/\lambda\) times as small as \(\Lambda\).
    Specifically, we compress all the additional ancillae within the segment \(s_1 \in (10L_1 j, 10L_1 (j+1))\) so that they are uniformly spaced in a single layer within this segment.
    With this changed geometry, \(\gamma'\) now has spread at most \(4L_1/\lambda\) along the first axis.
    Indeed, an operator supported on one side of a red segment in \autoref{fig:PleatCompress} can in general be mapped anywhere within the segment.
    Each such segment is a double-layer of width at most \(2 L_1\).
    Unfolding this produces a segment of width at most \(4L_1\), so the spread along the first axis of \(\gamma'\) in the unraveled geometry is at most \(4L_1\).
    Compressing the lattice reduces this to \(4L_1/\lambda\).

    We can express the ironed QCA as a conjugation of \(\gamma'\) by a circuit of disjoint swaps,
    \begin{equation}\label{eqn:PleatedIronedQCAConjugation}
        \mathsf{S}_{\Lambda \Lambda'} \gamma' \mathsf{S}_{\Lambda \Lambda'} = \left( \prod_{s' \in \Lambda'} \SWAP_{s s'} \right) \circ \gamma' \circ \left( \prod_{s' \in \Lambda'} \SWAP_{s s'} \right),
    \end{equation}
    where \(s\) is the site in the pleated structure corresponding to the ironed site \(s'\) (we split sites belonging to different pleats).
    This is related to \(\gamma\) by a circuit multiplying on either the right or left by~\ref{cor:QCAmodCircuitsIsAbelian},
    \begin{equation}
        (\mathsf{S}_{\Lambda \Lambda'} \gamma' \mathsf{S}_{\Lambda \Lambda'} \gamma^{\prime-1}) \gamma' = \mathsf{S}_{\Lambda \Lambda'} \gamma' \mathsf{S}_{\Lambda \Lambda'} = \gamma' (\gamma^{\prime-1} \mathsf{S}_{\Lambda \Lambda'} \gamma' \mathsf{S}_{\Lambda \Lambda'} ).
    \end{equation}
    In either case, as \(\gamma\) is itself related to \(\alpha\otimes \id\) by a circuit, we get a circuit \(\beta\) relating \(\alpha\otimes \id\) to \(\mathsf{S}_{\Lambda \Lambda'} \gamma' \mathsf{S}_{\Lambda \Lambda'}\).
    Indeed, the circuit still consists of \(\sigma_s\) and \(\mu_s\) from \ref{thm:QCAQCAInverseIsCircuit}, but now with swaps acting across different sites.

    Finally, we analyze the circuit structure of \(\beta\).
    First, \(\alpha \otimes\id\) is related to \(\gamma\) by a layer of single-site automorphisms
    followed by a single layer of automorphisms of support width \(2\mathbf{L}\)
    (the diagonal of a rectangle that fully contains the support of the local automorphisms). 
    In turn, \(\gamma'\) is related to \(\gamma\) by another layer of width-\(2\mathbf{L}\)
    commuting local automorphisms,
    which also commute with the local automorphisms of the previous layer.
    The rearrangement of ancilla positions from the layered geometry to the compressed geometry
    does not affect the action of the circuit as an algebra automorphisms,
    but it does change the width of the local automorphisms' support through the mapping of sites.
    It does not alter the support of the single-site automorphisms
    or the width of the multi-site automorphisms transverse to the first axis.
    To get the bound on the width in the first direction,
    we consider two cases.
    \begin{enumerate}
        \item Automorphisms which do not cross the cuts \(\{s \in \RR^\dd\mid s_1 \in 10 L_1 \ZZ\}\)
        have a support contained within the compression region
        \(s_1 \in (10 L_1 j, 10 L_1 (j+1))\),
        and so their support is bounded by \(10 L_1\), the width of this region.
        \item For automorphisms crossing one of the cuts \(\{s \in \RR^\dd\mid s_1 \in 10 L_1 \ZZ \}\),
        we observe that while the support is divided between two compression regions,
        it is in the latter half of the meander to the left
        and the early half of the meander to the right (\autoref{fig:PleatCompress}),
        which will both be mapped by the compression to within a distance \(5L_1\) of the cut.
        Thus, these automorphisms also have a width along the first axis bounded by \(10L_1\).
    \end{enumerate}
    This completes the description of the circuit \(\beta\).
\end{proof}

\begin{lemma}\label{rem:IdentityPreservation}
    If \(\alpha\) acts as the identity in a ball \(B\),
    then \(\beta\) from \ref{lem:PleatingSpreadReductionOneDimension}
    can be chosen so that \(\beta \circ (\alpha \otimes \id)\)
    also acts as the identity in \(B^{-(0, ..., 5L_k, ..., 0)}\),
    and \(\Lambda'\) involves \emph{no} new sites in this region.
\end{lemma}
\begin{proof}
    Indeed, this becomes clear from~\eqref{eqn:PleatedIronedQCAConjugation}%
    ---any swap automorphisms cancel out when acting on two sites in \(\Lambda'\)
    corresponding to sites in the pleated \(\Lambda\)
    where \(\alpha\) (and hence \(\gamma'\)) acts as the identity.
    In the ironing procedure, any site where \(\alpha\) did not act as the identity
    ends up at most \(5L_k\) from where it started,
    so that \(B^{-(0, ..., 5L_k, ..., 0)}\) contains no such sites.
    Thus, any new ancillas from \(\Lambda'\) that lie within this shrunken ball can be removed.
\end{proof}

\begin{corollary}[Global spread reduction]\label{lem:PleatingSpreadReduction}
    Given any spread-\(\ell\) QCA \(\alpha\) on \(\fMat(\Lambda,p)\) and \(\epsilon > 0\),
    there is a finite depth circuit \(\beta\) (depth independent of \(\epsilon\)) of spread \(11\cdot 2^{\dd-1}\ell\)
    on \(\fMat(\Lambda',pp')\) such that
    \(\beta\circ (\alpha\otimes\id)\) has spread \(\epsilon \ell\).

    Further, if \(\alpha\) acts trivially in a ball \(B\), then \(\beta\) can be chosen to involve no new ancillas and act trivially in the ball \(B^{-5\cdot 2^{\dd-1} \ell}\).
\end{corollary}

\begin{proof}
    The corollary follows by repeated application of \ref{lem:PleatingSpreadReductionOneDimension} along each axis of \(\RR^\dd\).
    Applying \ref{lem:PleatingSpreadReductionOneDimension} starting from the first axis (\(k=1\)) to the last (\(k=\dd\)), 
    the spread becomes \(2^{\dd-1} \cdot 4 \ell/\lambda\). 
    Taking \(\lambda\) such that \(2^{\dd-1} \cdot 4 /\lambda \leq \epsilon\) gives the desired spread.

    The spread of the circuit \(\beta\) is obtained by summing the anisotropic spreads of each of the circuits used in \ref{lem:PleatingSpreadReductionOneDimension}.
    That is,
    \begin{equation}\begin{split}
        \mathbf{L}_\beta 
        &= (10 \ell, 2 \ell, \ldots, 2 \ell) \\
        &\quad +
        (2 \cdot 4 \ell/\lambda, 10 \cdot 2 \ell, \ldots, 4\ell)\\
        &\quad + \cdots \\
        &\quad +
        (2^{\dd-1} \cdot  4 \ell/\lambda, 2^{\dd-1} \cdot  4 \ell/\lambda, \ldots,  10 \cdot 2^{\dd-1} \ell).
    \end{split}\end{equation}
    The \(\dd\)th component never benefits from the \(4/\lambda\) reduction in spread, and bounds all the other components,
    \begin{equation}
        \mathbf L_{\beta, \dd} = (2 + 4 + \cdots + 2^{\dd-2} + 10 \cdot 2^{\dd-1}) \ell \leq 11 \cdot 2^{\dd-1} \ell, 
    \end{equation}
    which gives the claimed bound on the spread of \(\beta\).

    That \(\beta\) can be chosen to involve no new ancillas in \(B^{-5\cdot 2^{\dd-1} \ell}\) similarly follows by applying \ref{rem:IdentityPreservation} repeatedly, with the largest spread inwards coming from the \(\dd\)th axis.
\end{proof}

To prove \ref{thm:QCAWithAHoleIsTrivial}, the pleating and ironing procedure must be adapted to apply locally.
That is, the spread reduction can be done just near the ball \(B\), without blowing up the spread outside of \(B\) very much.
If we were to apply spread reduction globally, then the depth of the total pleating-ironing-hemming circuit will grow as the hole in \(\alpha\) is enlarged.
It is also not sufficient to simply truncate the circuit \(\beta\), as the QCA we obtain from this truncation could have a spread as large as \((11\cdot 2^{\dd-1}+1)\ell\), even larger than the QCA we started with.

In this setting, we need to control an \emph{inner} and \emph{outer spread} of the modified QCA. 
The inner spread is a (larger) spread achieved only inside some ball, while the spread of operators supported outside the ball is the (smaller) outer spread.
As we pleat, iron, and hem \(\alpha\), the outer spread will remain \(\ell\), and we must make sure that the inner spread does not grow too much.

\begin{definition}[Inner and outer spread]
    A QCA \(\alpha\) on \(\fMat(\Lambda, p)\) has \emph{inner spread} \(L\)
    and \emph{outer spread} \(\ell \leq L\)
    relative to a set \(S \subseteq \RR^\dd\) if \(\alpha\) has spread \(L\) but
    \begin{equation}
        \Supp(\alpha(x)) \subseteq \Supp(x)^{+\ell}
        \quad\text{whenever}\quad
        \Supp(x) \subseteq \RR^{\dd} \setminus S.
    \end{equation}
    Similarly define inner and outer anisotropic spread.
\end{definition}

\begin{figure}
    \centering
    \begingroup
    \hspace*{-5ex}
    \resizebox{0.95\textwidth}{!}{\definecolor{cRed}{RGB}{214,62,58}       
\definecolor{cOrange}{RGB}{232,150,30}   
\definecolor{cBlue}{RGB}{58,110,200}     
\colorlet{cFold}{red!85!black}           
\definecolor{cOuter}{RGB}{150,150,150}   

\newlength{\interpolationFigUnit}
\setlength{\interpolationFigUnit}{24.2081pt}

\begin{tikzpicture}[
    x=\interpolationFigUnit, y=\interpolationFigUnit,
    line cap=round,
    line join=round,
    base/.style     = {black,  line width=1.0pt},
    fold/.style     = {cFold,  line width=1.7pt},
    boxsolid/.style = {black,  line width=1.0pt},
    boxouter/.style = {cOuter, line width=1.0pt},
    cut/.style      = {line width=1.2pt, dash pattern=on 5pt off 4pt},
    frame/.style    = {line width=1.3pt, rounded corners=3pt},
    panellab/.style = {anchor=north west, xshift=5pt, yshift=-4pt, black},
    arrow/.style    = {{Latex[length=2.6mm,width=2.0mm]}-{Latex[length=2.6mm,width=2.0mm]},
                       line width=0.8pt, black},
]

\def\Oa{0}\def\Ob{9}          
\def\Ia{2.5}\def\Ib{6.5}      
\def\cutlen{3.4}
\def\ycutR{4.2}               
\def\ycutO{1.25}              
\def\ycutB{0.20}              
\def\yarrow{4.5}              

\def\Pxa{9.7}\def\Pxb{18.5}   
\def\pxa{11.1}\def\pxb{17.9}  
\def\Ph{2.5}                  
\def\ybotP{0}                 
\def\ymidP{3.25}              
\def\ytopP{6.5}               

\def\nfold{4}                 
\def\nbundle{4}               
\def\dyP{0.45}                
\def\bwP{1.00}                
\def\pitch{1.70}              
\def\cOne{12.2}               
\def\foldreach{0.62}          

\def\tTop{1.00}
\def\tMid{0.55}
\def\tBot{0.00}

\def\hideTop{4}
\def\hideMid{0}
\def\hideBot{0}


\newcommand{\mapXY}[3]{
    \pgfmathsetmacro{\Xmap}{(1-\tpar)*(#1) + \tpar*(\pxa + \lam*(#2))}%
    \pgfmathsetmacro{\Ymap}{\yctr + (1-\tpar)*((#3) - \yctr)}%
}

\newcommand{\doflat}[1]{%
    \pgfmathsetmacro{\Lx}{#1}%
    \pgfmathsetmacro{\xnew}{\xcur+\Lx}%
    \pgfmathsetmacro{\snew}{\sacc+abs(\Lx)}%
    \mapXY{\xcur}{\sacc}{\ylev}%
    \pgfmathsetmacro{\XA}{\Xmap}\pgfmathsetmacro{\YA}{\Ymap}%
    \mapXY{\xnew}{\snew}{\ylev}%
    \draw[base] (\XA,\YA) -- (\Xmap,\Ymap);
    \xdef\xcur{\xnew}\xdef\sacc{\snew}%
}

\newcommand{\dofold}[2]{%
    \pgfmathsetmacro{\sg}{#1}%
    \pgfmathsetmacro{\dl}{#2}%
    \pgfmathsetmacro{\arc}{pi*abs(\dl)/2}
    \pgfmathsetmacro{\cr}{\foldreach*abs(\dl)*(1-\tpar)}%
    \pgfmathsetmacro{\ynew}{\ylev+\dl}%
    \pgfmathsetmacro{\snew}{\sacc+\arc}%
    \mapXY{\xcur}{\sacc}{\ylev}%
    \pgfmathsetmacro{\XA}{\Xmap}\pgfmathsetmacro{\YA}{\Ymap}%
    \mapXY{\xcur}{\snew}{\ynew}%
    \xdef\foldidx{\the\numexpr\foldidx+1\relax}%
    \ifnum\foldidx>\foldlast
        \draw[base] (\XA,\YA)
            .. controls ({\XA+\sg*\cr},\YA) and ({\Xmap+\sg*\cr},\Ymap) .. (\Xmap,\Ymap);
    \else
        \draw[fold] (\XA,\YA)
            .. controls ({\XA+\sg*\cr},\YA) and ({\Xmap+\sg*\cr},\Ymap) .. (\Xmap,\Ymap);
    \fi
    \xdef\ylev{\ynew}\xdef\sacc{\snew}%
}

\newcommand{\pleatFig}[3]{%
    \xdef\tpar{#1}%
    \pgfmathsetmacro{\yc}{#2}\xdef\yctr{\yc}%
    \xdef\foldidx{0}%
    \xdef\foldlast{\the\numexpr\nbundle*\nfold-(#3)\relax}%
    \pgfmathsetmacro{\Sarc}{%
        (\cOne+\bwP/2-\pxa)
        + \nbundle*(\nfold*pi*\dyP/2 + (\nfold-1)*\bwP)
        + (\nbundle-1)*(\pitch+\bwP)
        + (\pxb-(\cOne+(\nbundle-1)*\pitch)+\bwP/2)}%
    \pgfmathsetmacro{\lamv}{(\pxb-\pxa)/\Sarc}\xdef\lam{\lamv}%
    \xdef\xcur{\pxa}\xdef\sacc{0}%
    \pgfmathsetmacro{\ytopL}{\yc+\nfold*\dyP/2}\xdef\ylev{\ytopL}%
    \doflat{\cOne+\bwP/2-\pxa}%
    \foreach \i in {1,...,\nbundle}{%
        \ifodd\i \pgfmathsetmacro{\dstep}{-\dyP}\else \pgfmathsetmacro{\dstep}{\dyP}\fi
        \foreach \j in {1,...,\nfold}{%
            \ifodd\j \pgfmathsetmacro{\sgn}{1}\else \pgfmathsetmacro{\sgn}{-1}\fi
            \dofold{\sgn}{\dstep}%
            \ifnum\j<\nfold \doflat{-\sgn*\bwP}\fi
        }%
        \ifnum\i<\nbundle \doflat{\pitch+\bwP}\fi
    }%
    \doflat{\pxb-\xcur}%
}

\draw[boxouter] (\Oa,\Oa) rectangle (\Ob,\Ob);
\draw[boxsolid] (\Ia,\Ia) rectangle (\Ib,\Ib);
\node[panellab] at (\Oa,\Ob) {(a)};
\node[anchor=south west] at ({\Ia+0.15},{\Ia+0.12}) {$B$};

\pgfmathsetmacro{\xmid}{(\Ia+\Ib)/2}
\draw[arrow] (\Oa,\yarrow) -- (\Ia,\yarrow);
\node[above] at ({(\Oa+\Ia)/2},\yarrow) {$5L+3\ell$};
\draw[arrow] (\Ia,\yarrow) -- (\xmid,\yarrow);
\node[above] at ({(\Ia+\xmid)/2},\yarrow) {$r$};

\pgfmathsetmacro{\xRb}{\Ib-0.7}
\draw[cut,cRed]    (\xRb,\ycutR) -- ({\xRb+\cutlen},\ycutR);
\pgfmathsetmacro{\xOb}{2.7}
\draw[cut,cOrange] (\xOb,\ycutO) -- ({\xOb+\cutlen},\ycutO);
\draw[cut,cBlue]   (\xOb,\ycutB) -- ({\xOb+\cutlen},\ycutB);

\draw[frame,cRed]    (\Pxa,\ytopP) rectangle (\Pxb,{\ytopP+\Ph});
\node[panellab] at (\Pxa,{\ytopP+\Ph}) {(b)};
\pleatFig{\tTop}{\ytopP+\Ph/2}{\hideTop}

\draw[frame,cOrange] (\Pxa,\ymidP) rectangle (\Pxb,{\ymidP+\Ph});
\node[panellab] at (\Pxa,{\ymidP+\Ph}) {(c)};
\pleatFig{\tMid}{\ymidP+\Ph/2}{\hideMid}

\draw[frame,cBlue]   (\Pxa,\ybotP) rectangle (\Pxb,{\ybotP+\Ph});
\node[panellab] at (\Pxa,{\ybotP+\Ph}) {(d)};
\pleatFig{\tBot}{\ybotP+\Ph/2}{\hideBot}

\end{tikzpicture}}
    \endgroup
    \caption{
        The geometry for the local pleating-and-ironing spread reduction of a QCA.
        (a)~The spread along axis \(k\) (horizontal) inside a ball \(B\) of radius \(r\) can be reduced by applying a circuit in \(B^{+5 L_k + 2 \ell}\), where \(L_k\) is the inner spread along \(k\) and \(\ell\) is the outer spread.
        Inside \(B\), the new QCA is the pleated-and-ironed QCA from \autoref{fig:PleatCompress}.
        (b)~Across the boundary of \(B\) parallel to \(k\) (red dashed), the pleats are terminated so that there is a direct blend to the original QCA.
        (c) Along the other directions (orange dashed), the ironing in \(B\) is gradually undone.
        (d)~At the edge of \(B^{+5 L_k}\) (blue dashed), the QCA has the un-ironed, pleated form.
        Outside \(B^{+5 L_k + 2 \ell}\), all local automorphisms relating the pleating to the original QCA are removed.
    }
    \label{fig:PleatingInterpolation}
\end{figure}
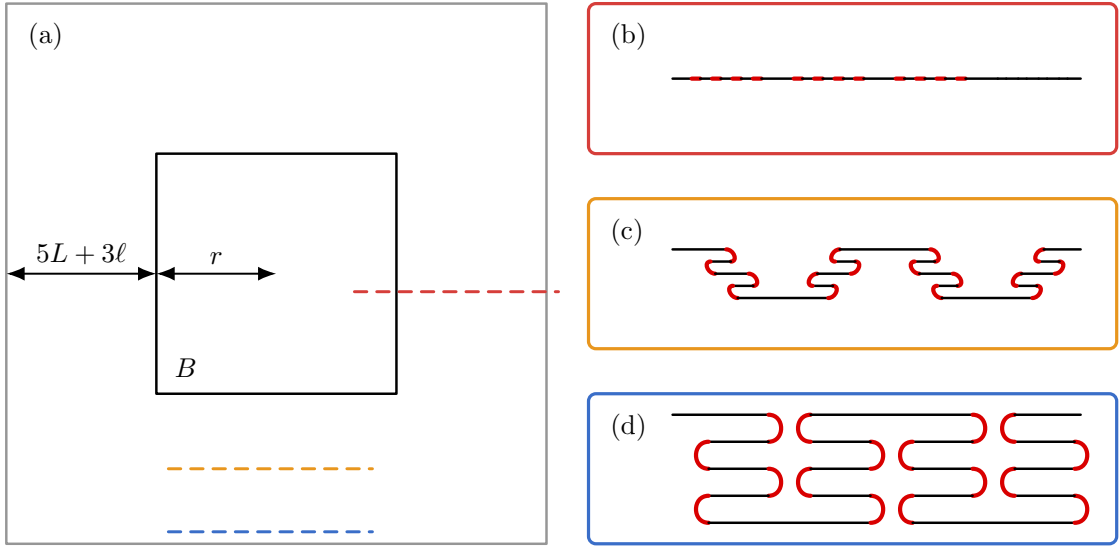

\begin{lemma}\label{lem:PleatingLocalOneDimension}
    Let \(\alpha\) be a QCA on \(\fMat(\Lambda,p)\)
    with inner anisotropic spread \(\mathbf{L}\) and outer spread \(\ell\)
    relative to a ball \(B\).
    Then, given an axis \(k\in \{1, \ldots, \dd\}\)
    and an odd integer scaling factor \(\lambda \in 2\NN+1\),
    there is a finite depth circuit \(\gamma\) on \(\fMat(\Lambda',pp')\)
    such that \(\gamma\circ (\alpha \otimes \id)\) has inner anisotropic spread 
    \begin{equation}
        \mathbf{L}' = \left(2L_1, \ldots,2L_{k-1}, \frac{4L_k}{\lambda} , 2 L_{k+1},\ldots, 2L_\dd\right)
    \end{equation}
    and outer spread \(8\ell\) relative to the ball \(B\).
    Furthermore, the circuit \(\gamma\) is supported in \(B^{+5L_k+2\ell}\)
    (with spread bounded by \(\mathbf{L}\) and \(\dd\),
    as in~\ref{lem:PleatingSpreadReductionOneDimension}).
    In particular, the spread of \(\gamma\circ(\alpha \otimes \id)\)
    outside of \(B^{+5L_k+3\ell}\) is still~\(\ell\).
\end{lemma}

\begin{proof}
    Again, take \(k=1\).
    We aim to construct \(\gamma\) by restricting the circuit \(\beta\)
    from~\ref{lem:PleatingSpreadReductionOneDimension} to~\(B^{+5 L_1 + 2\ell}\) while maintaining a small spread.
    The geometry we use is depicted in \autoref{fig:PleatingInterpolation}.

    First, we restrict \(\beta\) parallel to the first axis.
    This is done simply by terminating the pleats that define
    \(\beta\); see \autoref{fig:PleatingInterpolation}(b).
    Let $S$ be the \((\dd-1)\)-dimensional slab containing \(B^{+5 L_1 + 2 \ell}\)
    normal to the first axis,
    and $\tilde \beta$ be the QCA obtained by terminating pleats in~\(\beta\).
    The blocks defining the pleats are size \(10L_1\),
    so by possibly repositioning the location of the pleats%
    ---balancing them so that a pleat extends by at most \(5L_1\) from either side of~\(S\)%
    ---we can ensure that any block touching \(B\) does not reach outside \(S\).

    To restrict \(\tilde{\beta}\) from \(S\) to \(B^{+5 L_1 + 2 \ell}\), we interpolate the displacement of sites defining the ironing between the un-ironed positions (in \(\lambda\) layers) and the ironed positions (with reduced spread and one layer); see~\autoref{fig:PleatingInterpolation}(c).
    That is, defining the un-ironed position of a site in the pleated QCA to be \(s^i = (s_1, s_2, ..., s_\dd) \in \RR^\dd\) (where \(i\) represents the layer index in the pleat), the ironed position to be \(y^i(s) = (y^i_1, s_2, ..., s_\dd) \in \RR^\dd\), and the projection \(P s = (0, s_2, ..., s_\dd)\), a partial ironing is defined by the repositioning of sites
    \begin{equation}
        s^i \mapsto z^i(s) =  t(s) s^i + (1-t(s^i)) y^i(s), \qquad
        t(s) = \min\{ 1, \dist(Ps, P B)/5L_1\}
    \end{equation}
    The interpolation is based on the relative distance between \(s^i\) and the ball \(B\) compared to the spread \(5L_1\).
    
    Define a QCA \(\tilde{\gamma}\) that has the same action as \(\tilde\beta\) at the level of the algebra, but for which sites previously at \(s^i\) are now positioned at \(z^i(s)\).
    Subsequently remove all local automorphisms from \(\tilde{\gamma}\) that do not have support overlapping with \(B^{+5 L_1}\).
    All remaining automorphisms are supported in \(B^{+5 L_1+2\ell}\), and the resulting circuit is \(\gamma\).

    We must show that \(\gamma\) has the claimed spread. 
    Inside the ball \(B\), \(\gamma\circ(\alpha\otimes \id)\) agrees
    with \(\beta\circ(\alpha\otimes \id)\) and
    has anisotropic spread \((4L_1/\lambda, 2 L_2, \ldots,  2L_\dd)\).
    Outside of \(B^{+5 L_1+2\ell}\), \(\gamma\) has no support, so that outside of the slightly larger ball \(B^{+5 L_1+3\ell}\),
    \(\gamma\circ (\alpha\otimes\id)\) has the outer spread of \(\alpha\), namely \(\ell\).
    In the shell between \(B^{+5 L_1}\) and \(B^{+5 L_1+3\ell}\),
    \(\gamma\circ(\alpha\otimes \id)\) differs from \(\alpha\otimes\id\)
    by a depth-1 circuit of automorphisms with support diameter at most \(2 \ell\),
    and so has spread \(3 \ell\).
    Within the interpolating region \(B^{+5 L_1} \setminus B\),
    an operator \(\gamma\circ(\alpha\otimes \id)(x)\)
    is related to the un-interpolated \(\alpha \otimes \alpha^{-1} \otimes \cdots \otimes \alpha\)
    by the shear \(z^i(s)\)---%
    sites closer to \(B\) get displaced more than sites further from \(B\).
    The rest of the proof consists of analyzing the spread of \(\Supp(\gamma\circ(\alpha\otimes \id)(x)) = z(\Supp(\beta\circ(\alpha\otimes \id)(x)))\) in the interpolating region.

    We first show that \(z^i(s)\) has Lipschitz constant \(2\) when using a distance internal to the pleats, \(\dist_{\mathrm{pleat}}\).
    That is, the distance between points \(s^i\) and \(s^{\prime j}\) from different layers is measured along the meanders of the pleating.
    Equivalently, we measure the distance in terms of \(y^i(s)\) with the first coordinate scaled up by \(\lambda\).
    The function \(s^i \mapsto y^i(s)\) then has Lipschitz constant \(1\) (it is contractive in the first component, and preserves distances in other components).
    The function \(x^i \to t(x^i)\) has Lipschitz constant \(1/5 L_1\) (it is linear interpolation)
    and we have, measured in the usual \(\ell_\infty\) distance on the lattice, that \(\dist(y^i(s),s^i) \leq 5 L_1\) from the block structure of the pleats.
    Dropping layer indices, and using that our lattice distance comes from a norm, we have
    \begin{equation}\begin{split}
        \|z(s) - z(s')\| &= \|t(s) (s - s') + (1-t(s))(y(s)-y(s')) + (t(s) - t(s')) (s'- y(s'))\| \\
        &\leq [t(s) + (1-t(s)) + \|s'-y(s')\|/5L_1] \dist_{\mathrm{pleat}}(s,s') \\
        &\leq 2 \dist_{\mathrm{pleat}}(s,s'),
    \end{split}\end{equation}
    where we added and subtracted \(t(s) s' + (1-t(s)) y(s')\). We also used the triangle inequality, the Lipschitz conditions for \(t\) and \(y\), the bound on \(\dist(y^i(s),s^i) = \|s^i - y^i(s)\|\), and \(\|s-s'\| \leq \dist_{\mathrm{pleat}}(s,s')\).

    Now consider some point \(v\) in the interpolating region, and define its preimage under \(z\),
    \begin{equation}
        z^{-1}(v) = \{ (s,i) \,\mid\, z^i(s) = v\}.
    \end{equation}
    For any \(x \in \fMat(v)\), we have that
    \begin{equation}
        \Supp(\gamma\circ(\alpha\otimes \id)(x)) \subseteq \bigcup_{i=1}^\lambda z^\bullet([z^{-1}(v)]^{+4 \ell}),
    \end{equation}
    where \(z^\bullet(s,i) = z^i(s)\) and the \(4\ell\) growth (from the action of \(\beta(\alpha\otimes \id)\)) can spread between pleated layers.
    However, each \(z^i\) has Lipschitz constant \(2\), so no point at most \(4 \ell\) from \(s \in z^{-1}(v)\) can have \(z^i(s)\) further than \(8 \ell\) from \(v\). Thus
    \begin{equation}
        \Supp((\gamma\circ(\alpha\otimes \id))(x)) \subseteq \Supp(x)^{+8\ell},
    \end{equation}
    as required.
\end{proof}

\begin{lemma}[Local spread reduction]\label{lem:SpreadReduction}
    Let \(\alpha\) be a QCA with inner spread \(L\) and outer spread \(\ell\) relative to the ball \(B\).
    Then there is a local assignment \(p'\) and finite depth circuit \(\beta\)
    supported on \(B^{+5\cdot 2^{\dd} L + 8^{\dd} \ell}\)
    such that \(\gamma = \beta\circ (\alpha \otimes \id)\) has spread \(8^\dd \ell\).

    Further, if \(\alpha\) acts as the identity in a smaller ball \(b \subseteq B\), then \(\beta\) can be chosen to act trivially and involve no new ancillas in \(b^{-5 \cdot 2^{\dd-1}L}\).
\end{lemma}

\begin{proof}
    We apply \ref{lem:PleatingLocalOneDimension} to each axis, similarly to \ref{lem:PleatingSpreadReduction}.
    
    The inner anisotropic spread of \(\alpha\) is initially \(\mathbf{L} = (L, L, ..., L)\) relative to the ball \(B\). If \(L \leq 8^\dd \ell\), the theorem is vacuous.
    Otherwise, after applying \ref{lem:PleatingLocalOneDimension} to the first axis, we get a new QCA whose inner spread is \(\mathbf{L} = ( \epsilon L , 2L, ..., 2L)\) (where \(\epsilon = 4/\lambda\)) and has an outer spread \(8 \ell\).
    Repeating this, we get a QCA \(\gamma\) whose inner spread in \(B\) is bounded by \(2^{\dd-1} \epsilon L \) and whose outer spread is bounded by \(8^\dd \ell\).
    Choose \(\epsilon\) such that \(2^{\dd-1} \epsilon L \leq 8^\dd \ell\), so that the new spread is just \(8^\dd \ell\).

    Further, each QCA obtained from iteratively applying \ref{lem:PleatingLocalOneDimension} differs from the previous by a circuit localized to the \(5 L_k' + 3\ell'\) growth of \(B\), where \(L_k'\) and \(\ell'\) are the \(k\)th-inner and outer spreads from the previous iteration.
    Solving this recursion gives that the final \(\gamma\) differs from \(\alpha \otimes \id\) only in the
    \begin{equation}
        5(1 + 2 + \cdots + 2^{\dd-1} ) L 
        + 3(1 + 8 + \cdots + 8^{\dd-1}) \ell \leq 5\cdot 2^{\dd} L + 8^{\dd} \ell
    \end{equation}
    growth of \(B\).

    The addendum regarding preservation of the identity in the ball \(b^{-5 \cdot 2^{\dd-1}L}\) is identical to \ref{lem:PleatingSpreadReduction}---the local automorphisms of \(\beta\) fully inside \(B\) are identical to those for the global circuit from \ref{lem:PleatingSpreadReduction}, so the same result applies.
\end{proof}

\begin{remark}\label{rem:SpreadReduceOnLeftOrRight}
    In \ref{lem:SpreadReduction}, one can alternatively take  \(\beta\) to multiply \(\alpha \otimes \id\) on the right with no relaxation of the bounds.
    Indeed, all the pleating and deformation of \(\alpha\) can be done by multiplying on either the left or the right.
\end{remark}

\subsection{Induction by hemming}
\label{subsec:Hemming}

The strategy to construct the claimed circuit representation of~\(\alpha\)
in~\ref{thm:QCAWithAHoleIsTrivial}
is essentially to retract \(\alpha\) out from larger and larger balls, towards infinity. 
We call this process \emph{hemming}, continuing a tailoring analogy from~\cite{FHH2019}.

\begin{lemma}[Hemming]\label{thm:Hemming}
    If the spread-\(\ell\) QCA \(\alpha: \fMat(\Lambda, p) \to \fMat(\Lambda, p)\)
    acts as the identity in the interior \(B^\circ\) of a ball \(B\) of radius \(r\),
    then for every integer \(R \geq 0\)
    there is a local assignment \(p'\) and
    a circuit \(\beta: \fMat(\Lambda', pp') \to \fMat(\Lambda', pp')\) supported in \(B^{+R}\setminus B^{\circ}\)
    of depth bounded by \(R/r,\ell,\dd\),
    such that \(\beta\circ(\alpha\otimes \id)\) acts as the identity in the interior of \(B^{+R}\)
    and has spread \(\ell(R+r)/r\).
\end{lemma}

\begin{proof}
    Define a map \(P: B^{+R}\setminus B^{\circ} \to S_{R} = \{t \in \RR^\dd \,|\, \dist(t,s_B) = R+r\}\) (where \(s_B\) is the center of \(B\))
    by the projection along the (Euclidean) ray
    joining the center of \(B\) to \(s \in B^{+R}\setminus B^{\circ}\),
    produced to \(S_{R}\).
    The projection is Lipschitz with Lipschitz constant \(K = (R+r)/r\):
    if two points \(s_1,s_2 \in B^{+R}\setminus B^{\circ}\) have \(\dist(s_1,s_2) \leq \ell\),
    then \(\dist(P(s_1), P(s_2)) \leq \ell (R+r)/r\).%
    \footnote{
    The inner surface of the projected annulus is at radius \(r\),
    and the outer is at radius \(R+r\).
    }
    
    Introduce ancillas on \(t \in S_{R}\),
    \begin{equation}
        p'(t) = \bigotimes_{s \in P^{-1}(t)} p'_s(t) = \bigotimes_{s \in P^{-1}(t)} p(s),
    \end{equation}
    and define
    \begin{equation}
        \beta = \left[\prod_{t \in S_{R}} \prod_{s \in P^{-1}(t)} \SWAP_{p'_s(t) p(s)}\right] (\alpha\otimes \id) \left[\prod_{t \in S_{R}} \prod_{s \in P^{-1}(t)} \SWAP_{p'_s(t) p(s)}\right] (\alpha^{-1} \otimes \id).
    \end{equation}
    \(\beta\) is a circuit with local automorphisms of diameter at most \(R + 2 \ell\): the maximum distance \(P\) sends a point is \(R\), which is the spread of the swap, and conjugation by \(\alpha\) grows the support by \(\ell\), and the diameter by \(2 \ell\).

    We see that \(\beta \circ (\alpha \otimes \id)\) is \(\alpha \otimes \id\) conjugated by swaps which move all sites in \(B^{+R}\setminus B^{\circ}\) to \(S_{R}\).
    All the nontrivial action of \(\alpha\) has been hemmed to (the closure of) \(\RR^\dd \setminus B^{+R}\), enlarging the hole where it acts as the identity to the interior of \(B^{+R}\). 
    This comes at the cost of increasing the spread of \(\beta \circ (\alpha \otimes \id)\) to \(\ell(R+r)/r\), which follows from the distance bounds on \(P\).
\end{proof}

\begin{remark}\label{rem:HemOnLeftOrRight}
    Analogously, we can find a \(\beta\) such that \((\alpha \otimes \id) \circ \beta\) satisfies the conclusions of \ref{thm:Hemming}. One can hem from the left or right.
\end{remark}

Repeatedly hemming \(\alpha\) is not sufficient to find a finite depth circuit representation, because hemming increases the spread. 
Specifically, operators \(x\) supported in the shell \(S_{R}\) now only have \(\Supp\beta\alpha(x) \subseteq \Supp(x)^{+L}\), where \(L = \ell(R+r)/r\).
Attempting to iterate naively would cause the spread to diverge.
To mend this issue, we need to intersperse the hemming the spread reduction by pleating and ironing from \ref{subsec:PleatingIroning}.

\begin{lemma}\label{lem:InductionStep}
    Let \(\alpha\) be a QCA that acts trivially in the interior \(B^\circ\) of a ball \(B\) of radius \(r_0\), has inner spread \(L\) inside \(B^{+r_1}\), and outer spread \(\ell\).
    For every \(R \geq r_1\), there is a local assignment \(p\) and finite depth circuit \(\beta\) supported in \(B^{+R_0 + R_1}\setminus B^\circ\) such that \(\beta \circ (\alpha \otimes \id)\) acts as the identity in \((B^{+R_0})^\circ\), and has inner spread \(8^\dd \ell\) and outer spread \(\ell\) relative to \(B^{+R_0 + R_1}\), where
    \begin{equation}
        R_0 = R - 5 \cdot 2^{\dd-1} KL, 
        \quad
        R_1 = 5(2^\dd+ 2^{\dd-1}) K L + 8^\dd \ell,
    \end{equation}
    and
    \begin{equation}
        K = 1 + \frac{R}{r_0}.
    \end{equation}
\end{lemma}

\begin{figure}
    \centering
    \begingroup
        \hspace*{-5ex}
        \resizebox{0.95\textwidth}{!}{
%
%

\definecolor{cSheet}{RGB}{219,233,250}    
\definecolor{cCollar}{RGB}{130,178,235}   
\definecolor{cCollarM}{RGB}{46,92,168}    
\definecolor{cLine}{RGB}{58,110,200}      

\newlength{\indfigunit}
\setlength{\indfigunit}{21.2682pt}

\begin{tikzpicture}[
    line cap=round,
    line join=round,
    x=\indfigunit, y=\indfigunit,
    sheet/.style   = {fill=cSheet, draw=cLine, line width=1.1pt,
                      rounded corners=7pt},
    wsq/.style   = {fill=white,  draw=cLine, line width=1.1pt},
    ghost/.style   = {draw=black,  line width=0.9pt,
                      dash pattern=on 4pt off 3pt},
    stepar/.style    = {-{Latex[length=3mm,width=2.4mm]}, line width=1.0pt, black},
    dim/.style     = {{Latex[length=1.8mm,width=1.6mm]}-{Latex[length=1.8mm,width=1.6mm]},
                      line width=0.8pt, black},
]

\def\Wh{3.0}                  
\def\Gap{2.0}                 
\def\rZero{0.95}              
\def\Tcol{0.70}               
\def\tMid{0.20}               
\def\wMid{2.00}               
\def\wEnd{1.80}               

\pgfmathsetmacro{\cA}{\Wh}
\pgfmathsetmacro{\cB}{\cA+2*\Wh+\Gap}
\pgfmathsetmacro{\cC}{\cB+2*\Wh+\Gap}

\pgfmathsetmacro{\aA}{\rZero+\Tcol}     
\pgfmathsetmacro{\aB}{\wMid+\tMid}      
\pgfmathsetmacro{\aC}{\wEnd+\Tcol}      

\newcommand{\panel}[4]{%
    \draw[sheet] ({#1-\Wh},-\Wh) rectangle ({#1+\Wh},\Wh);
    \fill[#4]    ({#1-#2},{-#2}) rectangle ({#1+#2},{#2});
    \draw[wsq] ({#1-#3},{-#3}) rectangle ({#1+#3},{#3});
}

\panel{\cA}{\aA}{\rZero}{cCollar}
\draw[dim] (\cA,0) -- (\cA,\rZero);
\node[anchor=west, xshift=2pt] at (\cA,{\rZero/2}) {$r_0$};
\draw[dim] (\cA,\rZero) -- (\cA,\aA);
\node[anchor=west, xshift=2pt] at (\cA,{(\rZero+\aA)/2}) {$r_1$};
\node[anchor=south west, xshift=0pt, yshift=0pt]
     at ({\cA-\rZero},{-\rZero}) {$B$};

\panel{\cB}{\aB}{\wMid}{cCollarM}
\draw[ghost] ({\cB-\rZero},{-\rZero}) rectangle ({\cB+\rZero},{\rZero});
\draw[dim] (\cB,\rZero) -- (\cB,\wMid);
\node[anchor=west, xshift=2pt] at (\cB,{(\rZero+\wMid)/2}) {$R$};

\panel{\cC}{\aC}{\wEnd}{cCollar}
\draw[ghost] ({\cC-\rZero},{-\rZero}) rectangle ({\cC+\rZero},{\rZero});
\draw[dim] (\cC,\rZero) -- (\cC,\wEnd);
\node[anchor=west, xshift=2pt] at (\cC,{(\rZero+\wEnd)/2}) {$R_0$};
\draw[dim] (\cC,\wEnd) -- (\cC,\aC);
\node[anchor=west, xshift=2pt] at (\cC,{(\wEnd+\aC)/2}) {$R_1$};

\pgfmathsetmacro{\arA}{\cA+\Wh+0.30}\pgfmathsetmacro{\arB}{\cB-\Wh-0.30}
\draw[stepar] (\arA,0) -- (\arB,0);
\node[anchor=south, yshift=2pt] at ({(\arA+\arB)/2},0) {hem};

\pgfmathsetmacro{\arC}{\cB+\Wh+0.30}\pgfmathsetmacro{\arD}{\cC-\Wh-0.30}
\draw[stepar] (\arC,0) -- (\arD,0);
\node[anchor=south, yshift=2pt]  at ({(\arC+\arD)/2},0) {pleat};
\node[anchor=north, yshift=-2pt] at ({(\arC+\arD)/2},0) {iron};

\end{tikzpicture}}
    \endgroup
    \caption{A QCA with a hole---a ball \(B\) of radius \(r_0\) in which it acts as the identity---and a collar of width \(r_1\) in which it has a larger spread, is iteratively expanded by hemming (\ref{thm:Hemming}), pleating, and ironing (\ref{lem:SpreadReduction}).
    Hemming expands the size of the hole from \(r_0\) to \(r_0 + R\), at the cost of increasing the spread in the collar.
    Pleating and ironing then reduces the spread, while reducing the size of the expanded hole to \(r_0 + R_0\) and increasing the size of the collar to \(R_1\) (\ref{lem:InductionStep}).
    }
    \label{fig:HemPleatIronInduction}
\end{figure}
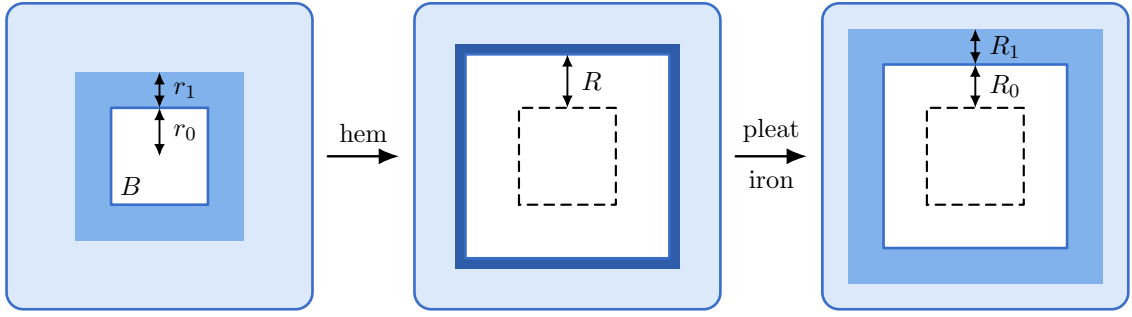

\begin{proof}
    First hem \(\alpha\) so that the radius in which it acts as the identity expands to \(r_0 + R\), where \(R \geq r_1\).
    Then, the Lipschitz constant appearing in \ref{thm:Hemming} is 
    \begin{equation}
        \frac{r_0+R}{r_0} = K,
    \end{equation}
    so that the hemmed QCA has inner spread \(KL\) in \(B^{+R}\) and outer spread \(\ell\).
    The hemmed QCA acts as the identity in \((B^{+R})^\circ\) (\autoref{fig:HemPleatIronInduction}).

    Then apply the pleating-and-ironing spread reduction from \ref{lem:SpreadReduction}, and let the circuit relating \(\alpha\) to the result be \(\beta\).
    The QCA \(\beta(\alpha \otimes \id)\) has an inner spread \(8^\dd \ell\) in a ball \(B^{+R_0+R_1}\), where \(R_0+R_1 = R + 5 \cdot 2^\dd K L + 8^\dd \ell\), and outer spread \(\ell\).

    Further, \(\beta(\alpha \otimes \id)\) can be chosen to continue to act as the identity in the interior of the ball \(B^{+R_0}\), where \(R_0 = R - 5 \cdot 2^{\dd-1} KL\) by \ref{lem:SpreadReduction}.
    This gives \(R_1 = 5( 2^\dd+ 2^{\dd-1}) K L + 8^\dd \ell\) (which still has \(R\) dependence through \(K\)).
\end{proof}

We now proceed to the proof of \ref{thm:QCAWithAHoleIsTrivial}.

\begin{proof}[Proof of \ref{thm:QCAWithAHoleIsTrivial}]
    The proof is by induction with \ref{lem:InductionStep}.
    The QCA \(\alpha =: \alpha_0\) has a hole of radius \(r =: r_0\), and \ref{lem:InductionStep} gives us a circuit \(\beta_0\) such that (dropping the explicit \(\otimes \id\) and changing lattice \(\Lambda'\)) \(\beta_0 \alpha_0 =: \alpha_1\) has a larger hole of radius \(r_1\).
    Iterate this enlargement of the hole, alternating whether the hole is enlarged by multiplication on the left \(\beta_{2 i} \alpha_{2i} = \alpha_{2i+1}\) or on the right \(\alpha_{2i+1} \beta_{2i+1} = \alpha_{2(i+1)}\).
    Then, if each \(\beta_j\) does not overlap in support with \(\beta_{j\pm 2}\), we can express \(\alpha\) as a circuit
    \begin{equation}\label{eqn:TailoringCircuit}
        \alpha = \beta_0^{-1} \alpha_1 = \beta_0^{-1} \alpha_2 \beta_1^{-1} = \cdots =
        \left(\prod_{i \in \NN} \beta_{2i}^{-1}\right) \left(\prod_{i \in \NN} \beta_{2i+1}^{-1}\right),
    \end{equation}
    where we used that the action of \(\alpha_j\) on any local \(x\) is the identity for large enough \(j\).
    Each factor in the left or right product is a non-overlapping finite depth circuit, where the circuit depth is bounded by the spread \(\ell_j\) of \(\alpha_j\) and the hemming step size \(R\) from \ref{lem:InductionStep} (and the dimension \(\dd\)).
    Thus, it remains to show that each \(\beta_j\) can be made sufficiently non-overlapping while maintaining a constant step size \(R\) and a constant bound on the spread.

    We set the following (suboptimal) targets for \(\alpha_j\) and the parameters of \ref{lem:InductionStep}:
    \begin{itemize}
        \item Inner spread \(L_j \leq 8^\dd \ell\). This is automatic from the pleating and ironing construction given the outer spread remains \(\ell\).
        \item Constant step size \(R\), independent of \(j\).
        \item A Lipschitz constant \(K \leq 2\). This places a constraint on \(R\) as
        \begin{equation}
            1+\frac{R}{r_j} \leq 2
            \implies
            R \leq r_j,
        \end{equation}
        where \(r_j\) is the radius within which \(\alpha_j\) acts as the identity. 
        The strongest constraint is set by \(R \leq r_0\).
        \item Finally, we need the distance \(R_0 = r_{j+1} -r_j\) to be large enough so that none of the circuits \(\beta_j\) and \(\beta_{j+2}\) overlap. 
        Demanding \(R_0 > R_1\) places the radius in which \(\alpha_j\) acts as the identity past the limits of the circuit from the previous step, which is set at the radius for the outer spread \(R_1\).
        This gives
        \begin{align}\label{eqn:RprimeLowerBound}
            R_0 = R - 5 \cdot 2^{\dd-1} K L &> 5(2^\dd+2^{\dd-1}) K L + 8^\dd \ell \nonumber
            \\
            R &> 5\cdot 2^{\dd+1} K L + 8^\dd \ell
        \end{align}
    \end{itemize}

    Using \(L \leq 8^\dd \ell\) and \(K \leq 2\), 
    to meet the condition~\eqref{eqn:RprimeLowerBound} it is sufficient to have
    \begin{equation}
        R > (1+5\cdot 2^{\dd+2})8^\dd \ell
    \end{equation}
    For this to be compatible with the constraint \(R \leq r_0 -1\), we must have
    \begin{equation}
        r_0 > (1+5\cdot 2^{\dd+2})8^\dd \ell.
    \end{equation}
    If this condition is met, then choosing \(R = r_0\) in every application of \ref{lem:InductionStep} maintains all the conditions above, so that Eq.~\eqref{eqn:TailoringCircuit} is a finite depth circuit.
    Thus, \(r_0 \geq 2^{4\dd + 5} \ell\) suffices to complete the induction,
    even after rounding lattice points, which increases the spread by \(1\) at most.
\end{proof}

Note that our proof of \ref{thm:QCAWithAHoleIsTrivial} uses the nonuniform stabilization essentially.
In particular, hemming uses as many new ancillas as there are local factor algebras previously in the annulus being hemmed.
Thus, at each shell of radius \(n R_0\) (\(n\) an integer)
we introduce enough ancillas to support the entire QCA inside that shell
and even more to do the spread reduction.

\begin{remark}
The nonuniform stabilization is likely essential for \ref{thm:QCAWithAHoleIsTrivial} to hold.
With only uniform ancillas, it seems plausible that the density of shifts across an infinite cut of the lattice will be a circuit-class invariant of \(2\dd\) QCAs.
Any real density is achievable by some product of shifts, so there should be a continuum of circuit classes in \(2\dd\) with only bounded ancillas.
The introduction of nonuniform ancillas trivializes all such products of lower-dimensional QCAs.
\end{remark}

\nocite{apsrev42Control}
\bibliographystyle{apsrev4-2}
\bibliography{refs.bib}
\end{document}